\documentclass{article}

\PassOptionsToPackage{numbers, compress}{natbib}

\usepackage[preprint]{neurips_2023}
\usepackage{natbib}
\usepackage{amsthm}
\usepackage{graphicx}
\usepackage{mathtools} 

\usepackage{amsmath,amsfonts,amssymb}

\usepackage{algorithm}
\usepackage{algorithmic,enumerate}
\usepackage{dsfont}
\usepackage{xcolor} 
\usepackage{caption,color}
\usepackage{subcaption}
\usepackage{hyperref}
\usepackage{ulem}
\usepackage[T1]{fontenc}
\usepackage[utf8]{inputenc}
\usepackage{comment}
\usepackage{units}
\newcommand{\am}[1]{{[\color{green}#1}]}

\newcommand{\prk}[1]{{[\color{red}PRK: #1}]}

\DeclareMathOperator{\Tr}{Tr}

\newcommand{\bE}{\mathbb{E}}
\newcommand{\bZ}{\mathbb{Z}}

\newcommand{\mbf}{\mathbf}
\newcommand{\cS}{\mathcal{S}}

\newcommand{\bN}{\mathbb{N}}

\newcommand{\ust}{^{\star}}
\newcommand{\ut}{^{(t)}}
\newcommand{\bR}{\mathbb{R}}
\newcommand{\te}{\theta}

\newcommand{\cI}{\mathcal{I}}

\newcommand{\bP}{\mathbb{P}}
\newcommand{\cE}{\mathcal{E}}
\newcommand{\cG}{\mathcal{G}}
\newcommand{\ui}{^{(\mathcal{I})}}

\newcommand{\cj}{\mathcal{J}}
\newcommand{\cT}{\mathcal{T}}

\newcommand{\pj}{\mbox{proj}}
\newcommand{\eps}{\epsilon}
\newcommand{\lm}{\lambda}

\newcommand{\cR}{\mathcal{R}}

\newcommand{\cF}{\mathcal{F}}

\newcommand{\Te}{\Theta}
\newcommand{\nal}[1]{\begin{align*}#1\end{align*}}
\newcommand{\al}[1]{\begin{align}#1\end{align}}
\newtheorem{assumption}{\textbf{Assumption}}

\newtheorem{definition}{Definition}[section]
\newtheorem{theorem}{Theorem}[section]
\newtheorem{proposition}{Proposition}[section]

\newtheorem{lemma}[theorem]{Lemma}
\newtheorem{remark}{Remark}
\newcommand{\norm}[1]{\left\lVert#1\right\rVert}

\DeclarePairedDelimiter\autobracket{(}{)}
\newcommand{\br}[1]{\autobracket*{#1}}

\title{Finite Time Regret Bounds for Minimum Variance Control of Autoregressive Systems with Exogenous Inputs}

\author{%
   Rahul Singh \\
Indian Institute of Science \\
   Bengaluru, Karnataka, India \\
   \texttt{rahulsingh@iisc.ac.in} \\
   \And
  Akshay Mete \\
  Texas A \& M University\\
  College Station, Texas, USA\\
  \texttt{akshaymete@tamu.edu} \\
  \AND
   Avik Kar \\
Indian Institute of Science \\
   Bengaluru, Karnataka, India \\
   \texttt{avikkar@iisc.ac.in} \\
   \And
  P. R. Kumar\\
  Texas A \& M University\\
  College Station, Texas, USA\\
  \texttt{prk@tamu.edu} \\
}

\begin{document}
\maketitle
	
\begin{abstract}
Minimum variance controllers have been employed in a wide-range of industrial applications. 
A key challenge experienced by many adaptive controllers is their poor empirical performance in the initial stages of learning \cite{lale22, mete2022augmented}. In 
this paper, we address the problem of initializing them so that they provide acceptable transients, and also provide an accompanying finite-time regret analysis, for adaptive minimum variance control of an auto-regressive system with exogenous inputs (ARX). Following \cite{lai1987asymptotically}, we consider a modified version of the Certainty Equivalence (CE) adaptive controller, which we call PIECE, that utilizes probing inputs for exploration. We show that it has a $C \log T$ bound on the regret after $T$ time-steps for bounded noise, and $C\log^2 T$ in the case of sub-Gaussian noise. 
The simulation results demonstrate the advantage of PIECE over the algorithm proposed in \cite{lai1987asymptotically} as well as the standard Certainty Equivalence controller especially in the initial learning phase.
To the best of our knowledge, this is the first work that provides finite-time regret bounds for an adaptive minimum variance controller. 

\end{abstract}

\section{Introduction}
Adaptive control theory focuses on developing controllers for systems with unknown models.
Due to its extensive industrial applications, adaptive control of linear systems is 
one of the most exhaustively studied problems in control theory.  Traditionally, the analysis of adaptive controllers has focused on the holy grail of three asymptotic properties: stability, i.e., is the overall nonlinear system formed by the linear system in feedback with the nonlinear adaptive controller is stable or bounded in some appropriate sense; self-optimality, i.e., does the adaptive controller result in a long-term average performance that is optimal; and self-tuning, i.e., do the parameters of the adaptive control law converge to those of an optimal control law for the system. With the advent of reinforcement learning, there has been a renewed focus on the topic of adaptive control of linear systems in recent years \cite{simchowitz2020naive,dean2019safely}. 
Although these recent works consider the same problem, they focus on the finite-time analysis and sample complexity of such adaptive controllers instead of only asymptotic properties. The recent works on adaptive linear systems have majorly focused on the popular LQG problem \cite{abbasi2011regret,abeille_17,faradonbeh_rce,faradonbeh2018input,dean2018regret,mete2022augmented,lale22,simchowitz2020naive,jedra2022minimal}. In this paper, we revisit another popular controller called the minimum variance (MV) controller.

Consider an auto-regressive linear system with exogenous inputs (ARX system),
\al{
y_t= a_1 y_{t-1} + a_2 y_{t-2} + \ldots +a_p y_{t-p} +b_1 u_{t-1} + b_2 u_{t-2}+\cdots + b_q u_{t-q}+ w_t, ~ \forall~ t \in \bN \label{def:arx}
}
where $y_t, u_t$ are the output and input at time $t$ respectively, and $w_t$ is the system noise at time $t$ that is i.i.d.~of mean 0 and variance $\sigma^2$. The MV controller has as its goal the minimization of the variance of the output, and is given by~\cite{aastrom2012introduction},
$$
u_t = -(1/b_1)\br{a_1 y_{t} + a_2 y_{t-1} + \ldots +a_p y_{t-p+1} + b_2 u_{t-1}+\cdots + b_q u_{t-q+1}}.\label{MV control law}
$$ 
The development of the MV controller has its origins in the paper-making industry \cite{astrom1967computer} where a minimum paper thickness has to be guaranteed to customers, and reduction of the variance of the thickness allows one to set a much lower ``mean plus three standard deviations" set-point and thereby save considerably on the paper pulp. 


If the system parameter vector
$\te\ust:=\br{a_1,a_2,\ldots,a_p,b_1,b_2,\ldots,b_q}'$
is not known, 
the ``self-tuning regulator'' approach was proposed in \cite{aastrom1973self}.
It has had a wide-range of industrial applications \cite{astrom1987adaptive,dumont1982self,seborg1986adaptive,sastry1977self,bongtsson1984experiences,kallstrom1979adaptive}. 
Based on the Certainty Equivalence (CE) approach, it replaces the unknown parameter $\te\ust$ in (\ref{MV control law}) by ${\theta}$ obtained as a least-squares estimate (LSE) or from a stochastic gradient estimator.
The adaptive control law with the stochastic gradient estimator was shown to be stable and self-optimizing in \cite{goodwin1981discrete}, and self-tuning in \cite{becker1985adaptive,kumar1987self}.
For LSEs, the stability and self-optimality, 
and self-tuning
under a modification to the adaptive control law featuring an additional diminishing excitation, were established in \cite{guo1991astrom}. The robustness of a minimum variance controller to modeling assumptions, e.g., perturbations in the space of transfer functions, which lead to infinite dimensional systems, was shown in \cite{praly1989robust}.

The finer issue of asymptotic regret, $\cR_T :=\sum_{t=1}^T(y_t-w_t)^2$, that has since become mainstream in reinforcement learning (RL), was examined in 
\cite{lai1986asymptotically,lai1987asymptotically} by introducing ``exploration episodes" where the algorithm, which we refer to as LW, utilizes probing inputs (in contrast to the diminishing excitation approach of \cite{guo1991astrom}).
Asymptotically logarithmic regret was established by showing that
\nal{\limsup_{T \to \infty} \frac{R_T}{\log T} = \sigma^2(p+q-1). \label{asymptotically-optimal-regret}}
It was also shown to be asymptotically optimal with the RHS a lower bound on achievable regret.

\subsection{Our contributions}
Recently there has been great interest in finite time analysis of adaptive controllers in the reinforcement learning literature \cite{abbasi2011regret,abeille_17,faradonbeh_rce,faradonbeh2018input,dean2018regret,mete2022augmented,lale22,simchowitz2020naive,jedra2022minimal}.
Our central concern in this paper is on improving the transient performance of the system and establishing finite-time bounds. 
Surprisingly, the finite time analysis of learning algorithms for a MV controller has been an open problem. 
In this paper, we take the first step in that direction and design an adaptive MV controller with a finite time regret bound, which also appears to give good empirical performance in the initial stage of learning.

A key challenge experienced by many adaptive controllers is their poor empirical performance especially at the initial stages of learning \cite{lale22, mete2022augmented}.
For improving the transient performance, it is important to properly adapt the system in the initial phase. Otherwise, the states of the system can reach arbitrarily high values before settling down to what is predicted by the asymptotic theory. 
This is even more exacerbated for the MV-CE control law which can be written as
$u_t = \lm_t' \psi_t$, 
where
\al{
\lm_t := (-1\slash b_{1,t}) \left(a_{1,t},\ldots,a_{p,t},{b}_{2,t},\ldots,{b}_{q,t} \right)',  \psi_t := \br{y_t,\ldots,y_{t-p+1},u_{t-1},\ldots, u_{t-q+1} }'
}
Since ${\lambda}_t$ involves a division by ${b}_{1,t}$ 
it is susceptible to large errors even for modest values of estimation error of ${b}_t$, which in-turn leads to high regret especially during the initial time steps. To overcome this shortcoming, we propose a modification that clips the inputs suggested by the CE rule to a compact set $[-B_u,B_u]$. The value of threshold $B_u$ is chosen based upon a fine-grained analysis of learning regret, and utilizes knowledge of $\Theta$, a compact set in which the true parameter $\te\ust$ is known to reside. To establish finite-time regret bounds, as well as improve initial performance, we 
also add exploratory episodes where noise is injected to gather information about $\theta^\star$.
The resulting algorithm, which we call PIECE, is a simplified and optimized version of~\cite{lai1987asymptotically}.

It is computationally efficient since it maintains an estimate of the optimal gain and LSEs of the unknown ARX parameters, both of which can be updated recursively and hence require $O(1)$ computation at each time step. 

The empirical results demonstrate that PIECE does not suffer a large regret at the beginning of the experiments, unlike LW and the CE controller. This highlights the benefit of the improved exploration strategy proposed here compared to LW.
Our empirical results also show that the resulting algorithm has much lower empirical regret compared to LW or the standard CE controller.
We also establish (see Theorems \ref{thm:regret_bounded} and \ref{th:estimation_error}) that when the process noise $\{w_t\}$ is bounded,
PIECE enjoys a finite-time regret less than $C \sigma^2 (p+q-1) \log T $ after $T$ steps, where $C \approx 1$.
This closely matches the asymptotically optimal regret of  (\ref{asymptotically-optimal-regret}).
If it is conditionally sub-Gaussian the bound is $C\sigma^2 (p+q-1)$ with $C \approx 1$.

Central to the finite-time performance results of the proposed algorithm is the fine analysis of the growth-rate of the minimum eigenvalue of the covariance matrix associated with LSE that is fed samples collected by the learning algorithm. We show that the minimum eigenvalue grows as $A\sqrt{N\ui_t}$, where $N\ui_t$ is the number of exploratory steps until $t$. The pre-factor $A$ is a function of system parameters, and (i) increases with the value of variance of the process noise, (ii) decreases with $a_1$, (iii) decreases as the stability margin $1-\rho$ (defined in the sequel) approaches 0. 
We also quantify the time (which unlike \cite{lai1987asymptotically} does not depend upon the sample path) after which the bound is guaranteed to hold. All of these improvements are visible when one compares the empirical performance of LW with PIECE. PIECE is seen to outperform LW and CE by a huge margin, as shown in the simulations Section \ref{sec:simulations}.

\subsection{Related Work}
Finite-time analysis of adaptive controllers for linear system was initiated by \cite{abbasi2011regret}.  They proposed an algorithm based on the OFU principle and showed that it has a regret bound of $C\sqrt{T}$ for LQG systems.
Since then, a variety of learning algorithms have been proposed for adaptive control of LQG systems. These algorithms can be categorized into three categories: (i) Algorithms based on the optimism under uncertainty (OFU) principle which include OFULQ \cite{abbasi2011regret}, StabL \cite{lale22},  OSLO \cite{cohen2019learning} and ARBMLE \cite{mete2022augmented}. (ii) Modified versions of CE algorithms (Self-tuning regulators)  which include CEC($\mathcal{T}$) \cite{jedra2022minimal}, CECCE \cite{simchowitz2020naive,mania2019certainty}, RCE \cite{faradonbeh_rce}, and IP \cite{faradonbeh2018input}. All of these algorithms are shown to achieve $C\sqrt{T}$ regret. (iii)  Another category is Thompson Sampling based algorithms which include  TS \cite{abeille_17}.

A lower bound of $C\sqrt{T}$ was shown for the LQG setting in \cite{simchowitz2020naive} for any algorithm. Notably, CEC$(\mathcal{T})$ is shown to have regret bound that matches the lower bound given in \cite{simchowitz2020naive} (under certain conditions on dimension of the system). In the scenario where either the $A$ or $B$ matrix is known, \cite{jedra2022minimal,cassel2020logarithmic} showed that it is possible to achieve logarithmic regret. OFU \cite{abbasi2011regret} lacked a efficient implementation. First efficient OFU based algorithm was given by \cite{cohen2019learning}. An exhaustive discussion of these recent works on the LQG setting and their comparison can be found in \cite{jedra2022minimal}.

\subsection{Organization of the Paper}
We describe the ARX system model and notation in Section \ref{sec:system_model}. In Section \ref{sec:algo}, we describe the adaptive MV control algorithm called PIECE. The main technical results on regret analysis are provided in Section \ref{sec:regret_main}. Section \ref{sec:simulations} includes the results of the simulation experiments for various ARX systems. We conclude with a brief discussion on open problems in Section \ref{sec:discussion}. All proofs and details on experiments are provided in the appendices.
\subsection{Notation}
Let $\det(M)$,$\Tr(M)$ and $\|M\|$  denote the determinant, trace, and operator norm induced by the Euclidean norm, respectively, of a matrix $M$. 
For a vector $x$, let $x'$ be its transpose,
and $\|x\|$ its Euclidean norm.
For two numbers $x\wedge y$ denotes their maximum and $x\vee y$ their minimum.~We use the abbreviations ``w.h.p.'' to denote ``with high probability.'' Throughout, to keep notation simple, we use $\lesssim$ and $\gtrsim$ in order to hide problem dependent constants.~All the vectors will be column vectors.~For a vector $x$ and a vector space $S$, we let $\pj\br{x,S}$ be the projection of $x$ onto $S$. $\bN$ denotes the set of natural numbers, $\bZ$ the set of integers and $\bZ_+$ the set of positive integers.

 \section{System Model}\label{sec:system_model}

The ARX model (\ref{def:arx}) can be written as:
\al{
y_t = \phi'_{t-1} \te\ust + w_t}
where,
$\te\ust:=\br{a_1,a_2,\ldots,a_p,b_1,b_2,\ldots,b_q}' \text{ and }\phi_{s-1}:= \br{y_{s-1},y_{s-2},\ldots,y_{s-p}, u_{s-1},\ldots,u_{t-q} }.$

We make the following assumption regarding the unknown linear system \eqref{def:arx}. 
\begin{assumption}\label{assum:unit_circle}
	The parameter $\te\ust$ associated with the ARX model~\eqref{def:arx} satisfies the following condition: the polynomials $s^p - a_1 s^{p-1} - a_2 s^{p-2}-\ldots-\alpha_p$ and $b_1 s^{q-1}+b_2 s^{q-2} +\ldots + b_q $ have all zeros inside the open unit disk. Moreover, $b_1 \neq 0$.
\end{assumption}
The latter minimum phase assumption is necessary for internal stability of the MV control law \cite{kumar1986stochastic}.
The former assumption can be replaced by the assumption that a stabilizing linear control law is known, and is also used for regret analysis in \cite{lai1986asymptotically,lai1987asymptotically}.

Define the vectors $Y_t := \br{y_t,y_{t-1},\ldots,y_{t-p+1}}'$ and $U_t := \br{u_t,u_{t-1},\ldots,u_{t-q+2}}'$, and the matrices
$
	A:= 
	\begin{pmatrix}
		a_1 & \cdots & a_{p-1} & a_p \\
		I_{p-1} &        &     & 0 
	\end{pmatrix}\label{def:A}
	, \mbox{ and }
$
$
	B :=  
	\begin{pmatrix}
		-b_2 \slash b_1 & \cdots & -b_q \slash b_1 \\
		I_{q-2} &   &  0
	\end{pmatrix},
$
where $I_{p-1},I_{q-2}$ are identity matrices of sizes $p-1$ and $q-2$ respectively. We have $\|A^{n}\|\le C_1 \rho^n,\|B^{n}\|\le C_1 \rho^n$, where $\rho<1$ can be taken to be any number greater than the spectral radii of $A$ and $B$. When we want to depict the dependence of these quantities upon the coefficients in a parameter vector $\te$, we use $\rho(\te),A(\te),B(\te),C_1(\te)$. For objects pertaining to $\te\ust$ (the true parameter), we suppress this dependence and simply write $\rho,A,B,C_1$, etc. 


\subsection*{Regret of an adaptive Minimum Variance Controller }
We will design algorithms which generate $\{u_t\}$ for the case when the system parameter $\te\ust$ is unknown. An adaptive control algorithm or a learning algorithm $\mathcal{A}$ 
is a sequence of measurable functions that at each time $t$ maps the observation history to control. Noting that the MV controller results in $y_t \equiv w_t$, we judge the performance of an algorithm $\mathcal{A}$ by its cumulative learning regret, 
\al{
	\cR_T(\mathcal{A}) := \sum_{t=1}^{T} (y_t-w_t)^{2}.\label{def:cr}
}


\begin{assumption}\label{assum:noise_var}
The noise $\{w_t\}$ is assumed to be a martingale difference sequence with respect to  filtration $\{\cF_t\}$, with conditional variance bounded away from 0, i.e.,
	\al{
		\inf_t \bE\br{w^2_t|\cF_{t-1}} >c_1>0,~\mbox{ a.s.}.~\label{cond:3}
	}
Also,
\al{
\bE\br{w^2_t|\cF_{t-1}} \le \sigma^2,~\mbox{ a.s. },~\forall t.
}
\end{assumption}


 Within this setup we consider two possibilities, either bounded or sub-Gaussian noise: 
\begin{assumption}[Bounded noise]\label{assum:bounded_noise}
	$\{w_t\}$ is uniformly bounded a.s., i.e.,
	\al{
		|w_t| \le B_w, a.s. ~\forall t.
	}
\end{assumption} 
We present the main result on regret and a proof sketch under Assumption~\ref{assum:bounded_noise} in Section~\ref{sec:bounded_noise}. This is then relaxed in Section~\ref{sec:unbounded} to allow for unbounded noise:
\begin{assumption}[Conditionally sub-Gaussian noise]\label{assum:sub_gaussian} 
For all $\gamma \in\bR$, and $\sigma$,
	\al{
		\sup_{t} \bE\left\{\exp \br{\gamma |w_t |} \Big| \cF_{t-1} \right\} \le \exp\br{\gamma^2\sigma^2 \slash 2}, a.s.   ~\forall t. \label{def:sub_gauss}
	}
\end{assumption}

We also assume the following prior information about the unknown system, as in \cite{lai1987asymptotically}:
\begin{assumption}\label{assum:Theta}
	The learning algorithm has knowledge of a compact set $\Theta \subset \bR^{p+q}$ that contains the true parameter value $\te\ust$.
\end{assumption}

\section{PIECE: An adaptive minimum variance control Algorithm} \label{sec:algo}
The PIECE algorithm is presented in~Algorithm~\ref{algo}.~It divides the total operation time into two parts: (i) Exploration: This consists of a sequence of intervals during which white noise, by which is meant i.i.d., mean 0 and constant variance noise, is used as the control input to ensure sufficient excitation of the system, which in-turn yields consistent estimates, and (ii) Exploitation: The rest of the time, where a standard CE controller is applied by generating controls that are optimal under the assumption that the least squares estimates are equal to the true parameter values. Though this structure is inspired by the algorithm of~\cite{lai1987asymptotically} (henceforth dubbed LW), we will highlight in the sequel some major differences which allow for finite-time regret analysis and much better transient performance. 

\textbf{\emph{Exploration}}: The set of exploration time instants is denoted by $\cI$. For $t\in \cI$, $u_t$ is an i.i.d.~mean 0 sequence that is independent of $\{w_t\}$,  and bounded:
\al{
	|u_t|\le B_w,~a.s.~t\in\cI.\label{def:bounded_input}
}
The reason why we clip inputs at $B_w$ is that during this phase the algorithm is essentially open-loop. Consequently, it behaves ``conservatively'' and avoids using inputs of large magnitudes.
Let $N\ui(t)$ to denote the number of exploratory steps until $t$. $\cI$ is composed of multiple episodes, each comprising of a set of consecutive time steps. For $i=1,2,\ldots,$ the $i$-th such exploratory episode begins at time $n_i$, and ends at time $n_i + m_i$. The first episode begins at time $t=1$, i.e., $n_1 = 1$, and lasts until the following stopping-time,
\al{
n_1 + m_1 := \max \left\{ \tau , \inf \{t: b_{1,t} \neq 0 \}, H_1(\Theta,\eps)   \right\},
}
where $\tau := \inf \left\{ t :  \sum_{s=1}^{t}\phi_s \phi'_s \mbox{ is invertible} \right\}$, $b_{1,t}$ denotes the estimate of $b_1$ generated at time $t$, and $H_1(\Theta,\eps)$ depends upon the model parameter set $\Theta$ and $\eps$ is a parameter choice that decides the length of the first exploratory phase in the algorithm, and is detailed in the Appendix. Its affect on the regret is shown in Theorem~\ref{thm:regret_bounded}.
The first exploratory phase serves as a special ``warm-up'' phase, and is of longer duration than the remaining ones. It arises naturally out of the regret analysis, with sufficient exploration in the first few time-steps allowing us to bound the regret as $C\log T$. For the remaining episodes, $i=2,3,\ldots$,
\al{
n_i = \exp(i^2),~m_i = H, \label{def:n_i} \text{ where } H  = \Big\lceil  m\ust + \log_{\rho}\br{\frac{1}{3C_1 q}}  \Big\rceil, 
}
and $m\ust =  \Bigg\lceil\frac{1}{\log \rho}\log \br{\frac{B_w}{\br{\sup_{\te\in\Theta}C_1(\te) \br{\| Y_{0}\| C_1(\te)  +B_u C_1(\te)\left\{1+  \sum_{\ell=1}^{q} |b_\ell(\te)|   \right\}} }}} \Bigg\rceil$, where $B_u$ is defined below.

\textbf{\emph{Estimates}}:~Let $\te^{(\cI)}_t$ be the least-squares estimate (LSE) of $\te\ust$ based upon only the samples in $\cI$, 
\al{
\te^{(\cI)}_t :=  \br{\sum_{s\le t,s\in\cI} \phi_s \phi'_s }^{-1}\left(\sum_{s\le t,s\in\cI} \phi_s y_{s+1}\right).\label{def:lse_I}
}
Since the estimate of $b_1$ in $\te^{(\cI)}_t$ might be $0$, we modify it slightly as follows so that the resulting estimate can be used for estimating $\lm$,
\al{
	\tilde{\te}^{(\cI)}_t :=
	\begin{cases}
		\te^{(\cI)}_t \mbox{ if } b^{\cI}_{1,t} \neq 0,\\
		\tilde{\te}^{(\cI)}_{t-1} \mbox{ otherwise }
	\end{cases}\label{eq:theta_estimate}
	.
}
Let $\lm := -\frac{1}{b_1}\br{a_1, a_2,\ldots,a_p,b_2,\ldots,b_q}'$.~Since $\tilde{b}^{(\cI)}_{1,t} \neq 0$, we use it to estimate $\lm$ as follows,
\al{
	\tilde{\lm}^{(\cI)}_t := -\frac{1}{\tilde{b}^{(\cI)}_{1,t} } \br{\tilde{a}^{(\cI)}_{1,t},\tilde{a}^{(\cI)}_{2,t},\ldots,\tilde{b}^{(\cI)}_{2,t},\ldots,\tilde{b}^{(\cI)}_{q,t}}.
}
Even though we later show $\tilde{\te}^{(\cI)}_t ,\tilde{\lm}^{(\cI)}_t$ to be consistent, they need not be efficient since they use only a small fraction of the total available samples. Hence, while generating $\{u_t\}$, for most of the time we will directly estimate the parameter $\lm$ using all the available samples by the following recursive estimator, 
\al{
	\lm_t = \lm_{t-1} +  P_t \psi_t\br{u_t -  \tilde{b}^{(\cI)}_{1,t-1} y_{t+1}  - \lm'_{t-1} \psi_t },\label{def:rec_lambda}
}
where $P_t$ is obtained recursively as follows,
\al{
	P^{-1}_t = P^{-1}_{t-1} + \psi_t \psi'_t.
}


\textbf{\emph{Exploitation}:} Re-write the system equation~\eqref{def:arx} as follows,
\al{
	y_{t+1}= b_1 \br{ u_t - \lm'\psi_t  }+ w_{t+1}, t=1,2,\ldots.\label{def:re_arx}
}
The inputs are chosen according to the certainty equivalence (CE) rule, i.e., we assume $\lm_{t-1}$ is the true value of the optimal gain which yields minimum variance.~More specifically, for times $t\not\in \cI$, we have,\footnote{We note that a similar control law $u_t = k_l \vee \frac{\hat{a}_t}{\hat{b}_t}\wedge k_u$ was proposed in econometrics 
\cite{anderson1976some} for the simple model $y_{t+1}=a+bu_t+w_{t+1}$.}
\al{
	u_t = \br{-B_u} \vee z_t  \wedge \br{B_u},
	\label{def:control_rule}
}
where,
\al{ 
z_t := 
\begin{cases}
	\lm'_{t-1} \psi_t ~\mbox{ if } \Big| \lm'_{t-1} \psi_t - \tilde{\lm}^{(\cI)}_{t-1} \psi_{t} \Big| \le B_2 \times \frac{\log N^{(\cI)}_t}{\sqrt{N^{(\cI)}_t}}\|\psi_{t}\|, \\
	\tilde{\lm}^{(\cI)}_{t-1} \psi_{t} \mbox{ otherwise }
\end{cases},
\label{def:_z_t}
}
where $\vee,\wedge$ denote maximum and minimum operators respectively, and the parameter $B_2>0$ is user-specified.
Note that $\tilde{\lm}^{(\cI)}_t$ is used to provide ``diagnostic checks'' on $\lm_{t-1}$, i.e., in the event that the inputs prescribed by $\tilde{\lm}^{(\cI)}_t$ and $\lm_{t-1}$ differ significantly, the algorithm detects that the input prescribed by $\lm_{t-1}$ is ``bad'' and  falls back on the estimate $\tilde{\lm}^{(\cI)}_{t-1}$. 

\textbf{\emph{Clipping Inputs:}} Let $M(\Theta) := \br{1+ \sup_{\te\in\Theta} \frac{C_1(\te)}{1-\rho(\te)} \left\{1+  \sum_{\ell=1}^{q} |b_\ell(\te)|  \right\}}$. The clipping threshold $B_u$~in \eqref{def:control_rule} is given by $B_u = \frac{B_w}{\delta^2_1}\br{1+ M(\Theta)}$, where $\delta_1$ is any constant that satisfies the following inequalities,
\al{
\delta_1 &\le \frac{1}{(p+q) \sup_{\te\in\Theta} \br{1+ \|\lm(\te)\| } },	\br{\sup_{\te\in\Theta}\|\lm(\te)\|+1} \left[M(\Theta)+q\right] \delta_1 \le 1\\
~\delta_1  &\le \frac{1}{3}  \inf_{\te\in\Theta}\left[\frac{b_1(\te)\br{1-\rho(\te)}}{C_1(\te)}\right]
\left[   \frac{\delta^2_1}{2B_w M(\Theta)} +\sup_{\te\in\Theta} \sum_{\ell=1}^{p}| a_\ell(\te) | \right]^{-1}.
}
To see why a solution exists, we note that the first two inequalities admit solution set trivially, while in the third case, the l.h.s. is a monotone decreasing function with value $0$ for $\delta_1=0$, while the r.h.s. is increasing and has a positive value for $\delta_1=0$.
\begin{algorithm}[ht]
   \caption{Probing Inputs for Exploration in Certainty Equivalence (PIECE)}
   \label{algo}
\begin{algorithmic}
   \STATE {\bfseries Input} The exploration set~$\cI$, $B_2>0$.
    \IF{$t \in \cI$} 
  \STATE Generate an exploratory white noise input $u_t$ such that $|u_t|\le B_w$ 
 and has mean $0$.
    \ELSE 
    \STATE  Compute the estimates $\tilde{\theta}^{(\cI)}_{t-1}$, $\tilde{\lm}^{(\cI)}_{t-1}$  and $\lm_t$ as defined in \eqref{eq:theta_estimate}.
    \STATE
    \nal{
	u_t = 
	\begin{cases}
		\br{-B_u} \vee \br{\lm'_{t-1} \psi_t} \wedge \br{B_u} ~\mbox{ if } \Big| \lm'_{t-1} \psi_t - \br{\tilde{\lm}^{(\cI)}_{t-1} }'\psi_{t} \Big| \le B_2 \times \frac{\log N^{(\cI)}_t}{\sqrt{N^{(\cI)}_t}}\|\psi_t\|, \\
		\br{-B_u} \vee \br{  \br{\tilde{\lm}^{(\cI)}_{t-1} }' \psi_{t}} \wedge \br{B_u} ~\mbox{ otherwise }.
	\end{cases}
}\ENDIF
\end{algorithmic}
\end{algorithm}
\section{Regret Analysis}\label{sec:regret_main}
We now state our key results which quantify (i) an upper-bound on regret, and (ii) the estimation error $\|\te\ust - \te_t\|$ of the PIECE algorithm. We perform the analysis under the two separate assumptions on $\{w_t\}$. Section~\ref{sec:bounded_noise} considers the case when the noise is uniformly bounded, i.e. $|w_s|\le B_w$, and shows that the regret of the algorithm is upper-bounded by $C\log T  + D$. This is relaxed in Section~\ref{sec:unbounded}, where the noise is allowed to be unbounded, but has to be conditionally sub-Gaussian (Assumption~\ref{assum:sub_gaussian}). Then PIECE suffers a regret that is at most $C \log^2 T +D$.
 Precise values of constants and bounds are given in the Theorems below and in the Appendix.

\subsection{Bounded Noise}\label{sec:bounded_noise}
\begin{theorem}\label{thm:regret_bounded}
Consider the ARX system~\eqref{def:arx} in which $\{w_t\}$ satisfies Assumptions~(\ref{assum:noise_var},\ref{assum:bounded_noise}).~For every $\delta>0$, there is a set having probability at least $1-6\delta$, such that for every $\eps>0$ the cumulative regret until $T$ can be bounded as follows,
\al{
\cR_T \le \br{1+c(\eps)}\sigma^2(p+q-1) \log T 
+ H L_1(\rho) \sqrt{\log T} + L_2(\eps,\delta,\rho),\label{res:mainres_bounded}
}
where $c(\eps):= \br{1- \frac{\br{1+\eps}^2}{\br{1-\eps}^2\br{1-2\eps}  }}^{-1}\left[\frac{1 }{\br{1-\eps}^2\br{1-2\eps} }\right]-1\to0$ as $\eps\to0$, 
while~$L_1(\rho) \to \infty$ as $\rho \nearrow 1$, and $L_2\to\infty$ as $\eps \searrow 0$, $\rho \nearrow 1$, or $\delta \searrow 0$.

\end{theorem}
\begin{proof}[Outline of Proof]
    The instantaneous regret at time $t$, denoted $r_t$, can be shown to be equal to $b^2_1\br{ u_{t-1} - \lm'\psi_{t-1}  }^2$. We analyze this separately for exploratory times $t\in \cI$ and for $t\notin \cI$.
    
    \emph{Regret during $t\in \cI$}: We bound $r_t$ for $t\in\cI$ by $|u_{t-1}|^2$ plus terms $\leq c' \|\psi_{t-1}\|^2$.~We derive an upper-bound on $|y_t|$ that holds uniformly for all times after an initial phase. Since during $\cI$, the magnitude of input is bounded by $B_w$, upon combining this with the bound on $|y_t|$ it yields an upper-bound on $\|\psi_t\|$. This shows that the regret incurred during the exploratory episodes is $\leq c'' N\ui_t$. 
    
    \emph{Regret during $t\notin \cI$}: To derive an upper-bound, we relate the instantaneous regret $r_t$ with the ``prediction error'' $e_t := y_{t+1} -b_{1,t}\br{u_{t} - \lm'_{t-1 } \psi_{t} }$, where $b_{1,t}\br{u_{t} - \lm'_{t-1 } \psi_{t} }$ is the prediction of the algorithm about the next observation; if $\te\ust$ were known, this error would have been $w_{t+1}$. This observation allows us to show that the instantaneous regret can be bounded by the ``mismatch'' $(e_t -w_t)^2$, but only when this mismatch is ``not too large.'' Following this, the proof for $t\notin \cI$ is split into the following two parts: 
    \begin{enumerate}[(i)]
        \item Ensuring that under the proposed algorithm, the mismatch $(e_t -w_t)^2$ does not become too large. While this cannot be ensured at all  times and for all sample paths, we show that under PIECE algorithm, this does hold for most of the time steps  
        after a sufficiently large duration, with a high probability.~We show that a sufficient condition for this to occur is that under PIECE, the inputs $u_t$ are not clipped too often, i.e., the condition $|z_t|< B_u$ holds. To ensure this, PIECE (i) uses exploratory episodes of sufficiently large duration $H$~\eqref{def:n_i}, (ii) explores using white noise of sufficiently small magnitude ($|u_t|\le B_w,~t\in \cI$). Since the roots of the polynomials are strictly inside the unit circle (Assumption~\ref{assum:unit_circle}), when the estimation error $\|\te\ust-\te_t\|$ is sufficiently small, the magnitude of the output for times $t\notin \cI$ lying between two consecutive episodes can be bounded.  
        \item Deriving an upper-bound on the cumulative mismatch $\sum_{t\notin \cI}(e_t -w_t)^2$. The analysis relies upon a recursion for the quantity $q_t := 
        \Tr\br{b^2_1\br{\lm_t - \lm}P^{-1}_t\br{\lm_t -\lm}'}$. Upon summing up this recursion, it can be shown  that after sufficiently large $t$, with high probability, the mismatch can be controlled by deriving upper-bounds on six terms which mostly involve ``discrete-time martingale transforms.'' The rest of the analysis relies upon carefully bounding these terms using concentration results for self-normalized martingales~\cite{pena2009self,abbasi2011improved} and the Azuma-Hoeffding inequality for unbounded martingale difference sequences~\cite{tao_vu}.   
    \end{enumerate}
In both (i) and (ii) above, we need to control the estimation error $\|\te_t-\te\ust\|$ associated with LSE. 
Therefore, we provide the following finite-time guarantees on the performance of the LSE operating under PIECE algorithm in Theorem~\ref{th:estimation_error_main}.    
\end{proof}

We note that by letting $\eps \searrow 0$, we are able to match the pre-constant as well as the logarithmic growth rate of the asymptotically optimal regret (\ref{asymptotically-optimal-regret}) of~\cite{lai1987asymptotically}. Furthermore, our bounds also quantify the transient performance and how it is affected by various parameters such as $\delta,B_w,B_u$, and the operator norm dependent quantity $\rho$.

\begin{theorem}\label{th:estimation_error_main}
Consider the ARX system~\eqref{def:arx} in which $\{w_t\}$ satisfies Assumptions~(\ref{assum:noise_var},\ref{assum:bounded_noise}), and LSE is given by~\eqref{def:lse_I}.~On a high-probability set having probability greater than $1-4\delta$, the estimation error can be bounded as follows,
	\al{
		\|\te\ust - \te_t\|^{2} \le \frac{\nicefrac{1}{2}\br{p+q}\log \br{\br{C_1  \| Y_{0}\|  +  \frac{C_1 B_u}{1-\rho}\left\{1+  \sum_{\ell=1}^{q} |b_\ell|  \right\} + q B_u}^2 N\ui_t } - \log\br{\delta}}{N\ui_t}.\label{eq:bound_I}
	}	
\end{theorem}
\begin{proof}[Outline of Proof]
Let $V_t = \sum_{s=1}^{t}\phi_s \phi'_s$ be the covariance matrix associated with the LSE at $t$. It can be shown that w.h.p. 	
$\|\te\ust - \te_t \|^{2} \leq c''\br{\frac{\log \lm_{\max}(V_t)}{\lm_{\min}(V_t)}}$, and hence it suffices to upper-bound $\lm_{\max}(V_t)$ and lower-bound $\lm_{\min}(V_t)$. $\lm_{\max}(V_t)$ is bounded by $\sum_s \|\phi_s\|^2$, which in turn is shown to be $\leq c'''tB^2_u$. The key challenge in the proof is to derive a lower-bound on $\lm_{\min}(V_t)$ that holds for all times $t$ greater than some finite time w.h.p.
\end{proof}

\subsection{Unbounded Noise}\label{sec:unbounded}
To deal with unbounded noise, we slightly modify the PIECE algorithm as follows. Firstly, the exploratory episodes are changed to have $n_i = \exp(i)$, and, in the definition of $m\ust$, $B_w$ is replaced by $\sqrt{\log\br{T\slash \delta}}$, with this quantity serving as a high-probability upper-bound on $\{w_t\}$. While deciding the threshold $B_u$ for clipping inputs, once again $B_w$ is replaced by $\sqrt{\log\br{T\slash \delta}}$. Let $\tilde{H}$ be the resulting episode duration.
\begin{theorem}\label{thm:regret_unbounded}
Under Assumption \ref{assum:sub_gaussian}, the regret of  PIECE can be bounded as follows: For every $\delta>0$, there is a set having probability atleast $1-7\delta$ such that for every $\eps>0$, the cumulative regret until $T$ can be bounded as follows,
\al{
\cR_T \le \br{1+c(\eps)} \left[\sigma^2(p+q-1) \log T + \log\br{T\slash \delta} \right]
+ \tilde{H} \tilde{L}_1(\rho)\log T + \tilde{L}_2(\eps,\delta,\rho),\label{res:mainres_unbounded}
}
where $c(\eps)$ is as in Theorem~\ref{thm:regret_bounded},~while~$\tilde{L}_1(\rho) \to \infty$ as $\rho \nearrow 1$, and $\tilde{L}_2\to\infty$ as $\eps \searrow 0$, $\rho \nearrow 1$, or $\delta \searrow 0$. 
\end{theorem}
We note that in comparison with Theorem~\ref{thm:regret_bounded}, there is an additional $\log\br{T\slash \delta}$ term that arises due to an increase in the high probability upper-bound on the norms of $\|Y_t\|,\|U_t\|$ . This, in-turn happens due to an increase in the magnitudes of noise, exploratory inputs and inputs during the exploitation phase as compared with the bounded noise case. It is shown in~\cite{lai1987asymptotically} that the regret of LW is asymptotically $\sigma^2(p+q-1)\log T$ under the assumption that $\sup_{t} \bE\left\{\exp \br{\gamma |w_t |} \Big| \cF_{t-1} \right\} < \infty,~\mbox{ a.s. }$ for some $\gamma>0$. It remains to be seen if the finite-time regret of our proposed algorithm can be improved so that it matches this asymptotically as $T\to\infty$.

\section{Simulations}\label{sec:simulations}
In this section, we compare the performance of the PIECE algorithm with the algorithm proposed in \cite{lai1986asymptotically} (LW), as well as the standard CE controller. Each simulation experiment is performed for $1000$ steps.
The reported results are the averaged values over the $50$ runs. Results for more examples, and technical details on implementation, are provided in the
Appendix. 
The examples of the ARX systems considered in the experiments are the following:\\
\textsc{Example I:} This represents a linear system with $p=4$ and $q=4$ given by $y_t= 1.18 y_{t-1} - 0.48 y_{t-2} + 0.45 y_{t-3} - 0.41 y_{t-4}+ 0.28 u_{t-1} + 0.14 u_{t-2} + 0.16 u_{t-3} + 0.03 u_{t-4} + w_t$.\\\\
\textsc{Example II:} This represents a linear system with $p=2$ and $q=3$ given by $
        y_t = - 0.01 y_{t-1} - 0.46 y_{t-2} + 0.1 u_{t-1} + 0.086 u_{t-2} + 0.02 u_{t-3} + w_t,$
\label{def:example2}\\\\
\textsc{Example III:} This represents a linear system with $p=6$ and $q=6$ given by
        $y_t = - 0.66 y_{t-1} - 0.79 y_{t-2} + 0.2 y_{t-3} - 0.03 y_{t-4} + 0.09 y_{t-6} + 0.32 u_{t-1} + 0.06 u_{t-2} - 0.2 u_{t-3} - 0.01 u_{t-4} - 0.03 u_{t-5} + 0.001 u_{t-6} + w_t$
\label{def:example3}\\\\
\textbf{Cumulative Regret}
In Figure \ref{fig:cumulative_regret}, we plot the logarithm of the cumulative regret, $\log(R_t)$. Table \ref{table:1} highlights the cumulative regret at the end of the experiment. One of the key issues with many adaptive controllers is their empirical performance in the initial phase of learning \cite{lale22,mete2022augmented}. It is evident from the empirical results that CE as well as LW both suffer from this issue. As described in Section \ref{sec:algo}, the PIECE algorithm differs from LW with regard to the clipping of the input as well as the choice of exploration episodes. The benefits of these modifications are clearly evident as they lead to significantly lower regret in the initial stages of the experiment. 

\textbf{Estimation Error:}
In Figure \ref{fig:estimation_error}, we plot the estimation error $||\theta_t-\te\ust||^2$. It is interesting to note that LW has better estimation error than PIECE. This reiterates the point that the exploration scheme in PIECE is more efficient in achieving lower regret, which is the primary objective of the controller at the cost of a higher estimation error.
  \begin{figure}[ht]
     \centering
          \begin{subfigure}[b]{0.32\textwidth}
         \centering
         \includegraphics[width=\textwidth]{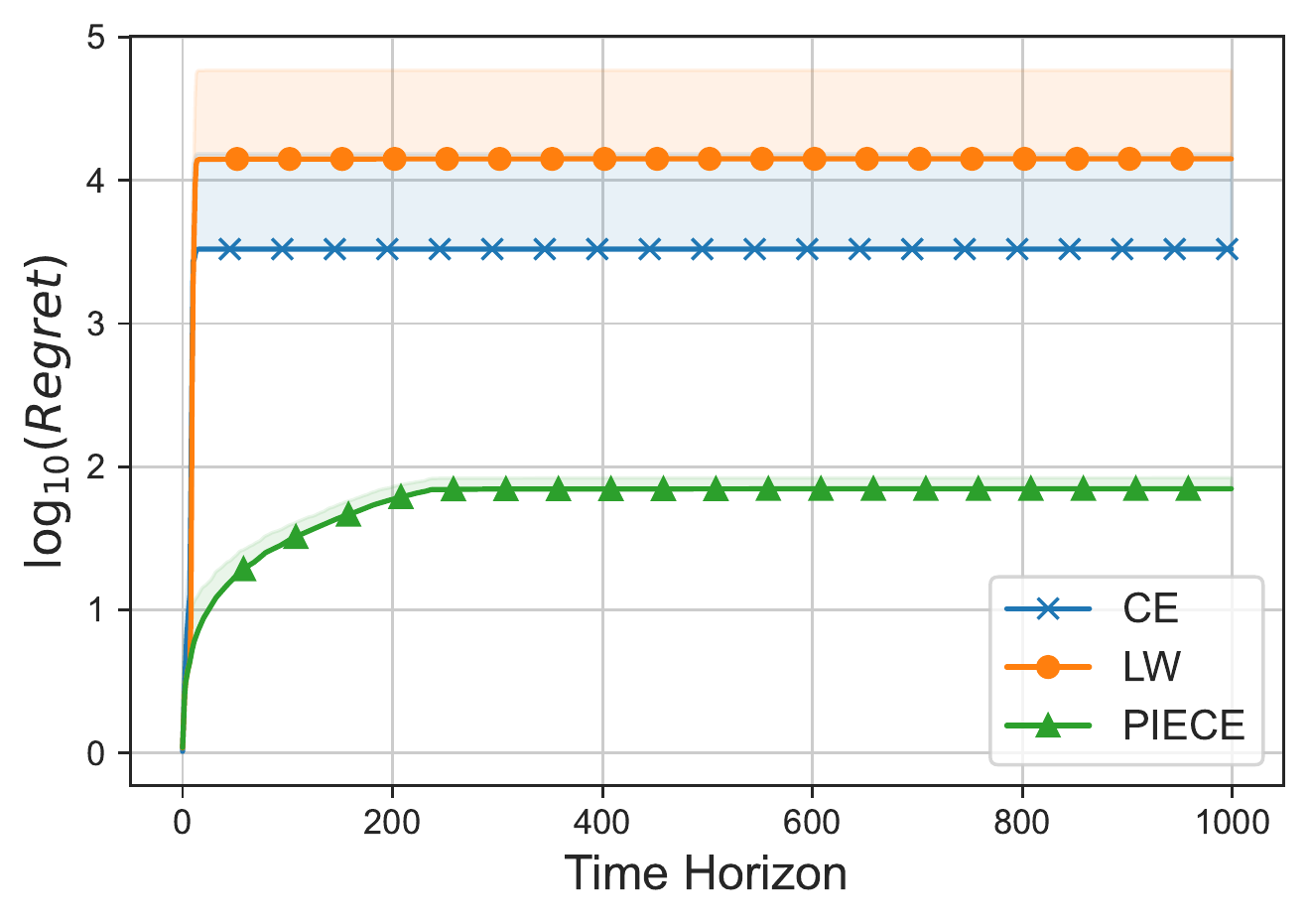}
         \caption{Example I}
         \label{fig:2_ex_1}
     \end{subfigure}
     \hfill
          \begin{subfigure}[b]{0.32\textwidth}
         \centering
         \includegraphics[width=\textwidth]{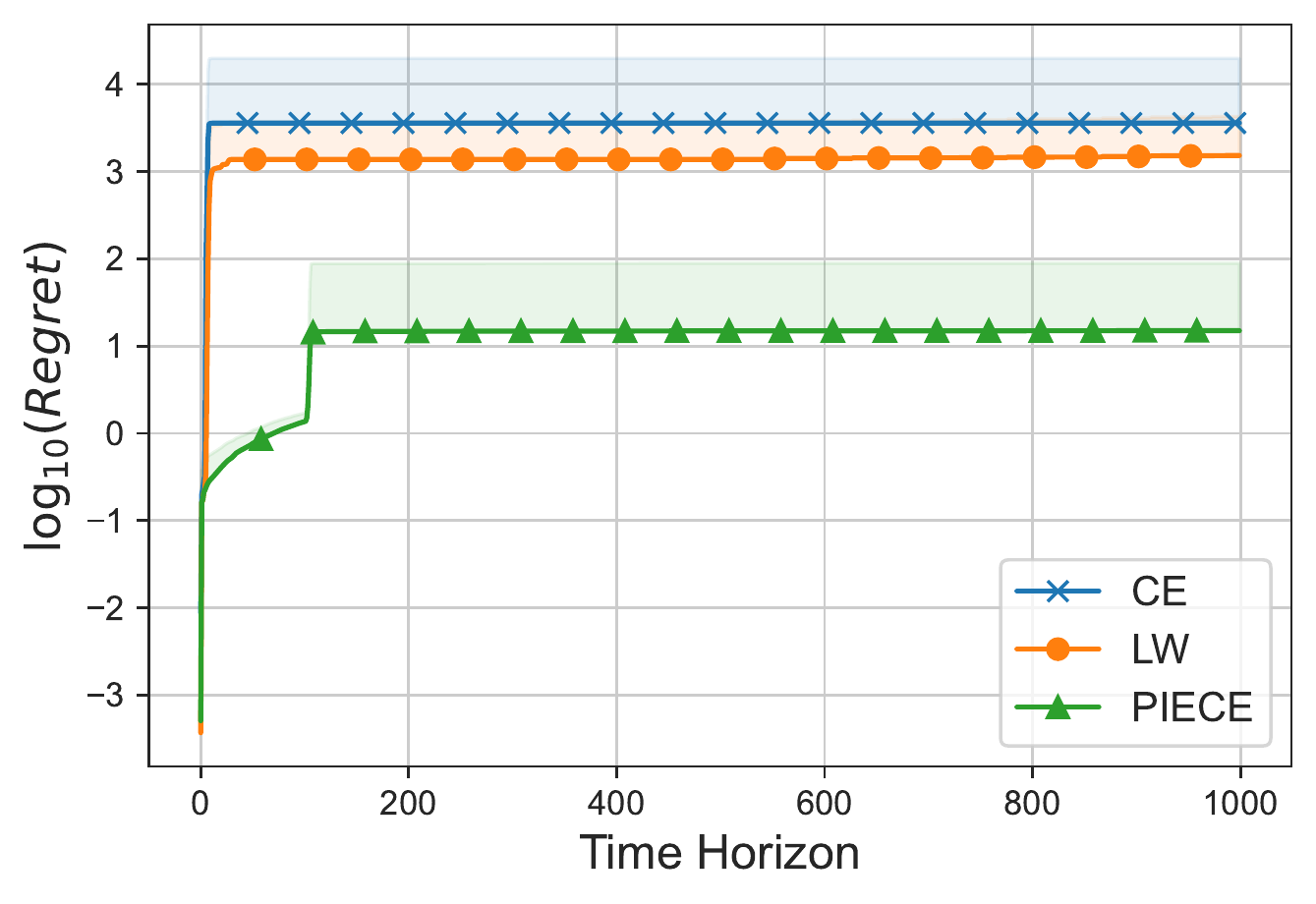} 
         \caption{Example II}         
         \label{fig:2_ex_2}
     \end{subfigure}
          \hfill
               \begin{subfigure}[b]{0.32\textwidth}
         \centering
         \includegraphics[width=\textwidth]{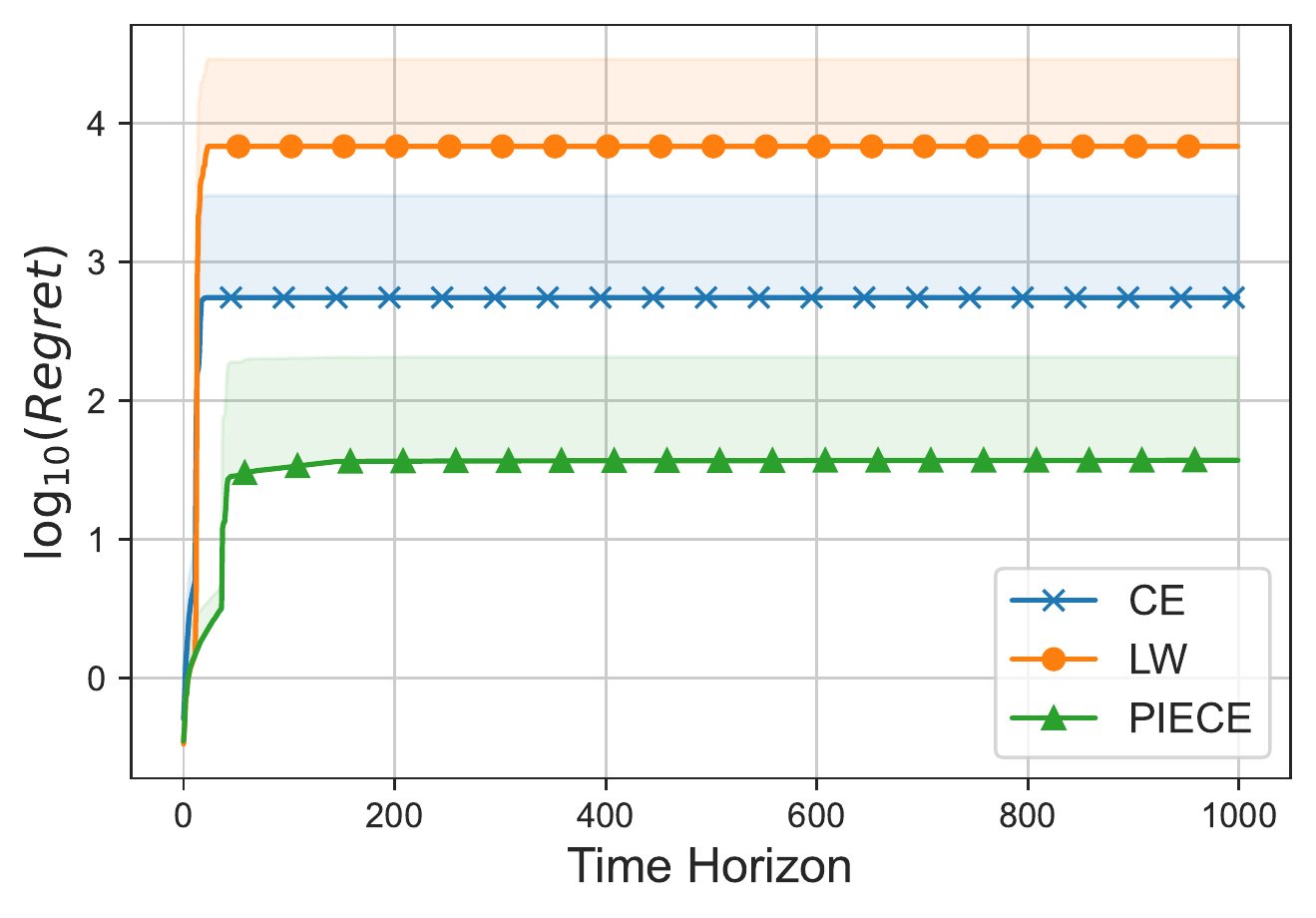}
         \caption{Example III}
         \label{fig:2_ex_3}
     \end{subfigure}
        \caption{Log(Cumulative Regret) averaged over 50 runs.}
        \label{fig:cumulative_regret}
\end{figure}
\begin{figure}[ht]
     \centering
          \begin{subfigure}[b]{0.32\textwidth}
         \centering
         \includegraphics[width=\textwidth]{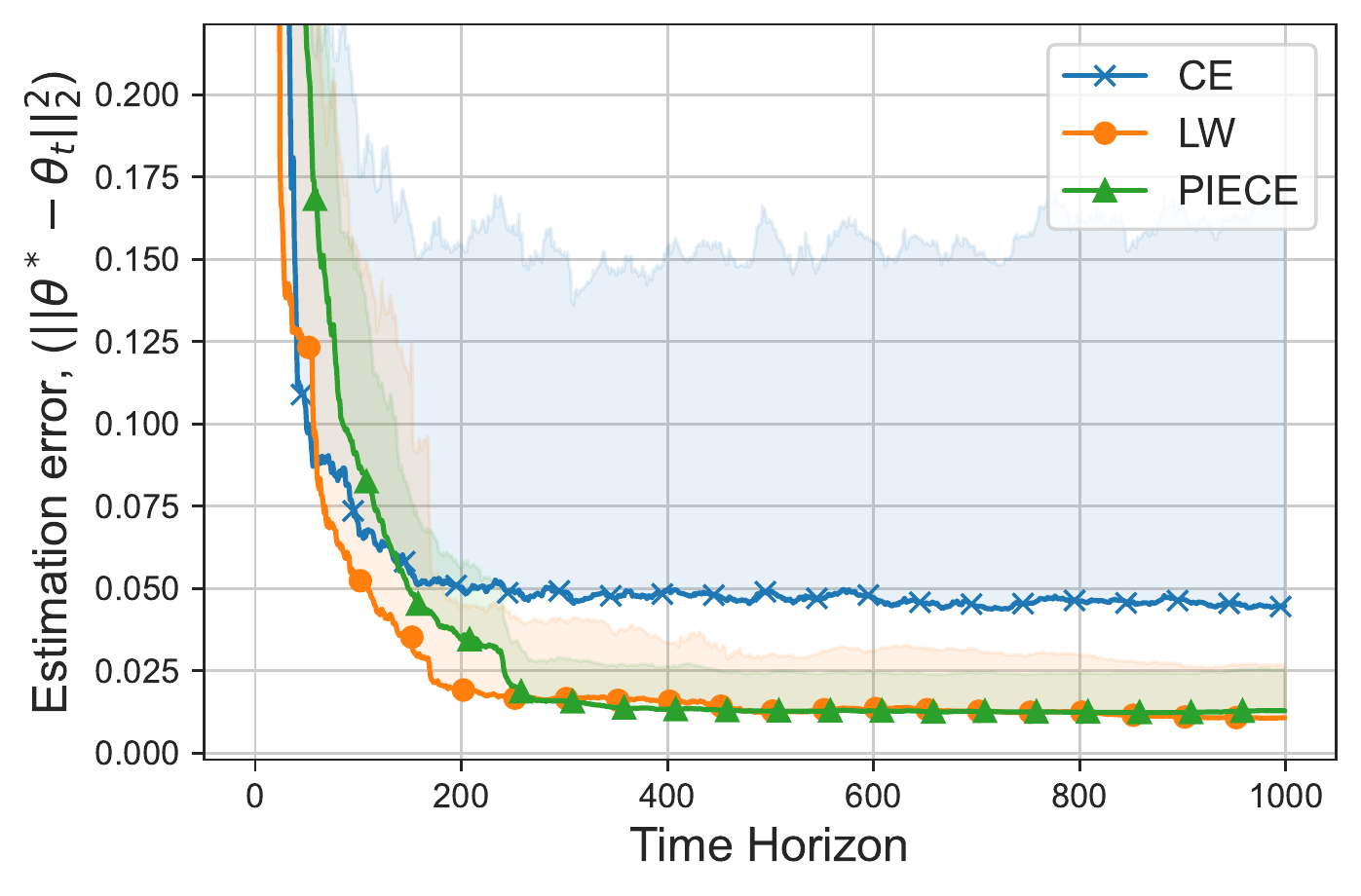}
         \caption{Example I}
         \label{fig:1_ex_1}
     \end{subfigure}
     \hfill
          \begin{subfigure}[b]{0.32\textwidth}
         \centering
         \includegraphics[width=\textwidth]{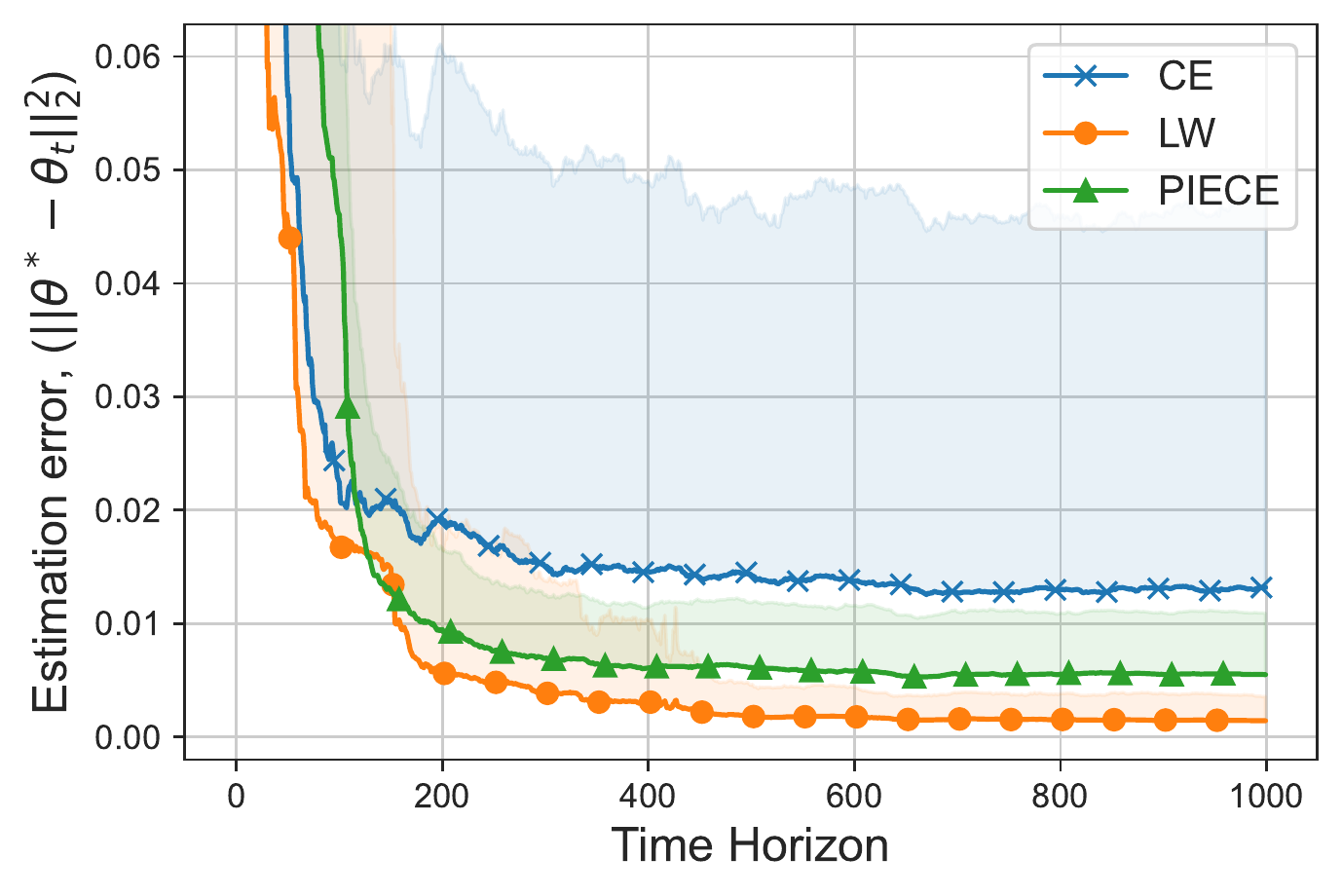}      
         \caption{Example II}         
         \label{fig:1_ex_2}
     \end{subfigure}
     \hfill
    \begin{subfigure}[b]{0.32\textwidth}
         \includegraphics[width=\textwidth]{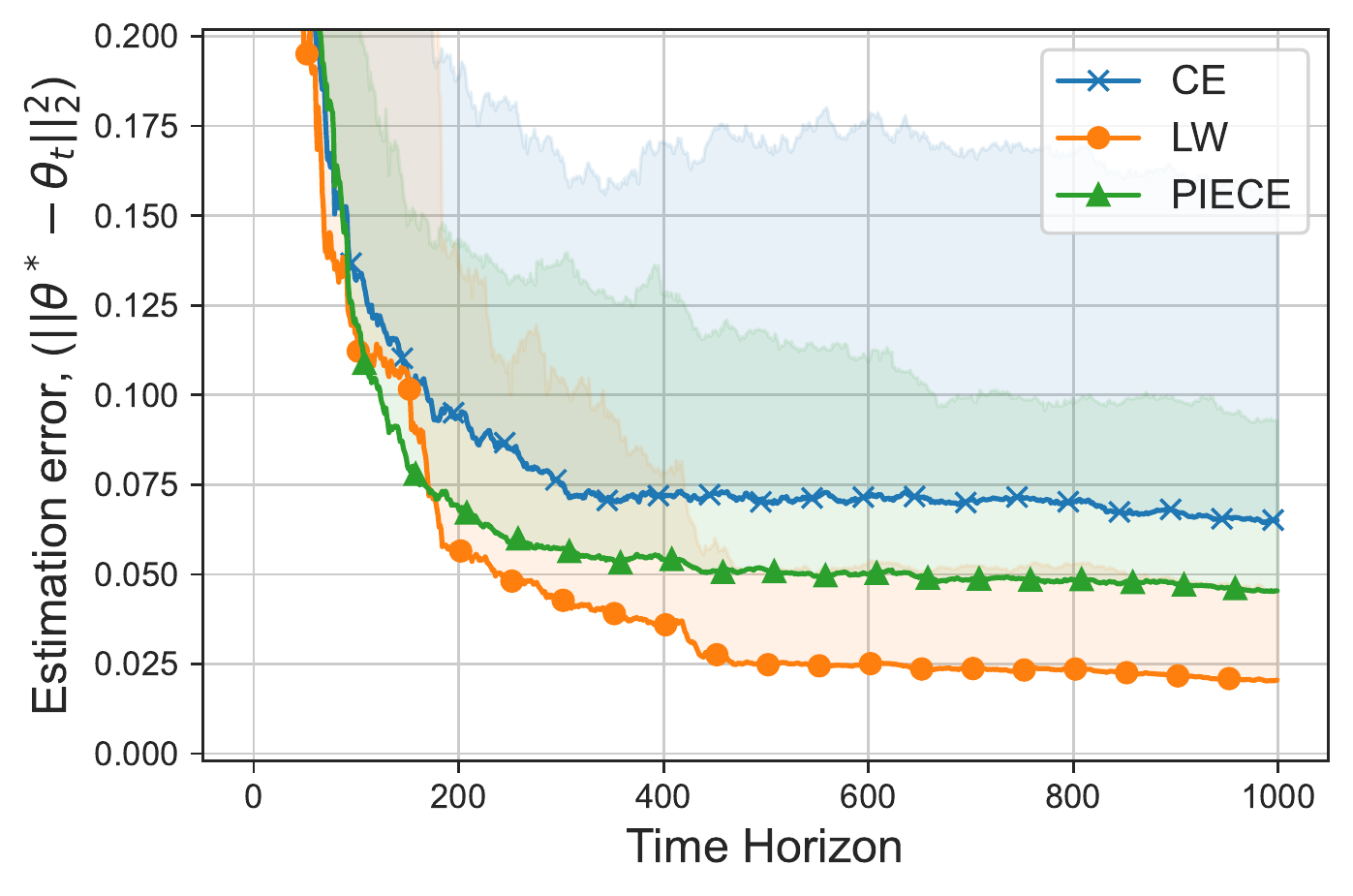}
         \caption{Example III}
         \label{fig:1_ex_3}
     \end{subfigure}
        \caption{Estimation Error ($||\theta^\star-\theta_t||^2_2$)}
        \label{fig:estimation_error}
\end{figure}  

\begin{center}
\begin{table}
    \centering
\begin{tabular}{|c |c | c| c |} 
\hline
Example & PIECE & LW & CE\\ \hline
I & 70.13 & 14124.45&  3308.96  \\ \hline 
II & 15.03& 1528.95 & 3610.14\\\hline   
III & 37.34& 6852.85 & 557.35\\\hline   
\end{tabular}
    \caption{Cumulative Regret Performance at $T=1000$.}
    \label{table:1}
\end{table}
\end{center}

\section{Conclusion}\label{sec:discussion}
In this paper, we have provided the first finite time regret bound for an adaptive minimum variance control problem. We analysed two different scenarios. (i) When the system noise is bounded, the regret of the PIECE algorithm is $C\log T$. 
(ii) When system noise is unbounded ,  the regret of the PIECE algorithm is $C\log^2 T$. 
We have also verified through simulations the advantage of the PIECE algorithm over LW and the standard CE controller. 

Whether the bound in the unbounded noise case can be improved to $C\log T$ remains an interesting question. A natural next step is to analyze performance of similar algorithms for an ARMAX system. 
One can potentially adapt similar algorithms which use ``probing inputs'' in other various reinforcement learning settings, including Markov Decision Processes and LQG systems.

\bibliography{ref}
\bibliographystyle{unsrt}
\clearpage
\tableofcontents
\clearpage
\appendix
\section{Organization of the Appendix}
In Appendix \ref{sec:regret_analysis}, we provide the proof of the regret bound that was provided in Theorem \ref{thm:regret_bounded} for the case of bounded noise. Appendix~\ref{sec:unbounded_noise} discusses how to extend this to the case of sub-Gaussian noise. Proof of Theorem~\ref{th:estimation_error_main}, i.e. the bound on estimation error, is provided in Appendix~\ref{sec:estimation_error}. Proof of Theorem~\ref{th:estimation_error_main} relies crucially upon the analysis done in Appendix~\ref{sec:lmin}. 
The design of PIECE algorithm involves the choice of length of the first episode and the clipping function $B_u$ which is discussed in Appendix \ref{sec:first_epi} and Appendix \ref{sec:threshold} respectively.
Finally, the details of simulation setup and additional results are provided in Appendix~ \ref{sec:simulations_appendix}.
\section{Regret Analysis (Proofs)}\label{sec:regret_analysis}
Let $r_t$ be the instantaneous at time $t$.~The regret equation~\eqref{def:cr} and re-parameterization of ARX model~\eqref{def:re_arx} yield,
\al{
	r_t  &= (y_t - w_t)^2 \notag \\
	& = b^2_1\br{ u_{t-1} - \lm'\psi_{t-1}  }^2. \label{eq:inst_reg}
}

The behavior of the term $u_{t} - \lm'\psi_{t}$ is different for times $t\not\in \cI$ and $t\in\cI$, and so we study them separately.

\subsection{$r_t$ for $t\in\cI$}\label{sec:regret}
We will discuss only the case of bounded noise, i.e. when $\{w_t\}$ satisfies Assumption~\ref{assum:bounded_noise}, since the proof under Assumption~\ref{assum:sub_gaussian} follows through using similar arguments by restricting to the set $\cG_w$ defined in \eqref{def:g1}.

For $t\in\cI$, the instantaneous regret~\eqref{eq:inst_reg} can be bounded as follows:
\nal{
	r_t &= b^2_1\br{ u_{t-1} - \lm'\psi_{t-1}  }^2\\
	&\le 2~b^2_1 \br{ |u_{t-1}|^2 + |\lm'\psi_{t-1}|^2   }\\
	& \le 2~b^2_1 \br{B^2_w+  |\lm'\psi_{t-1}|^2  }\\
	& \le 2~b^2_1 \br{B^2_w+ \|\lm\|^2 	\|\phi_t\|^2} \\
	& \le 2~b^2_1 \br{B^2_w+ \|\lm\|^2 	\br{\|U_t\|^2 + \|Y_t\|^2}    } \\
	& \le 2~b^2_1 \br{B^2_w+ \|\lm\|^2 	\br{q B^2_u +\left[C_1 \rho^{t} \| Y_{0}\| +\frac{B_u C_1}{1-\rho}  \left\{1+  \sum_{\ell=1}^{q} |b_\ell|  \right\}\right]^2}    } \\
	& \le 2~b^2_1 \br{B^2_w+ \|\lm\|^2 	\br{q B^2_u +\left[C_1 \| Y_{0}\| +\frac{B_u C_1}{1-\rho}  \left\{1+  \sum_{\ell=1}^{q} |b_\ell|  \right\}\right]^2}  ,  } 
} 
where the second inequality follows since $|u_t|\le B_w$ by the design of the algorithm, the
third inequality follows since $\|\psi_t\|\le \|\phi_t\|$, and the fifth inequality follows from Lemma~\ref{lemma:bound_phi}.




We therefore have the following bound on the cumulative regret incurred during the exploration steps $\cI$:
\begin{lemma}\label{lemma:regret_I}
	When $\{w_t\}$ satisfies Assumption~\ref{assum:bounded_noise},
	\al{
		\sum_{t\in\cI} r_t \le 2~b^2_1 \br{B^2_u+ \lm^2 	\br{q B^2_u +\left[C_1 \| Y_{0}\| +\frac{B_u C_1}{1-\rho}  \left\{1+  \sum_{\ell=1}^{q} |b_\ell|  \right\}\right]^2}    }  N\ui_t,\label{ineq:15}
	}
	where $N\ui_t$ is the number of exploratory steps until $t$, and $C_1,\rho$ are as discussed in Section~\ref{sec:bound_yt}. The same conclusion holds on the set $\cG_w$ where $\{w_t\}$ instead satisfies Assumption~\ref{assum:sub_gaussian}.
\end{lemma}
\subsection{$r_t$ for $t\notin \cI$}
On $t\notin \cI$, the input $u_t$ is chosen according to the rule~\eqref{def:control_rule}. This rule can be written equivalently as follows. Define 
\al{
z_t := 
\begin{cases}
	\lm'_{t-1} \psi_t ~\mbox{ if } \Big| \lm'_{t-1} \psi_t - \br{\tilde{\lm}^{(\cI)}_{t-1}}' \psi_{t} \Big| \le B_2  \frac{\log N^{(\cI)}_t}{\sqrt{N^{(\cI)}_t}}\|\psi_{t}\|, \\
	\br{\tilde{\lm}^{(\cI)}_{t-1}}' \psi_{t} \mbox{ otherwise }
\end{cases},
}
so that~\eqref{def:control_rule} is equivalently
\al{
u_t = \br{-B_u} \vee z_t \wedge B_u.
}
The re-parametrization~\eqref{def:re_arx} suggests that one can view the quantity 
$\tilde{b}\ui_{1,t-1}\br{u_{t} - \lm'_{t-1 } \psi_{t} }$ as the prediction of $y_{t+1}$ based on the information available until $t$ .~Hence, define the prediction error at time $t+1$ by
\al{
	e_{t+1} :&= y_{t+1} -\tilde{b}^{(\cI)}_{1,t-1}\br{u_{t} - \lm'_{t-1 } \psi_{t} }\notag\\	
	& = b_1 \br{u_t - \lm' \psi_t  }-\tilde{b}\ui_{1,t-1}\br{u_{t} - \lm'_{t-1 } \psi_{t} } + w_{t+1},
	\label{def:error}
}
where the second equality follows from~\eqref{def:re_arx}.~It is shown in Theorem~\ref{th:reg_n_i} that for times $t\ge \max\{t\ust_1(\rho),t\ust_2(\rho),t\ust_3(\rho)\}$, where $t\ust_1(\rho),t\ust_2(\rho),t\ust_3(\rho)$ are as in Definition~\ref{def:t_1}, the instantaneous regret for $t\notin \cI$ can be bounded by the quantity $\br{e_t - w_t}^2$. Hence, we will now focus on bounding $\sum_{t\ge \max\{t\ust_1,t\ust_2,t\ust_3\}} \br{e_t - w_t}^2$. Instead, we will bound $\sum_{t\ge t\ust} \br{e_t - w_t}^2$, where $t\ust$ is as in~\eqref{def:t_star}, since bounding this expression is simpler.

We begin with some definitions. Define,
\al{
g_t := \br{\lm_t - \lm}P^{-1}_t\br{\lm_t -\lm}',
}
where $P^{-1}_t$ is obtained recursively as follows,
\al{
	P^{-1}_t = P^{-1}_{t-1} + \psi_t \psi'_t.\label{def:P}
}
Also let,
\al{
q_t :&= \Tr\br{b^2_1 g_t},\label{def:qt}\\
\gamma_t :&= \frac{b_1}{b_{1,t}}
}
The following result is essentially (4.12),~(4.13) of~\cite{lai1986asymptotically}:
\begin{lemma}\label{lemma:recursion_qt}
	For the ARX model~\eqref{def:arx}, we have the following recursion for $q_t$ for times $t\ge \tau$, 
	\al{
		q_t -q_{t-1} &=  -\left[ b_1 \br{1-\gamma_t}\br{u_t - \lm' \psi_t } - b_1 \br{\lm_{t-1} - \lm  }'\psi_t	\right]^2  			\notag	\\
		& + \left[ b_1 \br{1-\gamma_t} \br{u_t - \lm' \psi_t}  \right]^2 -2 \gamma_t w_t b_1 \br{\lm_{t-1} - \lm}'\psi_t \notag \\
		& + \br{ \psi'_t P_t \psi_t } \left[ b_1 \br{1-\gamma_t}\br{u_t - \lm' \psi_t } - b_1 \br{\lm_{t-1} - \lm  }'\psi_t	-\gamma_t w_t\right]^2,~t=1,2,\ldots.
		\label{eq:q_recursion}
	}
\end{lemma}
Define the times,
\al{
t\ust_5(\eps_1) & := \frac{2\br{C_1\| Y_{0}\| }^2+2\br{\frac{B_u C_1}{1-\rho}  \left\{1+  \sum_{\ell=1}^{q} |b_\ell|  \right\}}^2 + q B^2_u}{\beta_3 \eps_1},~\eps_1 >0\label{def:t5_star} \mbox{  and } \\
	t\ust_6(\eps_3,\delta) & := \inf \left\{t \in\bN: \cE(t;\te\ust,\delta) \le b_1 \eps_3 \right\},~\eps_3>0.\label{def:t_6_star}
}
Also the times $t\ust_1(\rho),t\ust_2(\rho),t\ust_3(\rho)$ defined in Definition~\ref{def:t_1}, where $\rho$ is operator norm of the unknown system as in Section \ref{sec:system_model}. Let,
\al{\label{def:t_star}
t\ust := t\ust_1(\rho)\vee t\ust_2(\rho)\vee t\ust_3(\rho)\vee t\ust_5(\eps_1) \vee t\ust_6(\eps_3,\delta)\vee \tau.
}

Define
\al{
	\cT_2 :  & =  \sum_{t= t\ust}^{T}\left[ b_1 \br{1-\gamma_t} \br{u_t - \lm' \psi_t}  \right]^2,\label{def:t2}\\
	\cT_3 :  & = \sum_{t= t\ust}^{T} \gamma_t w_t b_1 \br{\lm_{t-1} - \lm}'\psi_t, \\
	\cT_4 & := \sum_{s= t\ust}^{T}\br{ \psi'_s P_s \psi_s } \left[ b_1 \br{1-\gamma_s}\br{u_s - \lm' \psi_s } - b_1 \br{\lm_{s-1} - \lm  }'\psi_s	-\gamma_s w_s\right]^2,
 } 
where $\beta_3$ is as in~\eqref{def:beta_3}.~We will occasionally omit dependence of $t\ust_5(\eps_1),t\ust_6(\eps_3,\delta)$ upon $\eps_1,\eps_3,\delta$ to ease notation.~The sets $\cG_q,\cG_{LSE}$ are defined in~\eqref{def:g_q} and~\eqref{def:glse} respectively. The parameter $\eps_1$ here controls the ``degree of excitation,'' i.e. $\lm_{\min}(V_t)$ during the first exploratory episode. The parameter $\eps_3$ controls the estimation error at the end of the first exploratory episode, and $\delta$ is the confidence parameter which decides the probability contained in the corresponding event.
\begin{lemma}\label{lemma:prediction_square_bound}
Let $\eps_1,\eps_2,\eps_3>0$, where $\eps_1,\eps_3$ are as in~\eqref{def:t5_star},~\eqref{def:t_6_star}. When $\{w_t\}$ satisfies Assumption~\ref{assum:bounded_noise}, then on the set $\cG_q \cap \cG_{LSE}\cap \cG_{\pj}\cap \cG_{\cI}$,
		\al{
	\sum_{s =t\ust}^{t} \br{w_s - e_s}^2 \le \frac{1}{\br{1-\eps^2_3}}\max\left\{\frac{\br{\cT_2 + \cT_3 + \cT_{4,2}} }{1- \eps_1 -\eps_2},\frac{\log\br{1\slash \delta}}{\eps^2_2} \right\}+ q_{t\ust},
}
for all $\eps_2 >0$, where $q_t$ is as in~\eqref{def:qt}.~The same conclusion holds on $\cG_q \cap \cG_{LSE}\cap \cG_{\pj}\cap \cG_{\cI}\cap \cG_w$ when $\{w_t\}$ instead satisfies Assumption~\ref{assum:sub_gaussian}.
\end{lemma}
\begin{proof}
The result is obtained by summing the recursions~\eqref{eq:q_recursion}, and bounding each of the terms separately as in the next section.
 
Upon substituting the bounds derived in Section~\ref{sec:T_bounds}, and summing up the recursions~\eqref{eq:q_recursion} for $t\ge t\ust$, we get the following bound on $q_t$:
\nal{
q_t - q_{t\ust}&\le  -\sum_{s=t\ust}^{t}\left[ b_1 \br{1-\gamma_s}\br{u_s - \lm' \psi_s } - b_1 \br{\lm_{s-1} - \lm  }'\psi_s	\right]^2	\notag	\\
			& + \sum_{s=t\ust}^{t} \br{ \psi'_s P_s \psi_s }  \left[b_1 \br{1-\gamma_s}\br{u_s - \lm' \psi_s } - b_1 \br{\lm_{s-1} - \lm  }'\psi_s	\right]^2\\	
			& + \sum_{s=t\ust}^{t} \left[ b_1 \br{1-\gamma_s} \br{u_s - \lm' \psi_s}  \right]^2 \notag \\
			& + \left(\sum_{s=1}^{t}\left\{\gamma_s b_1 \br{\lm_{s-1} - \lm}'\psi_s \right\}^2\right)^{1\slash 2} \log\left[\frac{1}{\delta}\left(\sum_{s=1}^{t}\left\{\gamma_s b_1\br{\lm_{s-1} - \lm}'\psi_s \right\}^2\right)\right] \\	
			& + \br{1+\eps_3} \log\left(\frac{T}{\delta}\right) \log \br{\det \br{\sum_{s=1}^{t}\phi_s \phi'_s}}\notag\\
			& + \br{1+\eps_3} \sqrt{ \sum_{s=1}^{t} \br{b_1 \br{1-\gamma_s}\br{u_s - \lm' \psi_s } - b_1 \br{\lm_{s-1} - \lm  }'\psi_s 	}^2 }\notag \\
			& \times \sqrt{\log \br{\frac{ \br{1+\eps_3}^2 }{\delta}}\br{\sum_{s=1}^{t}  \br{b_1 \br{1-\gamma_s}\br{u_s - \lm' \psi_s } - b_1 \br{\lm_{s-1} - \lm  }'\psi_s	}^2 } }.
		}
		Or equivalently,
		\al{
			q_t - q_{t\ust}+& \sum_{s=t\ust}^{t}\left[ b_1 \br{1-\gamma_s}\br{u_s - \lm' \psi_s } - b_1 \br{\lm_{s-1} - \lm  }'\psi_s	\right]^2 \left\{1-\psi'_s P_s \psi_s \right\}\notag\\
			&\le \cT_2 + \cT_3 + \cT_{4,2} + \cT_{4,3},\label{ineq:59}
		}
  where $\cT_{4,2},\cT_{4,3}$ are as in Section~\ref{sec:T_bounds}.
		We now derive an upper-bound on $\psi'_s P_s \psi_s$.
		Since the duration of the first exploratory episode is greater than $t\ust_5(\eps_1)$, we have $\frac{B^2_u}{N\ui_t} \le \eps_1$ for $t>t\ust_5$.~For $t>t\ust_5(\eps_1)$,
		\al{
			\psi'_{t} P_t \psi_t & \le \frac{\|\psi_t\|^2}{\lm_{\min}\br{\sum_{s=1}^{t} \psi_{s} \psi'_{s}}}\notag\\
			& \le \frac{\|\psi_t\|^2}{\lm_{\min}\br{\sum_{s=1}^{t\ust_5(\eps_1)} \psi_{s} \psi'_{s}}}\notag\\			
			& \le \frac{\|\psi_t\|^2}{\beta_3 N\ui_{t\ust_5}}\notag\\
			& \le \frac{\br{C_1\| Y_{0}\| +\frac{B_u C_1}{1-\rho}  \left\{1+  \sum_{\ell=1}^{q} |b_\ell|  \right\}}^2 + q B^2_u}{\beta_3 N\ui_{t\ust_5}}\notag\\
			& \le \frac{2\br{C_1\| Y_{0}\| }^2+2\br{\frac{B_u C_1}{1-\rho}  \left\{1+  \sum_{\ell=1}^{q} |b_\ell|  \right\}}^2 + q B^2_u}{\beta_3 N\ui_{t\ust_5}}\notag\\			
			& \le \frac{B^2_u}{N\ui_s}\notag\\
			& \le \eps_1,
		}
		where the third inequality follows from Theorem~\ref{th:1}, while the fourth follows from Lemma~\ref{lemma:bound_phi}, and $\beta_3$ is as in~\eqref{def:beta_3}. The last inequality follows since the first exploratory episode is of duration greater than $t\ust_5(\eps_1)$.
		
		
		Denoting 
		\al{
\cS_1 := 		\sum_{s=t\ust}^{t}\left[ b_1 \br{1-\gamma_s}\br{u_s - \lm' \psi_s } - b_1 \br{\lm_{s-1} - \lm  }'\psi_s	\right]^2,
	}
~\eqref{ineq:59} yields,
		\al{
			\cS_1\br{1-\eps_1} \le \cT_2 + \cT_3 + \cT_{4,2} + \sqrt{\cS_1\log \br{\cS_1}} \times \sqrt{\log\br{1\slash \delta} }.
		}
  Let $\eps_2>0$. After performing some algebraic manipulations, we get,
		\al{
			\cS_1&\le \frac{\br{\cT_2 + \cT_3 + \cT_{4,2}} }{1- \eps_1 -\eps_2},~\mbox{ when } \cS_1 \ge \frac{\log\br{1\slash \delta}}{\eps^2_2},\notag\\\label{ineq:73}
			\mbox{ or alternatively } \cS_1&\le   \max\left\{\frac{\br{\cT_2 + \cT_3 + \cT_{4,2}} }{1- \eps_1 -\eps_2},\frac{\log\br{1\slash \delta}}{\eps^2_2} \right\}.
		}
		We will now relate $\cS_1$ to $\sum_{s=t\ust}^{t} \br{e_s - w_s}^2$.

		From Proposition~\ref{prop:10} we have,
		$$
		b_1 \br{1-\gamma_s}\br{u_s - \lm' \psi_s } - b_1 \br{\lm_{s-1} - \lm  }'\psi_s  =  \frac{b_1}{b_{1,s}}\br{w_s - e_s}.
		$$
  Let $\eps_3>0$. For times $s > t\ust_6(\eps_3,\delta)$~\eqref{def:t_6_star}, we have,
		$$
		\br{b_1 \br{1-\gamma_s}\br{u_s - \lm' \psi_s } - b_1 \br{\lm_{s-1} - \lm  }'\psi_s}^2  \ge  \br{1-\eps_3}^2\br{w_s - e_s}^2.
		$$
		Upon summing this up over $s$ and substituting into~\eqref{ineq:73} we obtain, 
		\al{
		\sum_{s =t\ust}^{t} \br{w_s - e_s}^2 \le \frac{1}{\br{1-\eps^2_3}}\max\left\{\frac{\br{\cT_2 + \cT_3 + \cT_{4,2}} }{1- \eps_1 -\eps_2},\frac{\log\br{1\slash \delta}}{\eps^2_2} \right\}.
	}
 This completes the proof.
	\end{proof}

\subsection{Bounds on $\cT_2,\cT_3,\cT_4$ used in proof of Lemma~\ref{lemma:prediction_square_bound}}\label{sec:T_bounds}
We will upper-bound the process $\{q_t\}$ on the set
\al{ 
\cG_{q}:= \{\omega: \mbox{ (1, 2, 3) below hold}\}, \mbox{ where}
}
\al{
& 1)
\sum_{s=1}^{t}\gamma_s b_1 \br{\lm_{s-1} - \lm}'\psi_s  w_s 
\le \left(\sum_{s=1}^{t}\left\{\gamma_s b_1 \br{\lm_{s-1} - \lm}'\psi_s \right\}^2\right)^{1\slash 2}\notag\\
& \times  \log\left[\frac{1}{\delta}\left(\sum_{s=1}^{t}\left\{ \gamma_s b_1 \br{\lm_{s-1} - \lm}'\psi_s \right\}^2\right)\right] \forall t  \notag\\
& 2) \sum_{s=\tau+1}^{t} \br{\psi'_s P_s \psi_s}~\left\{w^2_s - \bE\br{w^2_s\Big| \cF_{s-1}}\right\}  \le 
\sqrt{\sum_{s=1}^{t} \br{\psi'_s P_s \psi_s}^2 \log \br{\frac{\sum_{s=1}^{t} \br{\psi'_s P_s \psi_s}^2}{\delta}} },~\forall t \notag
}	
\al{
3) & 
\sum_{s=1}^{t}\br{ \psi'_s P_s \psi_s } \br{\gamma_s}\left[b_1 \br{1-\gamma_s}\br{u_s - \lm' \psi_s } - b_1 \br{\lm_{s-1} - \lm  }'\psi_s	\right] w_s  \notag\\
&  \le \sqrt{ \sum_{s=1}^{t} \br{ \psi'_s P_s \psi_s }^2 \br{\gamma_s}^2\left[b_1 \br{1-\gamma_s}\br{u_s - \lm' \psi_s } - b_1 \br{\lm_{s-1} - \lm  }'\psi_s	\right]^2 }  \notag\\
& \times \sqrt{\log \br{\frac{\sum_{s=1}^{t} \br{ \psi'_s P_s \psi_s }^2 \br{\gamma_s}^2\left[b_1 \br{1-\gamma_s}\br{u_s - \lm' \psi_s } - b_1 \br{\lm_{s-1} - \lm  }'\psi_s	\right]^2}{\delta} } } 
\Bigg\}~\forall t.\label{def:g_q}
}
\begin{proposition}
	\al{
\bP\br{\cG_q \ge 1-3\delta}.	
}
\end{proposition}
\begin{proof}
	We will show that the probability with which any of the above conditions 1)-3) is violated can be bounded by $\delta$. For 1), we have $\sum_{s=\tau+1}^{t}\gamma_s b_1 \br{\lm_{s-1} - \lm}'\psi_s  w_s = \sum_{s=\tau+1}^{t}\frac{b^2_1}{b_{1,s}} \br{\lm_{s-1} - \lm}'\psi_s  w_s $. It then follows from the self-normalized inequality~\eqref{eq:self_normalized}, with $\eta_s$ set equal to $w_s$ and $X_s$ equal to $\frac{b^2_1}{b_{1,s}} \br{\lm_{s-1} - \lm}'\psi_s$, that the probability of violation of 1) is upper-bounded by $\delta$. The probability of violation of 2) is again upper-bounded by $\delta$ by using self-normalized martingale concentration~\eqref{eq:self_normalized} since $\{w^2_s - \bE\br{w^2_s\Big| \cF_{s-1}}\}$ is a martingale difference sequence. Similarly, the probability of violating 3) can also be bounded using~\eqref{eq:self_normalized}.
\end{proof}

\begin{proposition}[Bounding $\cT_2$]\label{prop:t2}
On the set $\cG_{LSE}\cap \cG_{\pj}\cap \cG_{\cI}$~\eqref{def:glse}, $\cT_2$~\eqref{def:t2} can be bounded as follows,
	\al{
\cT_2 \le 	\eps^2_3 \sum_{t\ge t\ust_6} r_t.
}
\end{proposition}
\begin{proof}
The term involved in $\cT_2$ at time $t$ is equal to $\br{b_1 \br{1-\gamma_s} \br{u_s - \lm' \psi_s}}^2$. Since $\gamma_s = \frac{b_1}{b_{1,s}}$ for $t\ge t\ust_6(\eps_3,\delta)$, we have $\br{b_1 \br{1-\gamma_s} \br{u_s - \lm' \psi_s}}^2\le \eps^2\br{b_1  \br{u_s - \lm' \psi_s}}^2$.~The proof is completed by noting that $r_t = b^2_1\br{u_t - \lm' \psi_t}^2$. 
\end{proof}

\begin{proposition}[Bounding $\cT_3$]\label{prop:t3}
Under Assumption~\ref{assum:bounded_noise}, on $\cG_q \cap \cG_{LSE}\cap \cG_{\pj}\cap \cG_{\cI}$,
\al{
	\cT_3 
	 & \le \sqrt{b^2_1 B^2_u\br{1+\eps_3}^2 \log \left(B_u\log T\right)\log \br{\log T}} \notag\\
	&\times \sqrt{ \log \left[  \frac{b^2_1 B^2_u\br{1+\eps_3}^2\log \left(B_u\log T\right)\log \br{\log T}}{\delta} \right]  }.
}
Same conclusion holds under Assumption~\ref{assum:sub_gaussian} on $\cG_q \cap \cG_{LSE}\cap \cG_{\pj}\cap \cG_w$.
\end{proposition}
\begin{proof}

On $\cG_q$,
	\al{
		& \cT_3 
		& \le \left(\sum_{s=1}^{t}\left\{\gamma_s b_1 \br{\lm_{s-1} - \lm}'\psi_s \right\}^2\right)^{1\slash 2} \log\left[\frac{1}{\delta}\left(\sum_{s=1}^{t}\left\{ \gamma_s b_1 \br{\lm_{s-1} - \lm}'\psi_s \right\}^2\right)\right].\label{ineq:45}
	}
Hence, to bound $\cT_3$, we will focus on bounding $\sum_{s=1}^{t}\left\{\gamma_s b_1 \br{\lm_{s-1} - \lm}'\psi_s \right\}^2$. We have,
\al{
	\sum_{s=1}^{t}\left\{\gamma_s b_1 \br{\lm_{s-1} - \lm}'\psi_s \right\}^2 & \le  b^2_1\sum_{s=1}^{t}\left\{\gamma^2_s  \frac{\log \left(B_u~N\ui_s\right)}{N\ui_s} \|\psi_s\|^{2} \right\}\notag \\
	& \le b^2_1 B^2_u\br{\sup_{s\ge t\ust} \gamma^2_s  }\sum_{s=1}^{t}\left\{\frac{\log \left(B_u~N\ui_s\right)}{N\ui_s} \right\} \notag\\
	&\le b^2_1 B^2_u\br{\sup_{s\ge t\ust} \gamma^2_s  }\sum_{s=1}^{t}\left\{\frac{\log \left(B_u~\log T\right)}{N\ui_s} \right\} \notag\\
	&\le b^2_1 B^2_u\br{\sup_{s\ge t\ust} \gamma^2_s  }\log \left(B_u\log T\right)\log \br{\log T},
}
where the first inequality follows from the bound on the estimation error derived in Theorem~\ref{th:estimation_error}, the second inequality follows from Lemma~\ref{lemma:bound_phi} since we have, 
\al{
\|\psi_s\|^2 &\le \|\phi_s\|^2 \\
& \le \br{C_1  \| Y_{0}\| +\frac{B_u C_1}{1-\rho}  \left\{1+  \sum_{\ell=1}^{q} |b_\ell|  \right\} }^2+ (B_uq)^2,
}
while the remaining follow since $N\ui_s$ is approximately equal to $\sqrt{\log s}$, and hence $N\ui_s \le \log s$. In summary,
\al{
	\cT_3 & \le \sqrt{b^2_1 B^2_u\br{\sup_{s\ge t\ust} \gamma^2_s  }\log \left(B_u\log T\right)\log \br{\log T}} \notag\\
	&\times \sqrt{ \log \left[  \frac{b^2_1 B^2_u\br{\sup_{s\ge t\ust} \gamma^2_s  }\log \left(B_u\log T\right)\log \br{\log T}}{\delta} \right]  }.\label{ineq:t3}
}
Proof is completed by noting that for $t\ge t\ust$, we have $\gamma_s \le 1+\eps_3$ (follows from the definition of $t\ust_6$).
\end{proof}

We will now derive an upper-bound on $\cT_4$. We have
\al{
	\cT_4	& = \sum_{s=t\ust}^{t}\br{ \psi'_s P_s \psi_s } \left[ b_1 \br{1-\gamma_s}\br{u_s - \lm' \psi_s } - b_1 \br{\lm_{s-1} - \lm  }'\psi_s	-\gamma_s w_s\right]^2\notag\\
	& = \sum_{s=t\ust}^{t} \br{ \psi'_s P_s \psi_s }  \left[b_1 \br{1-\gamma_s}\br{u_s - \lm' \psi_s } - b_1 \br{\lm_{s-1} - \lm  }'\psi_s	\right]^2 \notag\\
	&+ \sum_{s=t\ust}^{t} \br{ \psi'_s P_s \psi_s } \br{\gamma_s w_s}^2 \notag \\
	& - 2 \sum_{s=t\ust}^{t}\br{ \psi'_s P_s \psi_s } \br{\gamma_s w_s}\left[b_1 \br{1-\gamma_s}\br{u_s - \lm' \psi_s } - b_1 \br{\lm_{s-1} - \lm  }'\psi_s\right].  \label{ineq:69}
}
The term $\cT_4$ is therefore composed of the three terms $\cT_{4,1},\cT_{4,2},\cT_{4,3}$ which are bounded separately below.
\begin{proposition}[$\cT_{4,2}$]\label{prop:t_4_2}
	On $\cG_q$, 
		\al{
		\frac{\cT_{4,2}}{\br{1+\eps_3}^2} & \le  \sigma^2(p+q-1) \log T + \sigma^2 (p+q-1)\log \br{C^2_1 \br{ \| \psi_0 \| +\frac{B_u C_1}{1-\rho}  \left\{1+  \sum_{\ell=1}^{q} |b_\ell|  \right\} }^2} \notag\\
		& + \sqrt{(p+q-1)\log\left[T~C^2_1 \br{ \| \psi_0 \| +\frac{B_u C_1}{1-\rho}  \left\{1+  \sum_{\ell=1}^{q} |b_\ell|  \right\} }^2\right]} \notag\\
		& \times \sqrt{ \log 
			\br{\frac{(p+q-1)\log\left[ T~C^2_1 \br{ \| \psi_0 \| +\frac{B_u C_1}{1-\rho}  \left\{1+  \sum_{\ell=1}^{q} |b_\ell|  \right\} }^2\right]}{\delta}} }.\label{ineq:t_42_bound}
	}
\end{proposition}

\begin{proof}
~Consider 
	\al{
		\cT_{4,2} = \sum_{s=t\ust}^{t} \br{ \psi'_s P_s \psi_s } \br{\gamma_s w_s}^2 & \le \br{\sup_{s\ge t\ust} \gamma_s }^2  \sum_{s=1}^{t} \br{\psi'_s P_s \psi_s}~w^2_s \le \br{1+\eps_3}^2  \sum_{s=1}^{t} \br{\psi'_s P_s \psi_s}~w^2_s,\label{ineq:46}
	}
 where the last inequality follows from the definition of $t\ust_6$.~Now,
	\nal{
		\sum_{s=t\ust}^{t} \br{\psi'_s P_s \psi_s}~w^2_s = \sum_{s=t\ust}^{t} \br{\psi'_s P_s \psi_s}~\bE\br{w^2_s\Big| \cF_{s-1}} + \left[\sum_{s=\tau+1}^{t} \br{\psi'_s P_s \psi_s}~\left\{w^2_s - \bE\br{w^2_s\Big| \cF_{s-1}}\right\} \right].
	}
	The first summation above is bounded as follows,
	\nal{
		\sum_{s=t\ust}^{t} \br{\psi'_s P_s \psi_s}~\bE\br{w^2_s\Big| \cF_{s-1}} \le \sigma^2 \sum_{s=\tau+1}^{t} \br{\psi'_s P_s \psi_s}\le \sigma^2 \log\left[ \det\br{\sum_{s=1}^{t} \psi_s  \psi'_s}\right].
	}
	For second summation,
	\al{
		\sum_{s=\tau+1}^{t} \br{\psi'_s P_s \psi_s}~\left\{w^2_s - \bE\br{w^2_s\Big| \cF_{s-1}}\right\} & \le 
		\sqrt{\sum_{s=1}^{t} \br{\psi'_s P_s \psi_s}^2 \log \br{\frac{\sum_{s=1}^{t} \br{\psi'_s P_s \psi_s}^2}{\delta}} }\notag\\
		&\le 
		\sqrt{\sum_{s=1}^{t} \br{\psi'_s P_s \psi_s} \log \br{\frac{\sum_{s=1}^{t} \br{\psi'_s P_s \psi_s}}{\delta}} }\notag\\
		& \le 
		\sqrt{\log\left[ \det\br{\sum_{s=1}^{t} \psi_s  \psi'_s}\right] \log 
			\br{\frac{\log\left[ \det\br{\sum_{s=1}^{t} \psi_s \psi'_s}\right]}{\delta}} },
	}
	where the first inequality follows from the definition of $\cG_q$, the second inequality follows since $\psi'_s P_s \psi_s \le 1$, while the third inequality follows since~$\sum_{s=1}^{t} \psi'_s P_s \psi_s\le \log\left[ \det\br{\sum_{s=1}^{t} \psi_s \psi'_s}\right]$.
	
	In summary, we obtain the following,
	\al{
		\frac{\cT_{4,2}}{\br{1+\eps_3}^2 } \le  \sigma^2 \log\left[ \det\br{\sum_{s=1}^{t} \psi_s \psi'_s}\right]+ \sqrt{\log\left[ \det\br{\sum_{s=1}^{t} \psi_s  \psi'_s}\right] \log 
			\br{\frac{\log\left[ \det\br{\sum_{s=1}^{t} \psi_s  \psi'_s}\right]}{\delta}} }.\label{ineq:78}
	}
We now bound $\log\left[ \det\br{\sum_{s=1}^{t} \psi_s \psi'_s}\right]$.	For a matrix $M \in \bR^{(p+q-1)\times (p+q-1)}$, we have $\log \det(M)$ is the sum of logarithm of its eigenvalues, and hence can be upper-bounded by $(p+q-1)$ times the log of the maximum eigenvalue. The eigenvalues of the matrix $\sum_{s=1}^{t} \psi_s  \psi'_s$ are bounded by $\sum_{s=1}^{t} \|\psi_s\|^2$, which can be bounded using Lemma~\ref{lemma:bound_phi} as,
\al{
\sum_{s=1}^{t} \|\psi_s\|^2\le t \left\{\br{ \| Y_0 \| +\frac{B_u C_1}{1-\rho}  \left\{1+  \sum_{\ell=1}^{q} |b_\ell|  \right\} }^2 +  (B_u q )^2\right\} .
}
~Substituting this into~\eqref{ineq:78}, we obtain the following bound on $\cT_{4,2}$:
	\al{
		\frac{\cT_{4,2}}{\br{1+\eps_3}^2}& \le  \sigma^2(p+q-1) \log T + \sigma^2 (p+q-1)\log \br{C^2_1 \br{ \| \psi_0 \| +\frac{B_u C_1}{1-\rho}  \left\{1+  \sum_{\ell=1}^{q} |b_\ell|  \right\} }^2} \notag\\
		& + \sqrt{(p+q-1)\log\left[T~C^2_1 \br{ \| \psi_0 \| +\frac{B_u C_1}{1-\rho}  \left\{1+  \sum_{\ell=1}^{q} |b_\ell|  \right\} }^2\right]} \notag\\
		& \times \sqrt{ \log 
			\br{\frac{(p+q-1)\log\left[ T~C^2_1 \br{ \| \psi_0 \| +\frac{B_u C_1}{1-\rho}  \left\{1+  \sum_{\ell=1}^{q} |b_\ell|  \right\} }^2\right]}{\delta}} }.
	}
\end{proof}

We now bound the term $\cT_{4,3}$.
\begin{proposition}[$\cT_{4,3}$]
On $\cG_{q}$,
\al{
\cT_{4,3} &\le \br{1+\eps_3} \sqrt{ \sum_{s=1}^{t} \br{b_1 \br{1-\gamma_s}\br{u_s - \lm' \psi_s } - b_1 \br{\lm_{s-1} - \lm  }'\psi_s	}^2 }\notag \\
& \times \sqrt{\log \br{\frac{\br{1+\eps_3}^2}{\delta}}\br{\sum_{s=1}^{t}  \br{b_1 \br{1-\gamma_s}\br{u_s - \lm' \psi_s } - b_1 \br{\lm_{s-1} - \lm  }'\psi_s	}^2 } }.
}

\end{proposition}
\begin{proof}
We have,
\al{
	\cT_{4,3}& =\sum_{s=t\ust}^{T}\br{ \psi'_s P_s \psi_s } \br{\gamma_s}\left[b_1 \br{1-\gamma_s}\br{u_s - \lm' \psi_s } - b_1 \br{\lm_{s-1} - \lm  }'\psi_s	\right] w_s \notag \\
	& \le \sqrt{ \sum_{s=t\ust}^{T} \br{ \psi'_s P_s \psi_s }^2 \br{\gamma_s}^2\left[b_1 \br{1-\gamma_s}\br{u_s - \lm' \psi_s } - b_1 \br{\lm_{s-1} - \lm  }'\psi_s	\right]^2 }\notag \\
	& \times \sqrt{\log \br{\frac{\sum_{s=1}^{t} \br{ \psi'_s P_s \psi_s }^2 \br{\gamma_s}^2\left[b_1 \br{1-\gamma_s}\br{u_s - \lm' \psi_s } - b_1 \br{\lm_{s-1} - \lm  }'\psi_s	\right]^2}{\delta} } }\notag\\
	&\le \br{\sup_{s\ge t\ust} \Big| \gamma_s \Big|} \sqrt{ \sum_{s=1}^{t} \br{b_1 \br{1-\gamma_s}\br{u_s - \lm' \psi_s } - b_1 \br{\lm_{s-1} - \lm  }'\psi_s	}^2 }\notag \\
	& \times \sqrt{\log \left[ \br{\frac{\sup_{s\ge t\ust}  \gamma^2_s}{\delta}}\br{\sum_{s=1}^{t}  \br{b_1 \br{1-\gamma_s}\br{u_s - \lm' \psi_s } - b_1 \br{\lm_{s-1} - \lm  }'\psi_s	}^2 } \right]},\label{ineq:47}
}
where the first inequality follows from the definition of $\cG_{q}$, while the last inequality follows since $\psi'_s P_s \psi_s \le 1$. Proof is completed by noting that $\sup_{s\ge t\ust} \gamma_s \le 1+\eps_3$, from definition of $t\ust_6$.
\end{proof}

We now derive a bound on $q_{t\ust}$.
\begin{lemma}\label{lemma:bound_qstar}
On $\cG_q \cap \cG_{LSE}\cap \cG_{\pj}\cap \cG_{\cI}$ we have,
    \al{
q_{t\ust} \le b^2_1 \eps^2_3 t\ust \br{ \left\{C_1 \| Y_{0}\| +\frac{B_w C_1}{1-\rho}  \left\{1+  \sum_{\ell=1}^{q} |b_\ell|  \right\} \right\}^2 + B^2_w},
    }
    where $q_t$ is as in~\eqref{def:qt}, and $t\ust$ as in~\eqref{def:t_star}.
\end{lemma}
\begin{proof}
We have $q_t = \Tr\br{b^2_1\br{\lm_t - \lm}P^{-1}_t\br{\lm_t -\lm}'}$. Since the trace of a matrix is equal to the sum of its eigenvalues, by using the bound on estimation error derived in Theorem~\eqref{th:estimation_error}, we get,
\al{
q_{t\ust} &\le \cE(N\ui_{t\ust};\te\ust,\delta)^2 \sum_{s\le t} \|\phi_s\|^2 \notag\\
&\le cE(N\ui_{t\ust};\te\ust,\delta)^2 t\ust \br{ \left\{C_1 \| Y_{0}\| +\frac{B_w C_1}{1-\rho}  \left\{1+  \sum_{\ell=1}^{q} |b_\ell|  \right\} \right\}^2 + B^2_w} \\
& \le b^2_1 \eps^2_3 t\ust \br{ \left\{C_1 \| Y_{0}\| +\frac{B_w C_1}{1-\rho}  \left\{1+  \sum_{\ell=1}^{q} |b_\ell|  \right\} \right\}^2 + B^2_w},
}
where the second inequality follows from the definition of $t\ust_6$, and Lemma~\ref{lemma:state_bound}.
\end{proof}
\subsection{Bounding the cumulative Regret}
In order to prove Theorem~\ref{th:estimation_error_main}, we will show the following stronger result. The ``good sets'' $\cG_q,\cG_{LSE},\cG_{\pj},\cG_{\cI},\cG_w,\cG_{w^2_{B}}$ are defined in~\eqref{def:g_q}~\eqref{def:glse},~\eqref{ineq:52},
~\eqref{def:gi},~\eqref{def:g1},~\eqref{def:G_tilde4} respectively and our analysis is performed on the intersection of these.

\begin{theorem}\label{th:regret}
Consider the ARX system~\eqref{def:arx} in which $\{w_t\}$ satisfies Assumptions~(\ref{assum:noise_var},\ref{assum:bounded_noise}).~On the set $\cG_q \cap \cG_{LSE}\cap \cG_{\pj}\cap \cG_{\cI}\cap \cG_{w^2_{B}}$, for $\eps_2 \in (0,1-\eps_1)$, the cumulative regret of the PIECE algorithm until $T$ is bounded as follows:
\al{
\cR_T \le \br{1-\frac{\eps^2_3}{\br{1-\eps_3}^2\br{1-\eps_1-\eps_2}  }}^{-1} \left[\frac{\sigma^2(p+q-1) \log T  }{\br{1-\eps_3}^2\br{1-\eps_1-\eps_2} }\right]
+ C_1,\label{res:mainres}
}
where $C_1(\eps,\eps_1,\eps_2,\eps_3,\delta,\rho)$ is sum of: (i) $\log\log T$ times a polynomial function of $\sqrt{\log\frac{1}{\delta}}, (1+\epsilon)\sqrt{\log(1+\epsilon)^2},\frac{1}{1-\rho}$, (ii) $\lesssim \log(1/\delta)$, $\lesssim \frac{1}{\epsilon^2(1-\epsilon^2)}$, and additional terms that are $o(\log\log T)$, and where $\eps_1,\eps_3$ are as in~\eqref{def:t5_star},~\eqref{def:t_6_star}.


\end{theorem}	
\begin{proof}
	Since the cumulative regret during $\cI$ has already been bounded in Lemma~\ref{lemma:regret_I}, we begin by deriving a bound on the cumulative regret during $t\notin \cI$. From Theorem~\ref{th:reg_n_i}, for $t\not\in \cI$ and greater than $\max\{t\ust_1,t\ust_2,t\ust_3\}$, $r_t$ can be bounded by $\br{w_t - e_t}^2$. We decompose the regret into the following two parts,
\al{
	\sum_{t\notin \cI} r_t & = \sum_{t\notin \cI, t\le t\ust} r_t + \sum_{t\notin \cI, t\ge t\ust} r_t
 }
where $t\ust$ is as in~\eqref{def:t_star}. The first summation is bounded as follows,
\al{
\sum_{t\notin \cI, t\le t\ust} r_t  & = \sum_{t\notin \cI, t\le t\ust} |y_t - w_t|^2	\notag\\
&\le  t\ust \sup_t |y_t - w_t|^2	\notag\\
&\le 2t\ust \br{\br{C_1  \| Y_{0}\| +\frac{B_u C_1}{1-\rho}  \left\{1+  \sum_{\ell=1}^{q} |b_\ell|  \right\} }^2 +    B^2_w   }.\label{ineq:adhoc_1} 
}

The second summation is bounded as follows using Lemma~\ref{lemma:prediction_square_bound},
\al{
\sum_{t\notin \cI, t\ge t\ust} \br{w_t - e_t}^2	 &\le \sum_{t} \br{w_t - e_t}^2 \notag\\
	&\le  \frac{1}{\br{1-\eps^2_3}}\max\left\{\frac{\br{\cT_2 + \cT_3 + \cT_{4,2}} }{1- \eps_1 -\eps_2},\frac{\log\br{1\slash \delta}}{\eps^2_2} \right\}+ q_{t\ust}\notag\\
	&\le  \frac{1}{\br{1-\eps^2_3}}\left\{\frac{\br{\cT_2 + \cT_3 + \cT_{4,2}} }{1- \eps_1 -\eps_2} + \frac{\log\br{1\slash \delta}}{\eps^2_2}\right\}+q_{t\ust}.	\label{ineq:81}
}

From Lemma~\ref{lemma:regret_I} the regret during $\cI$ is bounded as 
\al{
	\sum_{t\in\cI} r_t \le 2~b^2_1 \br{B^2_u+ \|\lm\|^2\left[C_1\| Y_{0}\| +\frac{B_u C_1}{1-\rho}  \left\{1+  \sum_{\ell=1}^{q} |b_\ell|  \right\} + B_u q\right]^2 } N\ui_T.\label{ineq:80}
}
\eqref{ineq:adhoc_1}-~\eqref{ineq:80} and Lemma~\ref{lemma:bound_qstar} yield,
\al{
\sum_{t} r_t \le &\frac{\br{\cT_2 + \cT_3 + \cT_{4,2}} }{\br{1-\eps_3}^2\br{1-\eps_1-\eps_2}  }+b^2_1 \eps^2_3 t\ust \br{ \left\{C_1 \| Y_{0}\| +\frac{B_w C_1}{1-\rho}  \left\{1+  \sum_{\ell=1}^{q} |b_\ell|  \right\} \right\}^2 + B^2_w}\notag\\ 
&+ 2t\ust \br{\br{C_1  \| Y_{0}\| +\frac{B_u C_1}{1-\rho}  \left\{1+  \sum_{\ell=1}^{q} |b_\ell|  \right\} }^2 +    B^2_w   }\notag\\
& + \frac{1}{\br{1-\eps^2_3}} \cdot \frac{\log\br{1\slash \delta}}{\eps^2_2}\notag\\
& + 2~b^2_1 \br{B^2_u+ \|\lm\|^2\left[C_1\| Y_{0}\| +\frac{B_u C_1}{1-\rho}  \left\{1+  \sum_{\ell=1}^{q} |b_\ell|  \right\} + B_u q\right]^2 } N\ui_T.\label{ineq:adhoc_2}
}
Denote the second, third, fourth and fifth summation above by $\cS_5$.~We now analyze $\cT_2$ in the bound above.~The term involved in $\cT_2$ at time $t$ is equal to $\br{b_1 \br{1-\gamma_s} \br{u_s - \lm' \psi_s}}^2$, while $r_t = b^2_1\br{u_t - \lm' \psi_t}^2$. Now since $\gamma_s = \frac{b_1}{b_{1,s}}$, for $t\ge t\ust_6(\eps_3,\delta)$~\eqref{def:t_6_star} we have $\br{b_1 \br{1-\gamma_t} \br{u_t - \lm' \psi_t}}^2\le \eps^2_3 r_t$. Thus, $\cT_2 \le  \eps^2_3 \sum_t r_t$, which upon substituting into~\eqref{ineq:adhoc_2} yields,
\al{
& \br{1-\frac{\eps^2_3}{\br{1-\eps_3}^2\br{1-\eps_1-\eps_2}  }}\sum_{t} r_t \le \frac{\br{\cT_3 + \cT_{4,2}} }{\br{1-\eps_3}^2\br{1-\eps_1-\eps_2}  } + \cS_5,
}
or
\al{
	\sum_{t} r_t \le \br{1-\frac{\eps^2_3}{\br{1-\eps_3}^2\br{1-\eps_1-\eps_2}  }}^{-1}\left[ \frac{\br{\cT_3 + \cT_{4,2}} }{\br{1-\eps_3}^2\br{1-\eps_1-\eps_2}  } + \cS_5 \right],
}
The proof is then completed by substituting the bounds on $\cT_3,\cT_{4,2}$ derived in Propositions~\ref{prop:t3} and~\ref{prop:t_4_2}. 
\end{proof}

\subsection{Some Auxiliary Results}
\begin{proposition}\label{prop:10}
	\al{
		b_1 \br{1-\gamma_s}\br{u_s - \lm' \psi_s } - b_1 \br{\lm_{s-1} - \lm  }'\psi_s  =  \frac{b_1}{b_{1,t}}\br{w_t - e_t}.
	}
\end{proposition}
\begin{proof}  We have,
\al{
	y_t = b_1 \br{u_{t-1} - \lm' \psi_{t-1}} + w_t, \label{eq:21}
}
and also,
\al{
	e_t = y_t - b_{1,t} \br{u_{t-1} - \lm'_{t-1}\psi_{t-1}}.
}
Hence,
\al{
	-\frac{e_t}{b_{1,t}} & = u_{t-1} -\frac{y_t}{b_{1,t}} - \lm'_{t-1}\psi_{t-1} \notag\\
	& = u_{t-1} -\frac{b_1 \br{u_{t-1} - \lm' \psi_{t-1}} + w_t}{b_{1,t}} - \lm'_{t-1}\psi_{t-1} \notag\\
	&=\br{1-\gamma_t} \br{u_{t-1} - \lm' \psi_{t-1}} - \br{u_{t-1} - \lm' \psi_{t-1}} + u_{t-1} - \lm'_{t-1}\psi_{t-1} - \frac{w_t}{b_{1,t}}\notag\\
	&=\br{1-\gamma_t} \br{u_{t-1} - \lm' \psi_{t-1}}  -\br{ \lm_{t-1}-\lm}'\psi_{t-1} - \frac{w_t}{b_{1,t}},\notag\\
}
where the second equality follows from~\eqref{eq:21}. The proof is completed by re-arranging the terms.
\end{proof}


\section{Estimation Error}\label{sec:estimation_error}
PIECE uses multiple estimators while designing inputs $\{u_t\}$. Let $\te\ui_t= \br{ a\ui_{1,t},a\ui_{2,t},\ldots,a\ui_{p,t},b\ui_{1,t},b\ui_{2,t},\ldots,b\ui_{q,t} }$ be the least squares estimate of $\te\ust$ using only the samples collected during the exploratory instants $\cI$. A recursive estimate of $\lm$ using all the samples until $t$ is generated as in~\eqref{def:rec_lambda}, and denoted by $\lm_t$. Let $\te_t$ be the LS estimate of $\te\ust$ using all the samples until time $t$. 

The estimation error at time $t$ satisfies,
\al{
	\|\te\ust -\te\ui_t\|^{2} &= \norm{V^{-1}_t\br{\sum_{s\le t, s\in \cI} \phi(s) w(s)}}^2 \notag\\
	& \le \| V^{-1\slash 2}_t \|^2 \norm{V^{-1\slash 2}_t\br{\sum_{s\le t, s\in \cI} \phi_s w_s}}^2\notag\\
	& = \| V^{-1\slash 2}_t  \|^2   \br{\sum_{s\le t, s\in \cI} \phi_s w_s}' 
 \left[ \br{V\ui_t}^{-1\slash 2} \right]'    
 \br{V\ui_t}^{-1\slash 2} \br{\sum_{s\le t, s\in \cI} \phi_s w_s} \notag\\
	& = \| \br{V\ui_t}^{-1\slash 2} \|^2   \br{\sum_{s\le t, s\in \cI} \phi_s w_s}' \br{V\ui_t}^{-1}  \br{\sum_{s\le t, s\in \cI} \phi_s w_s} \notag\\
	&\le \frac{1}{\lm_{\min}(V\ui_t)} \br{\sum_{s\le t, s\in \cI} \phi_s w_s}'\br{V\ui_t}^{-1}\br{\sum_{s\le t, s\in \cI} \phi_s w_s},\label{ineq:87}
}
where $V\ui_t : = \sum_{s\le t, s\in \cI} \phi_s \phi'_s$.
\begin{definition}
	Define,
	\begin{align}
		&\cG_{LSE} := \notag\\
	&\left\{\omega: \br{\sum_{s\le t, s\in \cI} \phi_s w_s}' 
		\br{\sum_{s\le t, s\in \cI} \phi_s \phi'_s }^{-1} \br{\sum_{s\le t, s\in \cI} \phi_s w_s} \le \log \br{\frac{\det\br{V\ui_t}^{.5}}{\delta} },~\forall t \right\}.\label{def:glse}
	\end{align}
\end{definition}
\begin{lemma}
	\al{
\bP\br{\cG_{LSE}} \ge 1-\delta.	
}	
\end{lemma}
\begin{proof}
	Follows from~\eqref{eq:self_normalized} by letting $\eta_s = w_s$ and $X_s = \phi_s$.
\end{proof}
 
\begin{lemma}\label{lemma:est_error_1}
	 On $\cG_{LSE}$, 
	\al{
		\|\te\ust - \te_t \|^{2} & \le \frac{\log \br{\frac{\det(V_t)^{\frac{1}{2}}}{\delta} } }{\lm_{\min}(V_t)}\notag\\
		& \le \frac{\br{p+q}\log \lm_{\max}(V_t) - 2\log\br{\delta}}{2 \lm_{\min}(V_t)},~\forall t .\label{ineq:88}
	}
	Similarly,
	$$
	\|\te\ust - \te\ui_t\|^{2} \le \frac{\br{p+q}\log \br{\lm_{\max}\br{V\ui_t} } - 2 \log\br{\delta}}{2\lm_{\min}(V\ui_t)}.
	$$ 
\end{lemma}
\begin{proof}
Follows from the bound~\eqref{ineq:87} and the definition of $\cG_{LSE}$.
\end{proof}
It follows from~\eqref{ineq:88} that in order to bound the estimation error, we need to derive an upper-bound on $\lm_{\max}\br{V_t},\lm_{\max}\br{V\ui_t}$. This is done in the following result.
\begin{lemma}\label{lemma:phi_bounded}
	If $\{w_t\}$ satisfies Assumption~\ref{assum:bounded_noise}, then
\al{
		\lm_{\max}(V\ui_t) &\le \br{\br{C_1  \| Y_{0}\|  +  \frac{C_1 B_u}{1-\rho}\left\{1+  \sum_{\ell=1}^{q} |b_\ell|  \right\}}^2 + q B^2_u}N\ui_t,~\forall t.\label{ineq:18}	\\
  \lm_{\max}(V\ui_t) &\le \br{\br{C_1  \| Y_{0}\|  +  \frac{C_1 B_u}{1-\rho}\left\{1+  \sum_{\ell=1}^{q} |b_\ell|  \right\}}^2 + q B^2_u}t,~\forall t.
}	
\end{lemma}
\begin{proof} 
Since $\lm_{\max}(V\ui_t) \le \sum_{s\le t,s\in\cI} \|\phi_s\|^2$, we will derive an upper-bound on $\|\phi_s\|$:
\al{
	\|\phi_s\| \le C_1  \| Y_{0}\|  +  \frac{C_1 B_u}{1-\rho}\left\{1+  \sum_{\ell=1}^{q} |b_\ell|  \right\} + q B_u,~\forall s\in\cI.
}

Since $|u_s|\le B_u$ and $|w_s|\le B_w < B_u$, Lemma~\ref{lemma:state_bound}-(i) yields
\al{
	\| Y_{t} \| & \le 	C_1 \rho^{t} \| Y_{0}\| + C_1 B_u \sum_{s=0}^{t  -1} \rho^{s} \left\{1+  \sum_{\ell=1}^{q} |b_\ell|  \right\}\notag\\
	& \le C_1  \| Y_{0}\|  +  \frac{C_1 B_u}{1-\rho}\left\{1+  \sum_{\ell=1}^{q} |b_\ell|  \right\}.
}
Since $\phi_{t} := \br{y_{t},y_{t-1},\ldots,y_{t-p}, u_{t},u_{t-1},\ldots,u_{t-q+1}   }'$, we have,
\al{
\| \phi_t \|^2 & = \| Y_{t} \|^2 + \|U_t\|^2 \notag\\
&\le \br{C_1  \| Y_{0}\|  +  \frac{C_1 B_u}{1-\rho}\left\{1+  \sum_{\ell=1}^{q} |b_\ell|  \right\}}^2 + q B^2_u, 
} 
where $U_t := \br{u_t,u_{t-1},\ldots,u_{t-q+2}}'$. The proof is completed by noting that $\lm_{\max}(V\ui_t) \le \sum_{s\le t,s\in\cI} \|\phi_s\|^2$.
\end{proof}

In order to obtain a bound on the estimation error, a lower bound on $\lm_{\min}(V\ui_t)$ under the proposed learning algorithm is derived in Section~\ref{sec:lmin}. This is then used to prove the main result on estimation error below.
Define the function,
\al{
	& \cE(x;\te,\delta) := \frac{\br{p+q} \log \br{x}   }{2\beta_3 x} \notag\\
	&+ \frac{\br{p+q}\log\left[\br{C_1(\te)  \| Y_{0}\|  +  \frac{C_1(\te) B_u}{1-\rho(\te)}\left\{1+  \sum_{\ell=1}^{q} |b_\ell(\te)|  \right\} }^2 
		+ qB^2_u \right]	- 2\log\br{\delta}}{ \beta_3 x},\label{def:error_fn} 
}
where $\beta_3$ is as in~\eqref{def:beta_3}, and $\sigma^2_e$ is as in~\eqref{def:sigma_e}. The sets $\cG_q,\cG_{LSE},\cG_{\pj},\cG_{\cI},\cG_w,\cG_{w^2_{B}}$ are defined in~\eqref{def:g_q}~\eqref{def:glse},~\eqref{ineq:52},
~\eqref{def:gi},~\eqref{def:g1},~\eqref{def:G_tilde4} respectively.
\begin{theorem}\label{th:estimation_error}
	If $\{w_t\}$ satisfies Assumption~\ref{assum:bounded_noise}, then on the set $\cG_{LSE}\cap \cG_{\pj}\cap \cG_{\cI}\cap \cG_{w^2_B}$, the estimation error can be bounded as follows,
	\al{
		\|\te\ust - \te\ui_t\|^{2} \le \cE(N\ui_t;\te\ust,\delta),
	}	
where the function $\cE\br{\cdot,;\delta,\te}$ is as in~\eqref{def:error_fn}. If instead $\{w_t\}$ satisfies Assumption~\ref{assum:sub_gaussian}, then the same conclusion holds on the set $\cG_{LSE}\cap \cG_{\pj}\cap \cG_{\cI}\cap \cG_{w^2_{UB}}$ 
\end{theorem}
\begin{proof}
	Follows from Lemma~\ref{lemma:est_error_1}, Lemma~\ref{lemma:phi_bounded} and Theorem~\ref{th:1}.
\end{proof}

\section{Relation between prediction error and regret for $t\notin \cI$}\label{sec:predict_error}
\begin{definition}\label{def:t_1}
	Let
 \al{
		t\ust_1(\rho) := \inf \left\{ t\in \bN: B_2 \frac{\log N\ui_t}{\sqrt{N\ui_t}} + \cE(t;\te\ust,\delta) < 1, C_1 \rho^{t} \| Y_{0}\| < B_u \right\},\label{def:t1_star}
	}	 
	
	Define,
	\al{
		&		t\ust_2(\rho) := \inf\Big\{t\in\bN: \notag\\
		&						\left[ B_2 \frac{\log N^{(\cI)}_t}{\sqrt{N^{(\cI)}_t}} + \cE(t;\te\ust,\delta)\right] \cdot \left[b_1\br{p+ 1 + \frac{C_1}{1-\rho}  \left\{1  +  \sum_{\ell=1}^{q} |b_\ell|   \right\} }\right] \le \frac{\delta^2_1}{2},\notag\\
		&C_1 \rho^{t} \| Y_{0}\| < \delta_1 B_u\Big\}.\label{def:t2_star}
	}
 \al{
t\ust_3(\delta) := \inf \left\{t \in\bN: \cE(t;\te\ust,\delta) \le \frac{b_1}{2},\mbox{ and }\frac{B_2 \sqrt{\log N\ui_{\ell}}}{2}\ge 2\right\}.\label{def:t3_star}
}
\end{definition}

\begin{theorem}\label{th:reg_n_i}
	Under Assumption~\ref{assum:bounded_noise} on $\{w_t\}$, the following holds
 on the set $\cG_q\cap \cG_{LSE}$: When $t\ge \max\left\{t\ust_1(\rho) , t\ust_2(\rho),t\ust_3(\delta)\right\}$, the instantaneous regret $r_t$ can be bounded by $\br{e_t -w_t}^2$.
\end{theorem}
\begin{proof} From~\ref{eq:inst_reg}, the instantaneous regret is given by 
\al{
	b^2_1\br{ u_{t-1} - \lm'\psi_{t-1}  }^2. \label{eq:16}
}
From Lemma~\ref{lemma:1} we have $u_t = z_t$. We consider the following two cases, and show that in both of them the expression~\eqref{eq:16} can be bounded by $\br{e_t - w_{t}}^2$. Recall that from~\eqref{def:error}, we have
\al{
	e_{t+1} = b_1 \br{u_t - \lm' \psi_t  }-\tilde{b}\ui_{1,t-1}\br{u_{t} - \lm'_{t-1 } \psi_{t} } + w_{t+1},
}
so that 
\al{
	e_{t+1} - w_{t+1} = b_1 \br{u_t - \lm' \psi_t  }-\tilde{b}\ui_{1,t-1}\br{u_{t} - \lm'_{t-1 } \psi_{t} }.\label{eq:17}
}

\emph{Case-I}: $u_t = \lm'_{t-1}\psi_t$. 

In this case,
\al{
	e_{t+1} - w_{t+1} & = b_1 \br{u_t - \lm' \psi_t  }-\tilde{b}\ui_{1,t-1}\br{u_{t} - \lm'_{t-1 } \psi_{t} } \notag\\
	& = b_1 \br{u_t - \lm' \psi_t  }.
}
Thus, in this case, $\br{e_{t+1} - w_{t+1}}^2$ is equal to the instantaneous regret $ b^2_1 \br{u_t - \lm' \psi_t  }^2$.

\emph{Case-II}: $u_t = \br{\tilde{\lm}^{(\cI)}_{t-1}}' \psi_{t}$. 

We have,
\al{
	|u_t - \lm'_{t-1}\psi_t| & =  |\br{\tilde{\lm}^{(\cI)}_{t-1}}' \psi_{t} - \lm'_{t-1}\psi_t| \notag\\
	& \ge B_2  \frac{\log N^{(\cI)}_t}{\sqrt{N^{(\cI)}_t}}\|\psi_{t}\|,
}
where the inequality follows from~\eqref{def:_z_t}. 
Moreover, from guarantees on the estimation error provided in Theorem~\ref{th:estimation_error} on $\cG_{LSE}$, we have,
\al{
	\|\lm - \tilde{\lm}\ui_{t-1}\| \le \cE(N\ui_t;\te\ust,\delta).
}
Upon combining these two,
\al{
	|u_t - \lm'_{t-1}\psi_t| \ge B_2 \sqrt{\log N\ui_t}\|\lm - \tilde{\lm}\ui_{t-1}\| \|\psi_t\|.\label{ineq:71}
}

We have,
\al{
	|u_t - \lm' \psi_t| & = | \br{\tilde{\lm}^{(\cI)}_{t-1} }'\psi_{t} -   \lm' \psi_t| \notag \\
	& \le \|\tilde{\lm}^{(\cI)}_{t-1} -  \lm'  \| \| \psi_t \| \notag\\
	& \le \frac{|u_t - \lm'_{t-1}\psi_t|}{B_2 \sqrt{\log N\ui_t}}
	,\label{ineq:44}
}
where the first inequality follows from~\eqref{ineq:33}, while the second one follows from~\eqref{ineq:71}. From~\eqref{eq:17} we have, 
\al{
	|e_{t+1}-w_{t+1}|  \ge & \Big| \tilde{b}\ui_{1,t-1}\br{u_{t} - \lm'_{t-1 } \psi_{t} } \Big| - |b_1 \br{u_t - \lm' \psi_t  }|\notag\\
	\ge &   \Big| \tilde{b}\ui_{1,t-1}\br{u_t - \lm' \psi_t}B_2 \sqrt{\log N\ui_t} \Big|  - |b_1 \br{u_t - \lm' \psi_t  }|\notag\\
	\ge &   \Big| \frac{b_1}{2}\br{u_t - \lm' \psi_t}B_2 \sqrt{\log N\ui_t} \Big|  - |b_1 \br{u_t - \lm' \psi_t  }|\notag\\
	\ge &   \Big| b_1\br{u_t - \lm' \psi_t}B_2 \Big|,\notag 
}
\end{proof}

\subsection{Auxiliary Results}
\begin{proposition}\label{prop:6}
	Under Assumption~\ref{assum:bounded_noise}, we have, 
	\al{
		|z_t| \le  \left[ \|\lm\| + B_2 \frac{\log N^{(\cI)}_t}{\sqrt{N^{(\cI)}_t}} + \cE(N\ui_t;\te\ust,\delta)
\right]  \|\psi_t\|, \mbox{ on } \cG_{q} \cap \cG_{LSE}\cap \cG_{\pj}\cap \cG_{\cI}, 
	}
	where $z_t$ is as in~\eqref{def:_z_t}, and $B_2$ is as in~\eqref{def:control_rule}. If Assumption~\ref{assum:sub_gaussian} holds
 instead of Assumption~\ref{assum:bounded_noise}, then the same conclusion holds on $\cG_{q} \cap \cG_{LSE}\cap \cG_{\pj}\cap \cG_{\cI}\cap \cG_{w}$.
\end{proposition}
\begin{proof} The bound on the estimation error derived in Theorem~\ref{th:estimation_error} yields the following on $\cG_{LSE}$:
\al{
	\|\lm - \tilde{\lm}\ui_{t} \| \le \cE(N\ui_t;\te\ust,\delta).\label{ineq:33}
}
This means that
\al{
	\Big|\lm' \psi_t - \br{\tilde{\lm}\ui_{t}}' \psi_t \Big| \le  \cE(N\ui_t;\te\ust,\delta) \|\psi_t\|.\label{ineq:16}
}
Moreover, 
\al{
	\Big| \lm'_{t-1}\psi_t - \lm'\psi_t \Big| & \le \Big| \lm'_{t-1}\psi_t - \br{\tilde{\lm}\ui_{t}}' \psi_t  \Big| + |\lm'\psi_t - \br{\tilde{\lm}\ui_{t}}' \psi_t | \notag\\
	& \le  \Big| \lm'_{t-1}\psi_t - \br{\tilde{\lm}\ui_{t}}' \psi_t  \Big| +  \cE(N\ui_t;\te\ust,\delta)\|\psi_t\|.\label{ineq:17}
}
It then follows from the definition of $z_t$ in~\eqref{def:_z_t}, and the bounds~\eqref{ineq:16},~\eqref{ineq:17}, that
\al{
	\Big| z_t - \lm'\psi_{t}\Big| \le   \left[ B_2  \frac{\log N^{(\cI)}_t}{\sqrt{N^{(\cI)}_t}} +  \cE(N\ui_t;\te\ust,\delta)\right] \|\psi_{t}\|,\label{ineq:19}
}
or
\al{
	| z_t | \le   \left[ \|\lm\| + B_2\frac{\log N^{(\cI)}_t}{\sqrt{N^{(\cI)}_t}} + \cE(N\ui_t;\te\ust,\delta) \right] \|\psi_{t}\|.\label{ineq:20}
}
This completes the proof.
\end{proof}

\begin{proposition}\label{prop:7}
	Under Assumption~\ref{assum:bounded_noise}, on the set $\cG_{q} \cap \cG_{LSE}\cap \cG_{\pj}\cap \cG_{\cI}$, we have the following: For $t\ge t\ust_1$ and ,~$t\notin \cI$, on the event 	
	\al{
		\br{\|\lm\| + 1} \|\psi_t\|	\le B_u,
	}
	we have	
	\al{
		| y_{t+1} - w_{t+1} |  \le b_1\left[ B_2  \frac{\log N^{(\cI)}_t}{\sqrt{N^{(\cI)}_t}} + \cE(N\ui_t;\te\ust,\delta) \right]  \left[B_u \br{p+ 1 + \frac{C_1}{1-\rho}  \left\{1  +  \sum_{\ell=1}^{q} |b_\ell|   \right\} }\right],
	} 
	where $B_2$ is as in~\eqref{def:control_rule}. The same conclusion holds on $\cG_{q} \cap \cG_{LSE}\cap \cG_{\pj}\cap \cG_{\cI}\cap\cG_w$ if $\{w_t\}$ satisfies Assumption~\ref{assum:sub_gaussian} instead of Assumption~\ref{assum:bounded_noise}.
\end{proposition}
\begin{proof} From Lemma~\ref{lemma:bound_phi},
\al{
	\| Y_{t} \| \le 	C_1 \rho^{t} \| Y_{t_0}\| +\frac{B_u C_1}{1-\rho}  \left\{1+  \sum_{\ell=1}^{q} |b_\ell|  \right\},
}
where $Y_t := \br{y_t,y_{t-1},\ldots,y_{t-p+1}}'$. Let $U_t := \br{u_t,u_{t-1},\ldots,u_{t-q+2}}'$. ~Since $\|\psi_t\|\le \|U_t\| + \|Y_t\|$, and $ \|U_t\|  \le B_u q$, for $t\ge t\ust_1$ we have,
\al{
	\|\psi_t\| \le B_u  q+ B_u + \frac{B_u C_1}{1-\rho}  \left\{1  +  \sum_{\ell=1}^{q} |b_\ell|   \right\}.\label{ineq:24}
}
For $t\ge t\ust_1$,
\al{
	B_2 \frac{\log N^{(\cI)}_t}{\sqrt{N^{(\cI)}_t}} +  \cE(N\ui_t;\te\ust,\delta) < 1,	\label{ineq:56}
}
so that from Proposition~\ref{prop:6} we have, 
\nal{
	|z_t| \le \br{\|\lm\|+1}\|\psi_t\|.
}
Thus, when $\br{\|\lm\|+1}\|\psi_t\| < B_u$, we have $|z_t|<B_u$, so that from~\eqref{def:control_rule} we have $u_t = z_t$. This gives,
\al{
	| y_{t+1} - w_{t+1} | & = | b_1 \br{u_t - \lm' \psi_t} |\\
	& =  | b_1 \br{z_t - \lm' \psi_t} | \notag\\
	& \le b_1 \left[ B_2 \frac{\log N^{(\cI)}_t}{\sqrt{N^{(\cI)}_t}} +  \cE(N\ui_t;\te\ust,\delta) \right] \|\psi_{t}\| \notag\\
	& \le b_1  \left[ B_2  \frac{\log N^{(\cI)}_t}{\sqrt{N^{(\cI)}_t}} +  \cE(N\ui_t;\te\ust,\delta) \right] \left[B_u \br{p+ 1 + \frac{C_1}{1-\rho}  \left\{1  +  \sum_{\ell=1}^{q} |b_\ell|   \right\} }\right]\notag\\
} 
where the second equality follows since $u_t = z_t$, the first inequality follows from~\eqref{ineq:19}, and the second inequality from~\eqref{ineq:24}.
\end{proof}

\begin{proposition}\label{prop:8}
	Under Assumption~\ref{assum:bounded_noise}, we have the following on $\cG_{q} \cap \cG_{LSE}\cap \cG_{\pj}\cap \cG_{\cI}$: For $t \ge \max \{t\ust_1,t\ust_2\}$~\eqref{ineq:56},
	\al{
		|u_t| & \le \delta_1 B_u,
	} 	
	where $\delta_1>0$ satisfies~\eqref{cond:delta_1}-\eqref{cond:delta_4}.	 
\end{proposition}
\begin{proof} It follows from Proposition~\ref{prop:6} that for $t \ge \max \{t\ust_1,t\ust_2\}$,
\al{
	|z_t| \le  \br{\|\lm\| +  1}  \|\psi_t\|.	\label{ineq:25}
}
Moreover, from Proposition~\ref{prop:7} for $t \ge \max \{t\ust_1,t\ust_2\}$ with $t\notin \cI$, on the event 	
\al{
	\br{\|\lm\| + 1} \|\psi_t\|	\le B_u,\label{ineq:26}
}
we have	
\al{
	| y_{t+1} - w_{t+1} |  \le b_1 \left[ B_2  \frac{\log N^{(\cI)}_t}{\sqrt{N^{(\cI)}_t}} +  \cE(N\ui_t;\te\ust,\delta) \right] \left[B_u \br{p+ 1 + \frac{C_1}{1-\rho}  \left\{1  +  \sum_{\ell=1}^{q} |b_\ell|   \right\} }\right].\label{ineq:27}
} 
From definition of $t\ust_2$ we have that for $t\ge t\ust_2$, 
\al{
	| y_{t+1} - w_{t+1} |  \le \frac{\delta^2_1}{2}  B_u.\label{ineq:38}
} 
We have,
\al{
|w_t|&\le B_w \\
& \le \frac{\delta^2_1}{2} B_u,
}
where the second inequality follows from~\eqref{eq:24}.

Upon combining this with~\eqref{ineq:38} we obtain that when $(\|\lm\|+1)\psi_t \le B_u$, then we have
\al{
	| y_{t+1} | & \le B_w + \frac{\delta^2_1}{2}  B_u \notag\\	
	&\le \delta^2_1 ~B_u.\label{ineq:31}
} 

Now, choose a sufficiently large episode $i$, so that its start time satisfies $n_i \ge \max\{t\ust_1,t\ust_2\}$. We will now show that for times $t\in\left\{n_i + m\ust,\ldots, n_{i+1}  \right\}$, where $m_i \ge m\ust$\footnote{$m\ust$ is as in~\eqref{def:must}.}, the following holds:
\al{
	|u_t| & \le \delta_1  B_u,\label{ineq:28}\\
	\mbox{and }	|y_{t+1}| & \le  \delta^2_1  B_u.\label{ineq:29}
}

We consider the following two cases separately.

Case 1). $t\in\left\{n_i + m\ust,\ldots, n_{i}+m_i  \right\}$:

From~\eqref{ineq:2} we have the following bound,
\al{
	\| Y_{n_i + m} \| \le 	C_1 \rho^{m} \| Y_{n_i}\| + B_w  C_1 \sum_{s=0}^{m} \rho^{s} \left\{1+  \sum_{\ell=1}^{q} |b_\ell|  \right\},
}
where $m\le m_i$.~When $m\ge m\ust$, so that it satisfies~\eqref{ineq:30}, this bound yields, 
\al{
	\| Y_{n_i + m} \| & \le B_w \br{	1+ C_1 \sum_{s=0}^{m} \rho^{s} \left\{1+  \sum_{\ell=1}^{q} |b_\ell|  \right\}}\notag\\
	& = \frac{\delta^2_1}{1+M\br{\Theta}} B_u \br{	1+ C_1 \sum_{s=0}^{m} \rho^{s} \left\{1+  \sum_{\ell=1}^{q} |b_\ell|  \right\}}\notag\\
	&\le \delta^2_1 B_u	\label{ineq:40}
}
where the equality follows from~\eqref{eq:24} and the last inequality from~\eqref{def:M_th}. This shows that~\eqref{ineq:29} holds.
Since $|u_t|\le B_w,~\forall t\in \cI$,~\eqref{ineq:28} also clearly holds. 


Case 2). $t\in \left\{n_i + m_i +1, n_i + m_i + 2, \ldots,n_{i+1}\right\}$: For $\nu < n_{i+1}$, assume that~\eqref{ineq:28},~\eqref{ineq:29} hold for times $t\in\left\{n_i + m\ust,\ldots, n_{i} + m_i, \ldots, \nu \right\}$.~We will show that~\eqref{ineq:28},~\eqref{ineq:29} also hold for $t=\nu+1$. Note that we have already shown above that they hold for $t\in\left\{n_i + m\ust,\ldots, n_{i}+m_i  \right\}$.

We have,
\al{
	\br{\|\lm\|+1}\| \psi_{\nu+1}\| & \le \br{\|\lm\|+1}\br{|y_{\nu+1}| + |y_{\nu}| + \ldots + |y_{\nu + 1 -p}| + |u_{\nu}| + |u_{\nu -1}| + \ldots + |u_{\nu-q+1}|} \notag\\
	& \le \br{\|\lm\|+1}\br{p \delta^2_1 +q \delta_1 }B_u\notag\\
	& \le B_u,
}
where the second inequality follows from the induction hypothesis~\eqref{ineq:28},~\eqref{ineq:29}, and the last inequality follows since $\delta_1$ has been chosen to satisfy~\eqref{cond:delta_1} for $\te=\te\ust$.~Since $\br{\|\lm\|+1}\| \psi_{\nu+1}\| \le B_u$, it follows from~\eqref{ineq:31} that 
\al{
	| y_{\nu+2} |  \le \delta^2_1~B_u.
} 
This shows that~\eqref{ineq:29} holds for $t=\nu+1$.~It remains to show that we have $|u_{\nu +1}|\le \delta_1 B_u$. Consider the bound on $\|U_t\|$ derived in Lemma~\ref{lemma:state_bound} (ii),
\al{
	\|U_{t_1} \| \le C_1 \rho^{t_1 - t_0} \|U_{t_0}\| + \frac{C_1}{b_1} \sum_{s=0}^{t_1 - t_0 -1} \rho^{s} \left\{ | w_{t_1 + 1 - s}|  + \sum_{\ell=1}^{p}| a_\ell | | y_{t_1 + 1 - s - \ell}  | \right\},\label{ineq:32}
}
where $t_1 > t_0$. Consider time $t_0\in\left\{ n_i + m\ust, n_i + m\ust +1,\ldots,n_i + m_i  \right\}$.~Then $|w_t| \le B_w \le \frac{B_u \delta^2}{2}$, where the second inequality follows from~\eqref{eq:24}. Our induction hypothesis yields that $|y_t|\le  B_u \delta^2_1$ for times $t \in \left\{t_0, t_0 +1,\ldots, \nu \right\}$, and also $|u_t|\le \delta_1 B_u$ for $t \in \{t_0 +1,\ldots,\nu-1\}$, and we have shown above that $| y_{\nu+2} |  \le \delta^2$. Consider the vector $U_{\nu+1}= \br{u_{\nu+1},u_{\nu},\ldots,u_{\nu+1-q+2}}'$. Note that it follows from our induction hypothesis that $\|U_{t_0}\| \le \delta_1 B_u q$. Upon substituting these bounds into~\eqref{ineq:32}, and setting $t_1 = \nu+1$, we obtain the following,
\al{
	\|U_{\nu +1} \| & \le C_1 \rho^{\nu + 1 - t_0}\delta_1  B_u  q+ \frac{C_1}{b_1} \sum_{s=0}^{\nu - t_0} \rho^{s} \left\{ \frac{B_u\delta^2_1}{2} + \sum_{\ell=1}^{p}| a_\ell | \delta^2_1  B_u \right\}\notag\\
	&\le B_u \br{C_1 \rho^{\nu + 1 - t_0}\delta_1  q+ \frac{C_1}{b_1} \sum_{s=0}^{\nu - t_0} \rho^{s} \left\{ \frac{\delta^2_1}{2} + \sum_{\ell=1}^{p}| a_\ell | \delta^2_1  \right\}}\notag\\
	&= \delta_1~B_u \br{C_1 \rho^{\nu + 1 - t_0}q+ \delta_1 \frac{C_1}{b_1} \sum_{s=0}^{\nu - t_0} \rho^{s} \left\{ \frac{1}{2} + \sum_{\ell=1}^{p}| a_\ell |   \right\}}\notag\\
	& \le \delta_1  B_u,
}
where the last inequality holds since $\nu$ is sufficiently large and $\delta_1$ sufficiently small, i.e.,
\al{
	\nu + 1 - t_0 &\ge \log_{\rho}\frac{1}{2C_1 q},\\
	\delta_1 & \le \frac{1}{2}\left[ \frac{C_1}{b_1\br{1-\rho}}  \left\{ \frac{1}{2} + \sum_{\ell=1}^{p}| a_\ell |   \right\}\right]^{-1}.
}
The first condition holds since episode duration $H$ is sufficiently large, i.e. 
\al{
	H + 1 - \br{n_i + m\ust} \ge \log_{\rho}\br{\frac{1}{2C_1 q}}.\label{def:H_cond}
}
while the second holds because $\delta_1$ satisfies~\eqref{cond:delta_1}-\eqref{cond:delta_3}.

We have thus completed the induction step, and shown that~\eqref{ineq:28},~\eqref{ineq:29} holds for all $t\in \{n_i + m\ust,\ldots,n_{i+1} \}$. 

It remains to be shown that $|u_t|  \le \delta_1  B_u$ for times $t\in \cI$. But we already have $|u_t| \le B_w < \delta_1 B_u$, where the last inequality follows from~\eqref{cond:B1_B3_4}.~This proves the claim.
\end{proof}


\begin{proposition}\label{prop:9}
	When $\{w_t\}$ satisfies Assumption~\ref{assum:bounded_noise}, then the following holds on the set $\cG_{q} \cap \cG_{LSE}\cap \cG_{\pj}\cap \cG_{\cI}$ for $t\ge \max\left\{t\ust_1,t\ust_2\right\}$:
	\al{
		\|\psi_t\|\le \left[\br{1 +  C_1 \sum_{s=0}^{t} \rho^{s} \left\{1+  \sum_{\ell=1}^{q} |b_\ell|  \right\}}+q\right]  B_u \delta_1.\label{ineq:58}
	}
If instead $\{w_t\}$ satisfies Assumption~\ref{assum:sub_gaussian}, then the same conclusion holds on $\cG_{q} \cap \cG_{LSE}\cap \cG_{\pj}\cap \cG_{\cI}\cap \cG_w$.
\end{proposition}
\begin{proof} We have shown in Proposition~\ref{prop:8} that for $t\ge \max\left\{t\ust_1,t\ust_2\right\}$ w.h.p. we have
$$
|u_t| \le \delta_1 B_u,
$$
and hence 
\al{
	\|U_t\| \le q B_u\delta_1. \label{ineq:61}
}
Also,
\al{
	|w_t| &\le B_w\notag \\
	&<  B_u \delta_1,
}
where the second inequality follows from~\eqref{cond:B1_B3_4}.~Upon substituting these bounds in Lemma~\ref{lemma:state_bound}-(i), we obtain,
\al{
	\| Y_{t} \| &\le 	C_1 \rho^{t} \| Y_{0}\| + C_1 \sum_{s=0}^{t} \rho^{s} \left\{|w_{t-s}| +  \sum_{\ell=1}^{q} |b_\ell|  | u_{t -s - \ell } |  \right\}\notag\\
	&\le 	C_1 \rho^{t} \| Y_{0}\| +   B_u \delta_1 C_1 \sum_{s=0}^{t} \rho^{s} \left\{1+  \sum_{\ell=1}^{q} |b_\ell|  \right\}\notag\\
	&\le  B_u \delta_1	\br{1 +  C_1 \sum_{s=0}^{t} \rho^{s} \left\{1+  \sum_{\ell=1}^{q} |b_\ell|  \right\}},	\label{ineq:42}
}
where the first inequality follows since $|u_s| \le \delta_1 B_u$, and the last inequality follows from the definition of $t_2(\rho)$~\eqref{def:t2_star}, and since $t\ge t\ust_2$.~Since 
$$
\psi_t = \br{y_t,y_{t-1},\ldots,y_{t-p+1},u_{t-1},u_{t-2},\ldots, u_{t-q+1} }',
$$
we have $\|\psi_t\| \le \|Y_t\| + \|U_t\|$. The proof is then completed by substituting the bounds on $\|Y_t\|$ and $\|U_t\|$.
\end{proof}

\begin{lemma}\label{lemma:1}
	When $\{w_t\}$ satisfies Assumption~\ref{assum:bounded_noise}, then on the set $\cG_{q} \cap \cG_{LSE}\cap \cG_{\pj}\cap \cG_{\cI}$, for $t\ge \max\left\{t\ust_1,t\ust_2\right\}$, and $t\notin \cI$, we have $u_t = z_t$. If instead $\{w_t\}$ satisfies Assumption~\ref{assum:sub_gaussian}, then the same conclusion holds on $\cG_{q} \cap \cG_{LSE}\cap \cG_{\pj}\cap \cG_{\cI}\cap \cG_w$.
\end{lemma}
\begin{proof} From Proposition~\ref{prop:6} we have,
$$
|z_t|\le \br{\|\lm\|+1}\|\psi_t\|.
$$
From Proposition~\ref{prop:9}, we have,
\al{
	\|\psi_t\|\le \left[\br{1 +  C_1 \sum_{s=0}^{t} \rho^{s} \left\{1+  \sum_{\ell=1}^{q} |b_\ell|  \right\}}+q\right]  B_u \delta_1.
}
Upon combining these two bounds, we obtain, 
$$
|z_t|\le \br{\|\lm\|+1} \left[\br{1 +  C_1 \sum_{s=0}^{t} \rho^{s} \left\{1+  \sum_{\ell=1}^{q} |b_\ell|  \right\}}+q\right]  B_u \delta_1.
$$
Now, it follows from~\eqref{def:control_rule} that for $t\notin \cI$ whenever $|z_t|< B_u$, the input $u_t$ is set equal to $z_t$. Thus, when 
$$
\br{\|\lm\|+1} \left[\br{1 +  C_1 \sum_{s=0}^{t} \rho^{s} \left\{1+  \sum_{\ell=1}^{q} |b_\ell|  \right\}}+q\right] \delta_1 < 1,
$$
we have $u_t = z_t$ for sufficiently large $t$. This condition is true since $\delta_1$ satisfies~\eqref{cond:delta_4}.

\end{proof}

\section{Lower Bound on $\lm_{\min}(V\ui_t)$ {\color{red} }  }\label{sec:lmin}
The required theory of projections is developed in Section~\ref{sec:projection}.~This is utilized in order to obtain a lower bound on the quantity $\lm_{\min}(V\ui_t)$ that is used in Section~\ref{sec:estimation_error} for controlling the estimation error. This is stated in Theorem~\ref{th:1}, which is the main result of this section.

We begin with the bound on $\lm_{\min}(V_t)$ that was derived in Section~\ref{sec:lmin}.

Consider the ARX model~\eqref{def:arx}, repeated here for convenience,
\al{
	y_t= a_1 y_{t-1} + a_2 y_{t-2} + \ldots +a_p y_{t-p} +b_1 u_{t-1} + b_2 u_{t-2}+\cdots + b_q u_{t-q}+ w_t, t=1,2,\ldots, 
}
where $u_s$ is $\cF_{s-1}$-measurable, while $y_s$ is $\cF_s$-measurable. Consider the design matrix associated with $\te_t$, the LS estimate of $\te$ at time $t$,
\al{
	\Psi_t := 
	\begin{pmatrix}
		y_{I-1} & \cdots & y_{I-p} & u_I & \cdots & u_{I-q} \\
		\vdots & 	& & & & \\
		y_{t-1}& \cdots & y_{t-p}& u_t & \cdots & u_{t-q}
	\end{pmatrix}
	\label{def:design_mat}
}
where $I> \max\{p,q\}$. Also let,
\al{
	U_s = \br{u_s,u_{s-1},\ldots,u_{s-q}}'.
}

\begin{lemma}\label{lemma:input_excite}
	Define the event,
	\al{
		\cG_{\cI} := \left\{ \lm_{\min}\br{\sum_{s=1}^{t} U_s U'_s  } \ge   \sigma^2_e N\ui_t - B^2_w\sqrt{2N\ui_t \log\br{\frac{qN\ui_T}{\delta}}  }    \right\},	\label{def:gi}
	}
	where $\sigma^2_e$ is the variance of the exploratory noise, i.e.
 \al{
\sigma^2_e := \bE\br{u^2_t},~t\in\cI. \label{def:sigma_e}
 }
We have,
	\al{
		\bP\br{\cG_{\cI}} \ge 1-\delta,
	} 
	and on $\cG_{\cI}$ the following holds,
	\al{
		\lm_{\min}\br{\sum_{s\le t} U_s U'_s } \ge N\ui_t  .
	}	
	For 
		\al{
	t \ge \frac{2B^{4}_w  }{\sigma^4_e}	\log\br{\frac{qN\ui_T}{\delta}},\label{ineq:89}
	}
 on $\cG_{\cI}$, we have,
	\al{
		\lm_{\min}\br{\sum_{s=1}^{t} U_s U'_s  } \ge   \frac{\sigma^2_e}{2} N\ui_t.\label{ineq:76}
	}
	
\end{lemma}

\begin{proof}
	Since,
	\al{
		\lm_{\min}\br{\sum_{s\le t} U_s U'_s } \ge 	\lm_{\min}\br{\sum_{s\in\cI,s\le t} U_s U'_s },
	} 
	we will instead derive a lower bound on the latter quantity. For $i\neq j$, consider the $(i,j)$-th element of the matrix $\sum_{s\in\cI,s\le t} U_s U'_s $. This is given by $\sum_{s \in \cI} u_{s-i}u_{s-j}$, without loss of generality assume $i<j$. Define new random variables $\{\tilde{u}_s \}$ such that $\tilde{u}_s$ is the $s$-th exploratory input. 
	
	This sum $\sum_{s \in \cI} u_{s-i}u_{s-j}$ is equivalent to $\sum_{s} \tilde{u}_{s}\tilde{u}_{s+ j - i}$. Now consider the filtration $\{\tilde{\cF}_s\}$ defined as follows: $\tilde{\cF}_s$ is the sigma-algebra generated by $\{\tilde{u}_{\ell}\}_{\ell=1}^{s}$. Now, $\{\tilde{u}_{s}\tilde{u}_{s+ j - i} , \tilde{\cF}_{s+j-i-1}\}$ is a martingale difference sequence. By using the Azuma-Hoeffding inequality we deduce,
	\nal{
		\bP\br{ \Big| \sum_{s=1}^{t} \tilde{u}_{s}\tilde{u}_{s+ j - i} \Big| > \eps  } \le \exp\br{-\frac{\eps^2}{2t(B^2_w)^2} },
	}
	where $t< N\ui_T$.~Letting $\eps = \sqrt{2t (B^2_w)^2\log\br{\frac{qN\ui_T}{\delta}}  }$, we deduce that the event 
	$
	\left\{\Big| \sum_{s=1}^{t}  \tilde{u}_{s}\tilde{u}_{s+ j - i} \Big| > \sqrt{2t (B^2_w)^2\log\br{\frac{qN\ui_T}{\delta}}  }  \right\},
	$ 
	has a probability less than $\frac{\delta}{2qN\ui_T}$. Upon using a union bound over $t$, and all possible $i\neq j$, we conclude that the probability of the following event is less than $\delta\slash 2$,
	\al{
		\left\{\Big| \sum_{s=1}^{t} \tilde{u}_{s}\tilde{u}_{s+ j - i} \Big| \le  \sqrt{2t (B^2_w)^2\log\br{\frac{2qN\ui_T}{\delta}}  }  ,~\forall t = 1,2,\ldots,N\ui_T, \forall i,j \in \{1,2,\ldots,q\}, i\neq j \right\}.\label{ineq:74}
	}
	One may note that \eqref{ineq:74} is equivalent to the off-diagonal entries of $\sum_{s\in\cI,s\le t} U_s U'_s$ being less than $B^2_w\sqrt{2t \log\br{\frac{qN\ui_T}{\delta}}  }$. The diagonal terms of $\sum_{s\in\cI,s\le t} U_s U'_s$ are $\sum_{s\le N\ui_t} \tilde{u}^2_s$. Upon using Azuma-Hoeffding and a union-bound on $t$ on the process $\{\tilde{u}^2_s - \bE\br{\tilde{u}^2_s}\}$, we deduce that the following event has a probability less than $\delta\slash 2$,
	\al{
		\left\{\Big| \sum_{s=1}^{t}\tilde{u}^2_s - t \bE\br{\tilde{u}^2_s} \Big|  \le  \sqrt{2t (B^2_w)^2\log\br{\frac{2N\ui_T}{\delta}}  }  ,~\forall t = 1,2,\ldots,N\ui_T \right\}.\label{ineq:75}
	}
	The proof then follows from the Gershgorin circle theorem~\cite{horn2012matrix}.
\end{proof}

Let $I> \max\{p,q\}$. For $\ell \in \bZ_+$, define,
\al{
	y\ut(\ell) := (y_{I-\ell},\ldots,y_{t-\ell}),\label{def:yl}
}
\al{
	u\ut(\ell) := \br{u_{I-\ell},u_{I+1-\ell},\ldots, u_{t-\ell}}',\label{def:ul} \mbox{ and}
}
\al{
	\underline{w}\ut(\ell) := \br{w_{I-\ell},w_{I+1-\ell},\ldots,w_{t-\ell}}',\label{def:wl}
}
Define,
\al{
	D_t : =
	\begin{pmatrix}
		u_I & \cdots u_{I-q} \\
		\vdots & 		\\
		u_t & \cdots u_{t-q}
	\end{pmatrix}
	= \br{u(0),u(1),\ldots,u(q)}.
}
The design matrix $\Psi_t$~\eqref{def:design_mat} can thus be written as 
\al{
	\Psi_t = \br{y\ut(1),\ldots,y\ut(p),D_t}.
}
Let $d$ be a column of $\Psi_t$, and $\hat{d}$ its projection onto the linear space spanned by the remaining columns of $\Psi$. We will derive a lower bound on the quantity $\|d-\hat{d}\|$. This will yield us a lower bound on $\lambda_{\min}\br{\Psi'_t \Psi_t}$ since from Lemma~\ref{lemma:lmin_proj} we have,
\al{
	(p+q)~ \| d -\hat{d} \| \ge 	\lambda_{\min}\br{\Psi'_t \Psi_t} \ge \br{p+q}^{-1} \| d -\hat{d} \|.\label{ineq:rank}
} 

In the sequel, we will omit the superscript $t$ when it is clear from the context. Also define
\al{
	&t\ust_{cov}(\rho,\delta) := \inf \Bigg\{ t\in\bN: \frac{pqB^2_u}{2c_1}\vee  \frac{8pqB^2_u}{c_1}\vee \frac{2c^2 \log\br{\frac{1}{\delta}}}{c^2_1} \notag\\
	&\vee p\frac{B^{2}_u+2\|Y_0\|^2 + \frac{4\|b\|^2}{\br{1-\rho}^2} B^2_u q + \frac{4B^2_w }{\br{1-\rho}^2} }{ \br{c_1\slash 3} \delta}\vee \frac{2c^2 \log\br{\frac{1}{\delta}}}{\br{c_1 \slash 3}^2}\vee \frac{2B^{4}_w  }{\sigma^4_e}	\log\br{\frac{qN\ui_T}{\delta}} \Bigg\}.\label{def:t_cov}
} 
To ease notation, we will occasionally omit the dependence of $t\ust_{cov}$ on $(\rho,\delta)$.
\begin{proposition}\label{prop:11}
	Consider the ARX system~\eqref{def:arx} and assume that $\{w_s\}$ satisfies Assumptions~\ref{assum:noise_var} and \ref{assum:bounded_noise}.~Let $d$ be a column of $D_t$, and $\hat{d}$ its projection onto the linear space spanned by the remaining columns of $\Psi_t$.~Let $\tilde{c}> 2c_1$, and define,
	\al{
		\beta_1 := \frac{c_1 \slash 4}{6\left[\frac{2}{p} \|Y_0\|^2 + \frac{4\|b\|^2}{p\br{1-\rho}^2} B^2_u q + \frac{4B^2_w}{p\br{1-\rho}^2} \right]+4B^2_w },~\label{def:beta}
	}
where $c_1$ is as in~\eqref{cond:3}.~On the set $\cG_{\pj}\cap \cG_{\cI}$, we have,
	\al{	
		\|d-\hat{d} \| \ge \frac{\sigma^2_e N\ui_t}{2q} \left[ 1\wedge \min_{\ell \in \{1,2,\ldots,p\} } \beta^{\ell\slash 2}_1 \right],~\mbox{ for } t\ge t\ust_{cov}.
	}
\end{proposition}

\begin{proof}

Consider a column $d = \br{d_I,\ldots,d_t}$ of $D_t$. Let $D\ust$ be the sub-matrix of $D_t$ consisting of all the other columns except $d$. For $\ell=1,2,\ldots,$ let $D\ust(\ell)$ be the matrix $\br{y(\ell),y(\ell+1),\ldots,y(p),D\ust}$.~For $\ell =1,2,\ldots,p$, let $\hat{d}_\ell$ denote the projection of $d$ onto $L\br{y(\ell),y(\ell+1),\ldots,y(p),D\ust}$. Let $\hat{d}_0$ be the projection of $d$ onto $L(D\ust)$. 

We begin with deriving a lower-bound on $\|d - \hat{d}_{0}\|$.~we have,
\al{
	q~ \| d -\hat{d}_0 \| &\ge 	\lambda_{\min}\br{D'_t D_t} \notag\\
	& \ge \frac{\sigma^2_e N\ui_t}{2q}.\label{ineq:12}
}
where the first inequality follows from Lemma~\ref{lemma:lmin_proj}, while the second follows from Lemma~\ref{lemma:input_excite} since $D'_t D_t = \sum_{s=1}^{t} U_s U'_s$. This shows $\lambda_{\min}\br{D'_t D_t}\ge \frac{\sigma^2_e}{2} N\ui_t$.~Upon substituting this into~\eqref{ineq:12}, we obtain
\al{
	\| d-\hat{d}_0 \| \ge \frac{\sigma^2_e N\ui_t}{2q}.\label{ineq:13}
}

Now we will derive lower bounds for $\|d-\hat{d}_{\ell}\|$ for $\ell=1,2,\ldots,p$. We will show that
\al{
	\| d-\hat{d}_{\ell} \|^2 \ge \| d-\hat{d}_0 \|^2 ~ \beta^{p+1-\ell}_1,~\ell=1,2,\ldots,p,\label{ineq:64}
}
where $\beta>0$ is as in~\eqref{def:beta}

We will prove this via induction. We begin with $\ell=p$.~Consider the vector $y\ut(p)$. Its $i-I$-th element ($i\ge I$) is $y_{i-p}$, and is equal to
\nal{
	y_{i-p} = a_1 y_{i-p-1}  + a_2 y_{i-p-2} +\ldots + y_{i-p-p} + b_0 u_{i-p} + \ldots + u_{i-p-q}+w_{i-p}.
}
This can be written in vector form as,
\al{
	y\ut(p) &= \br{y_{I-p},\ldots,y_{t-p}}'\notag\\
	& = \br{v_I,\ldots,v_t}' + \br{w_{I-p},\ldots,w_{t-p}}'\label{eq:18}\\
	& = \br{v_I,\ldots,v_t}' + \underline{w}(p),
}
where 
\al{
	v_i = a_1 y_{i-p-1}+\ldots + a_p y_{i-p-p}+b_1 u_{i-p}+\ldots+b_q u_{i-p-q}. \label{def:v}
}
Note that $v_i$ is $\cF_{i-p-1}$ measurable. Hence $u_i$, and therefore also $h_i$ are $\cF_{i-k-1}$ measurable. From~\eqref{eq:18}, $\hat{d}_p$ is the projection of $d$ onto $L(D\ust,v+\underline{w}(p))$. Let $\hat{\underline{w}}_0(p)$ be projection of $\underline{w}(p)$ onto $D\ust$.~Let $v\ust$ be the projection of $v$ onto $L(D,y(p))$.~Let $\hat{y}_{0}(p)$ be the projection of $y_n(p)$ onto $L(D\ust)$. Define,
\al{
\cS_3 :&= 1 \vee \sqrt{ \log^{+}\br{\frac{\|d - \hat{d}_p \|}{\delta}} }
			 \vee 
			 \sqrt{2\log\br{\sum_{s=1}^{t}\|\phi_s\|^2} },\label{def:S3}\\ 
 \cS_4:&= 1 \vee \sqrt{ \log^{+}\br{\frac{\| v-\hat{v} \|}{\delta}} } \vee \sqrt{2\log\br{\sum_{s=1}^{t}\|\phi_s\|^2
 } }.  \label{def:S4}
}
~It follows from~ \eqref{ineq:10} that the conditions of Theorem~\ref{th:proj_3} are satisfied, and hence we can use Theorem~\ref{th:proj_3} and obtain the following after performing some algebraic manipulations,
\al{
	\|d - \hat{d}_p\|^2& \ge \frac{\|d-\hat{d}_0\|^2\left\{\|v-v\ust\|^2 + \| \underline{w}(p) - \hat{\underline{w}}_0(p)\|^2 - \cS_3\right\}}{\|\br{v-\hat{v}}\|^2 + \| \underline{w}(p) -\hat{\underline{w}}_0(p) \|^2 + 2 \| v-\hat{v} \|\left\{  \| v-\hat{v} \|) + \cS_4 \right\}}\notag\\
	&\ge \frac{\|d-\hat{d}_0\|^2\left\{ \| \underline{w}(p) - \hat{\underline{w}}_0(p)\|^2 - \cS_3 \right\}}{\|\br{v-\hat{v}}\|^2 + \| \underline{w}(p)-\hat{\underline{w}}_0(p)\|^2 + 2 \| v-\hat{v} \|\left\{  \| v-\hat{v} \|) + \cS_4 \right\}}\notag\\
	&= \frac{\|d-\hat{d}_0\|^2\left\{ \| \underline{w}(p) - \hat{\underline{w}}_0(p)\|^2 - \cS_3 \right\}}{3\|\br{v-\hat{v}}\|^2 + \| \underline{w}(p)-\hat{\underline{w}}_0(p)\|^2 + 2 \| v-\hat{v} \| \cS_4}\notag\\	
	&\ge \frac{\|d-\hat{d}_0\|^2\left\{ \| \underline{w}(p) - \hat{\underline{w}}_0(p)\|^2 - \cS_3 \right\}}{6\|\br{y(p)-\hat{y}_0(p)}\|^2 + 4\| \underline{w}(p)-\hat{\underline{w}}_0(p)\|^2 + 2 \| v-\hat{v} \| \cS_4 }\label{ineq:49}\\
	&\ge \frac{\|d-\hat{d}_0\|^2\left\{ \| \underline{w}(p) - \hat{\underline{w}}_0(p)\|^2 - \cS_3 \right\}}{6\|y(p)\|^2 + 4\| \underline{w}(p) - \hat{\underline{w}}_0(p)\|^2 + 2 \| v-\hat{v} \| \cS_4 }, \label{ineq:14}
}
where~\eqref{ineq:49} follows since $y(p) = v + \underline{w}$, so that $y(p)-\hat{y}(p) = v - \hat{v} +  \underline{w}(p) - \hat{\underline{w}}(p)$, and hence $\|v - \hat{v} \|^2 \le 2 \|y(p)-\hat{y}(p)\|^2+ 2\|\underline{w}(p) - \hat{ \underline{w}}(p)\|^2$.~We will now derive bounds on various terms in the numerator and denominator of~\eqref{ineq:14}, which will allow us to lower-bound this expression. 

Since $y(p)=(y_{I-p},\ldots,y_{t-p})$, we have,
\al{
	\|y(p)\|^2	& = \|Y_{t-p} \|^2 +  \|Y_{t-2p} \|^2 + \ldots + \|Y_{t - \lfloor t \rfloor } \|^2.\label{eq:25}
}
After bounding $\|Y_{t-p} \|^2,\|Y_{t-2p} \|^2,\ldots,\|Y_{t - \lfloor t \rfloor } \|^2$ using Proposition~\eqref{prop:5}, and performing algebraic manipulations, we obtain,
\al{
	\|y(p)\|^2	 \le \frac{2t}{p} \|Y_0\|^2 + \frac{4\|b\|^2t}{p\br{1-\rho}^2} B^2_u q + \frac{4B^2_w t}{p\br{1-\rho}^2}.\label{ineq:83}
}
~Proposition~\ref{prop:2} gives us the following lower bound,
\al{
	& \| \underline{w}(p) - \hat{\underline{w}}(p) \|^2 \ge t c_1 -  \sqrt{2t c^2 \log\br{\frac{1}{\delta}}} \notag\\
	& -
	p~\left\{ 1 \vee \log^{+}\br{\frac{q\times \sum_{j=1}^{t}u^2_{j}}{\delta}    } \vee 	2\log\br{q\times \sum_{j=1}^{t}u^2_{j}}
	\right\},~\forall t,
	\label{ineq:35}
} 
where $c_1$ is as in Assumption~\ref{assum:noise_var}.

We also have,
\al{
	\| \underline{w}(p) - \hat{\underline{w}}(p)\|^2 & \le \|\underline{w}\ut(p)\|^2\notag\\
	& \le B^2_w t,~\forall t. \label{ineq:36}
}
Since from~\eqref{def:v} we have $v_i = a_1 y_{i-p-1}+\ldots + a_p y_{i-p-p}+b_1 u_{i-p}+\ldots+b_q u_{i-p-q}=(a,b)\cdot \phi_{i-p}$ we get,
\al{
	\| v-\hat{v} \| & \le  \|v\|\notag\\
	& \le \|(a,b)\| \sqrt{\sum_{s=1}^{t-p} \|\phi_s\|^2} \notag\\
	&\le   \|(a,b)\|   \sqrt{\sum_{s=1}^{t}2\|A^s\|^2 \|Y_0\|^2 + \frac{4\|b\|^2}{1-\rho} \br{\sum_{s=1}^{t}\sum_{\ell=1}^{s} \rho^{s-\ell} \|U_{\ell}\|^2 } + \frac{4}{\br{1-\rho}^2}B^2_w t +  q B^2_ut}\notag\\
	&\le   \|(a,b)\|   \sqrt{\sum_{s=1}^{t}2\|A^s\|^2 \|Y_0\|^2 + \frac{4\|b\|^2}{\br{1-\rho}^2} \br{\sum_{s=1}^{t}\|U_{s}\|^2 } + \frac{4}{\br{1-\rho}^2}B^2_w t +  q B^2_ut}  \notag\\
	&\le   \|(a,b)\|   \sqrt{\sum_{s=1}^{t}2\|A^s\|^2 \|Y_0\|^2 + \frac{4\|b\|^2}{\br{1-\rho}^2} B^2_u q t + \frac{4}{\br{1-\rho}^2}B^2_w t +  q B^2_ut} \quad ,\label{ineq:37}
}
where the third inequality follows from~\eqref{ineq:84}.

Substituting the bounds~\eqref{ineq:83}-\eqref{ineq:37}, and also $|u_s|\le B_u$ into~\eqref{ineq:14}, we get,
\al{
	\|d\ut - \hat{d}\ut_p\|^2& \ge \|d\ut-\hat{d}\ut_0\|^2 \cdot \cT_5 ,~\forall t,
}
where,
\al{
	\cT_5 := 	\frac{t c_1 -  \sqrt{2t c^2 \log\br{\frac{1}{\delta}}} -  p~\left\{ 1 \vee \log^{+}\br{\frac{qB^2_ut}{\delta}    } \vee 	2\log\br{qB^2_ut}
		\right\}- \cS_3 }{6\left[\frac{2t}{p} \|Y_0\|^2 + \frac{4\|b\|^2t}{p\br{1-\rho}^2} B^2_u q + \frac{4B^2_w t}{p\br{1-\rho}^2}\right]+4tB^2_w +2\cdot \cT_6 \cdot  \cS_4}.
}
%
with,
$\cT_6 :=\|(a,b)\|   \sqrt{\sum_{s=1}^{t}2\|A^s\|^2 \|Y_0\|^2 + \frac{4\|b\|^2}{\br{1-\rho}^2} B^2_u q t + \frac{4}{\br{1-\rho}^2}B^2_w t +  q B^2_ut}$.

After performing some algebraic manipulations, we get that for $t\ge t\ust_{cov}$~\eqref{def:t_cov}, we have,
\al{
	\cT_5  &\ge \frac{c_1 \slash 4}{6\left[\frac{2}{p} \|Y_0\|^2 + \frac{4\|b\|^2}{p\br{1-\rho}^2} B^2_u q + \frac{4B^2_w}{p\br{1-\rho}^2} \right]+4B^2_w }\notag\\
	& = \beta_1,
}
or equivalently,
\al{
	\|d\ut - \hat{d}\ut_p\|^2& \ge \|d\ut-\hat{d}\ut_0\|^2\beta_1,~\forall t.
}

Next, suppose that~\eqref{ineq:64} holds for $\ell =m+1,m+2,\ldots,p$. We will show that~\eqref{ineq:64} holds for $\ell=m$. Once again, similar to~\eqref{eq:18}, we have,
\al{
	y(m) = v + \underline{w}(m),
}
where the $i-I$-th element ($i\ge I$) of $v$ is $y_{i-m}$, and is given by,
\nal{
	y_{i-m} = a_1 y_{i-m-1}  + a_2 y_{i-m-2} +\ldots + a_p y_{i-m-p} + b_0 u_{i-m} + \ldots + b_q u_{i-m-q}.
}
This shows that $\hat{d}_{m}$ is the projection of $d$ onto $L\br{D\ust(m+1),v+ \underline{w}(m)}$. Let $\hat{\underline{w}}_{\ell_2}(\ell_1)$ be the projection of $\underline{w}(\ell_1)$ onto $D\ust(\ell_2)$, and $v\ust$ the projection of $v$ onto $L\br{D\ust_{m+1},y(m)}$.~It follows from~\eqref{ineq:10} that the conditions of Theorem~\ref{th:proj_3} are satisfied, and hence we can use Theorem~\ref{th:proj_3} and arguments similar to~\eqref{ineq:14} to obtain the following:
\al{
	\|d - \hat{d}_m\|^2 \ge \|d - \hat{d}_{m+1}\|^2 \beta_1.
}
This completes the induction step, and hence we have shown~\eqref{ineq:64}.

The proof of the claim then follows by substituting~\eqref{ineq:12} into~\eqref{ineq:64}.  
\end{proof}

\begin{proposition}\label{prop:12}
	Consider the ARX system~\eqref{def:arx} and assume that $\{w_s\}$ satisfies Assumptions~\ref{assum:sub_gaussian}-(ii) and~\ref{assum:noise_var}.
	Let $c$ be a column of the matrix $\br{y_n(1),\ldots,y_n(p)}$, and $\hat{c}$ be its projection onto the linear space spanned by the remaining columns of 
	\al{
		\Psi_t = \br{y_n(1),\ldots,y_n(p),D_t}, \mbox{ where } D_t = \br{u(0),u(1),\ldots,u(q)}.
	}
	Let $\{w_t\}$ satisfy Assumption~\ref{assum:sub_gaussian}-(i). On $\cG_{\pj}\cap \cG_{\cI}$ we have,
	\al{	
		\|c-\hat{c} \| \ge \sqrt{\frac{c_1 t}{4} }\left[1 \wedge \min_{\ell \in \{1,2,\ldots,p\}}  \beta^{\ell\slash 2}_2 \right],~\forall t \ge t\ust_{cov},
	}	
where $t_{cov}$ is as in~\eqref{def:t_cov}, and
	\al{
		\beta_2 : =  \frac{\br{c_1 \slash 3}}{3a^2_1\left[\frac{2}{p} \|Y_0\|^2 + \frac{4\|b\|^2}{p\br{1-\rho}^2} B^2_u q + \frac{4B^2_w}{p\br{1-\rho}^2} \right]+B^2_w}.\label{def:gamma}
	}
\end{proposition}

\begin{proof}
	
	Let the column $c$ be $y_n(\ell)$, where $\ell\in \{1,2,\ldots,p\}$. Also let $\Psi\ust$ be the sub-matrix of $\Psi$ obtained by removing $y(\ell)$. Recalling that $y_n(\ell) = \br{y_{I-\ell},\ldots,y_{n-\ell}}'$ and
	$$
	u(m) = \br{u_{I-m},u_{I+1-m},\ldots, u_{n-m}}',
	$$
	define the matrix,
	$$
	\Psi(\ell):= \br{y(\ell+1), y(\ell+2),\ldots,y(\ell+p), u(0),u(1),\ldots,u(\ell+q) }.
	$$
	Let $L_{\ell}$ be the linear space spanned by $y(1),y(2),\ldots,y(\ell-1)$ and the columns of $\Psi(\ell)$.~Since $L(\Psi\ust)$ is a subspace of $L_{\ell}$, clearly, 
	$$
	\| y(\ell) - \pj\br{y(\ell), L(\Psi\ust)}   \| \ge \| y(\ell) - \pj\br{y(\ell), L_{\ell} }  \|.
	$$
	Thus, in order to show the claim, we will derive a lower bound on $\| y(\ell) - \pj\br{y(\ell), L_{\ell} }  \|$. Let $\hat{ \underline{w}}_0(\ell)$ be the projection of $\underline{w}(\ell)$ onto $L\br{\Psi(\ell)}$.
	
	For $i = 1,2,\ldots,\ell-1$, let $\pi_i$ be the projection of $y(\ell)$ onto 
	$L(y(\ell-i),\ldots,y(\ell-1),\Psi(\ell))$. Also let $\pi_0$ be the projection of $y(\ell)$ onto $L(\Psi(\ell))$.~We will now derive lower bounds on $\|y(\ell)- \pi_i\|$. We begin with $i=0$. Now, $y(\ell)$ is a linear combination of the columns of $\Psi(\ell)$, and $\underline{w}(\ell)= (w_{I-\ell},\ldots, w_{t-\ell})'$,
	\al{
		y(\ell) = \sum_{s=1}^{p} a_s y(\ell+s) + \sum_{s=1}^{q} b_s u(\ell+s) + \underline{w}(\ell).
	}
	Since the vectors $\{y(\ell+s)\}_{s=1}^{p}$, $\{u(\ell+s)\}_{s=1}^{q}$ belong to $\Psi(\ell)$, we have,
	\al{
		y(\ell) - \pi_0 = \underline{w}(\ell) -\hat{ \underline{w}}_0(\ell).\label{ineq:67}
	}
	Hence, we will now derive a lower bound on $ \underline{w}(\ell) -\hat{ \underline{w}}_0(\ell)$.~From Proposition~\ref{prop:2} we have ($c$ is as in~\ref{def:c}),
	\al{
	&	\| \underline{w}(\ell) -\hat{ \underline{w}}_0(\ell)\|^2 \ge c_1 t -  \sqrt{2t c^2 \log\br{\frac{1}{\delta}}} - p~\Bigg\{ 1 \vee \log^{+}\br{\frac{\sum_{s=0}^{\ell+q} \|u(s)\|^2 +\sum_{s=\ell+1}^{\ell+p} \|y(s) \|^2}{\delta}    } \notag\\
	&	\vee 	\log\br{\sum_{s=0}^{\ell+q} \|u(s)\|^2 +\sum_{s=\ell+1}^{\ell+p} \|y(s) \|^2}
		\Bigg\}.
	} 
Now $\sum_{s=0}^{\ell+q} \|u(s)\|^2 +\sum_{s=\ell+1}^{\ell+p} \|y(s) \|^2 $ can be bounded by $B^{2}_ut+2t\|Y_0\|^2 + \frac{4\|b\|^2t}{\br{1-\rho}^2} B^2_u q + \frac{4B^2_w t}{\br{1-\rho}^2}$ using techniques as in~\eqref{eq:25},~\eqref{ineq:83}.~Thus, when $t\ge t\ust_{cov}$,
	\al{
		\| \underline{w}(\ell) -\hat{ \underline{w}}_0(\ell)\|^2 \ge \frac{c_1 t}{3},
	} 
	which when combined with~\eqref{ineq:67} yields,
	\al{
		\|y(\ell) - \pi_0\|^2 \ge \frac{c_1 t}{3}.
	}
	
	Next, consider $i=1$.~We have
	\al{
		y(\ell-1) = a_1 y(\ell) + \sum_{s=1}^{p-1} a_{s+1}y(\ell+s) + \underline{w}(\ell-1),
	}
	where $\underline{w}(\ell-1):= \br{w_{I-\ell+1},\ldots, w_{n-\ell+1}  }'$. Since the columns $\{y(\ell+s)\}_{s=1}^{p-1}$ belong to $\Psi(\ell)$, we get,
	\al{
		L\br{y(\ell-1),\Psi(\ell)} = L\br{a_1 y(\ell) +  \underline{w}(\ell-1),\Psi(\ell) }.
	}
	Hence, setting $v=a_1 y(\ell)$, we can use Theorem~\ref{th:proj_3} to obtain 
	\al{
		& \| y(\ell)- \pi_1 \|^2 \ge \| y(\ell)- \pi_0 \|^2\times \notag\\
		&
		\frac{\| \underline{w}(\ell-1) -\hat{\underline{w}}_0(\ell-1)  \|^2 - \cS_3}
		{3a^2_1\|y(\ell) - \pi_0 \|^2 +\| \underline{w}(\ell-1) -\hat{\underline{w}}_0(\ell-1)  \|^2+ 2 \cS_4a_1\|y(\ell) - \pi_0 \| },\label{ineq:50}
	}   
	where $\hat{\underline{w}}_0(\ell-1)$ is the projection of $\underline{w}(\ell-1)$ onto $L\br{\Psi(\ell)}$, and 
 $\cS_3,\cS_4$ are as in~\eqref{def:S3},~\eqref{def:S4}.
	
	Similar to~\eqref{ineq:83} we have,
	\al{
		\|y(\ell) - \pi_0 \|^2 & \le \|y(\ell)\|^2\notag \\
		& \le \frac{2t}{p} \|Y_0\|^2 + \frac{4\|b\|^2t}{p\br{1-\rho}^2} B^2_u q + \frac{4B^2_w t}{p\br{1-\rho}^2}.
	}
	Upon substituting this and the bounds~\eqref{ineq:35}-\eqref{ineq:37} into~\eqref{ineq:50}, and performing algebraic manipulations similar to the proof of Proposition~\ref{prop:11}, we obtain,
	\al{
		\| y(\ell)- \pi_1 \|^2 \ge \| y(\ell)- \pi_0 \|^2\cdot \beta_2.
	}
	
	The proof is then completed by induction.

\end{proof}
The following is the main result of this section and provides a lower bound on the minimum eigenvalue of $\Psi'_t \Psi_t$ that holds w.h.p.

\begin{theorem}\label{th:1}
	Consider the ARX system~\eqref{def:arx} and let $\{w_t\}$ satisfy Assumption~\ref{assum:bounded_noise}.~On $\cG_{\pj}\cap \cG_{\cI}\cap \cG_{w^2_B}$, for times $t\ge t\ust_{cov}$,
	\al{
		\lambda_{\min}\br{\Psi'_t \Psi_t} \ge \beta_3 N\ui_t,\label{ineq:48}
	}
where 
\al{
\beta_3 := \br{\frac{\sigma^2_e}{2q}\min_{\ell \in \{0,1,\ldots,p\} } \beta^{\ell\slash 2}_1 }\wedge \br{\frac{c_1}{4}\min_{\ell \in \{0,1,\ldots,p\}}  \beta^{\ell\slash 2}_2}.\label{def:beta_3}
}
Same conclusion holds for $\{w_t\}$ satisfying Assumption~\ref{assum:sub_gaussian} on the set $\cG_{\pj}\cap \cG_{\cI}\cap \cG_{w^2_{UB}}$
 
\end{theorem}
\begin{proof}
	Follows from Propositions~\ref{prop:11},~\ref{prop:12},~\eqref{ineq:rank} and noting that $N\ui_t \le \sqrt{t}$ under the proposed algorithm, after some algebraic manipulations. 
\end{proof}

\subsection{Properties of Projections}\label{sec:projection}

The material in this section contains ``finite-time version'' of the results in~\cite{lai1982asymptotic}.~More specifically, the proof of Theorem~\ref{th:1} relies upon finite-time versions of Corollary~2 and Theorem~5 of~\cite{lai1982asymptotic}, but we will derive non-asymptotic versions of several results from~\cite{lai1982asymptotic} since the finite-time versions of Corollary~2 and Theorem~5~\cite{lai1982asymptotic} depend upon them. The results in this section are of independent interest, and have much wider application. The main result of this section is Theorem~\ref{th:proj_3}, and it is used in the proof of Propositions~\ref{prop:11} and~\ref{prop:12} while deriving a lower bound on the minimum eigenvalue of the covariance matrix.

Within this section we consider stochastic processes $\{x_s\}_{s=1}^{T},\{w_s\}_{s=1}^{T},\{z_s\}_{s=1}^{T},\{v_s\}_{s=1}^{T}$ where $z_s = (z_{s,1},z_{s,2},\ldots,z_{s,p})$ is a vector-valued process. While performing analysis, we will be interested in $t$-dimensional vectors created from the first $t$ components of these processes, with time index ranging from $1$ to $t$. Hence denote $Z\ut = \{z_{i,j}\}_{1\le i \le t,1\le j \le p}$, $x\ut = (x_1,x_2,\ldots,x_t)'$, $w\ut = (w_1,w_2,\ldots,w_t)'$ and $v\ut = (v_1,v_2,\ldots,v_t)'$.~For a matrix $M$, we let $L(M)$ be the linear space spanned by its columns. When the time $t$ is clear from the context, we will omit the superscript $t$, which will be mostly the case in this section since the analysis is performed by fixing $t$. So we will write $x$ in lieu of $x\ut$, and so on. Only when we explicitly want to depict the dependence upon $t$, will we use a super-script.~Let $\hat{x},\hat{w},\hat{v}$ be the projections of the vectors $x,w,v$ onto $L(Z)$. Note that $\{w_s\}$ is a martingale difference sequence w.r.t. $\{\cF_s\}$. For each $s\ge 1$, $x_s,v_s,z_s$ are $\cF_{s-1}$ measurable random variables. In this section, we will derive the results for the case when $\{w_s\}$ is either bounded (Assumption~\ref{assum:bounded_noise}) or sub-Gaussian (Assumption~\ref{assum:sub_gaussian}).

\begin{remark}
	Theorem~\ref{th:proj_3}, the main result of this section,  gives a lower bound on $\pj\br{x\ut,L(z\ut,v\ut+w\ut)}$. In the context of our problem this requires us to lower-bound the minimum eigenvalue of the covariance matrix, the (random) vector $z\ut$ corresponds to a submatrix of the design matrix $\Psi_t$~\eqref{def:design_mat} obtained by deleting a few columns, and the vector $v\ut$ corresponds to a linear combination of some columns of $\Psi_t$, while $w\ut$ is simply the column vector $\br{w_1,w_2,\ldots,w_t}'$ obtained from the noise process $\{w_s\}_{s=1}^{t}$ associated with ARX process~\eqref{def:arx}.
\end{remark}

The following is the finite-time version of Theorem~4 of~\cite{lai1982asymptotic}. 
\begin{theorem}\label{th:proj_1}
If $\{w_s\}$ satisfies Assumption~\ref{assum:bounded_noise}, then on $\cG_{\pj}$~\eqref{ineq:52} we have,
	\al{
		\br{x - \hat{x}}\cdot \br{w - \hat{w}} & = \br{x - \hat{x}}\cdot w  = x\cdot \br{w - \hat{w}}\notag \\
		& \le  \| x - \hat{x} \|\left\{ 1 \vee 		
		\sqrt{ \log^{+}\br{ \frac{\| x - \hat{x} \|}{\delta}} }		
		\vee 	
		\sqrt{ 2\log\br{\sum_{s=1}^{t}\sum_{j=1}^{p}z^{2}_{s,j}} }
		 \right\}.~\label{ineq:th_1}
		 } 
If instead $\{w_s\}$ satisfies Assumption~\ref{assum:sub_gaussian}, then~\eqref{ineq:th_1} holds on $\cG_{\pj}\cap \cG_w$. 	 
\end{theorem}
\begin{proof} 
	Let $\cj$ be a non-empty subset of $\{1,2,\ldots,p\}$. We will use $Z_{\cj}$ to denote the $t\times|\cj|$ matrix formed by the vectors $\{(z_{1,j},z_{2,j},\ldots,z_{t,j})':j\in \cj\}$. Also, for $s=1,2,\ldots,t$, let $Z_{s,\cj}$ be the $|\cj|$-dimensional column vector with components $z_{s,i}, i\in \cj$. For each non-empty subset $\cj$ of $\{1,2,\ldots,p\}$, define the following stopping-time,
$$
\tilde{\tau}_{\cj} := \inf \left\{ \ell: \sum_{s=1}^{\ell} Z_{s,\cj} Z^{'}_{s,\cj} \mbox{ is non-singular } \right\},
$$
where we let $\tilde{\tau}_{\cj} =\infty$ if the set on the r.h.s. is empty.~For times $\ell \ge \tilde{\tau}_{\cj}$ define 
\al{
V_{\ell,J}:= \br{\sum_{s=1}^{\ell} Y_{s,J} Y^{'}_{s,J}}^{-1},
}
while for $\ell < \tilde{\tau}_J$ we let $V_{\ell,J}$ be the Moore-Penrose generalized inverse of $\sum_{s=1}^{\ell} Z_{s,\cj} Z^{'}_{s,\cj}$.~Let,
\al{
X_{\ell,J} := \sum_{s=1}^{\ell} x_s Z_{s,\cj},~\mbox{ where } \ell = 1,2,\ldots,t.\label{def:1}
}
Let $\hat{x}(J)$ be the projection of $x$ onto $L(Z_J)$. We have,
\al{
\pj\br{x,L(Z_J)} = \br{Z_{1,\cj},Z_{2,\cj},\ldots,Z_{t,\cj}}' V_{t,\cj} X_{t,\cj}.
}
Thus,
\al{
	\br{x- \pj\br{x,L(Z_{\cj})}}\cdot w = \sum_{s=1}^{t} \left\{ x_s - X^{'}_{s,\cj} V_{s,\cj} Z_{s,\cj} \right\}w_s,
}
and also,
\al{
	\| x - \pj\br{x,L(Z_{\cj})}\|^{2} = \sum_{s=1}^{t} \br{x_s - X^{'}_{t,\cj} V_{t,J} Z_{s,\cj}}^2.
}
Assume that we have that there exists a $\cj$ s.t. we have $\tilde{\tau}_{\cj} < \infty$.~By using Lemma~\ref{lemma:82_2}, we obtain that the following holds on $\cG_{\pj}$~\eqref{ineq:52},
\al{
	& | \br{x-\pj\br{x,L(Z_{\cj})}}\cdot w | \notag\\
	& \le \| x - \pj\br{x,L(Z_{\cj})}\| \cdot \max \left\{ 
	\sqrt{\log^{+}\br{\frac{\| x -\pj\br{x,L(Z_{\cj})}\|}{\delta}} }  , \sqrt{ 2\log\br{\sum_{j\in \cj}\sum_{s=1}^{t}z^{2}_{s,j}} }  \right\}.\label{ineq:7}
}
For times $t>\tilde{\tau}_{\cj}$, we have $\hat{x}\ut=\pj\br{x\ut,L(Z\ut_{\cj})}$. 

This completes the proof.
\end{proof}

\begin{theorem}\label{th:proj_2}
If $\{w_s\}$ satisfies Assumption~\ref{assum:bounded_noise}, then on $\cG_{\pj}$~\eqref{ineq:52} the following holds,
	\al{
	\| \hat{w}\ut \|^2 \le p~\left\{ 1 \vee \max_{j\in\{1,2,\ldots,p\}}\log^{+}\br{\frac{ \sum_{s=1}^{t} z^2_{s,j} }{\delta}    } 
	\vee 	
	2\log\br{\sum_{s=1}^{t}\sum_{\ell=1}^{j}z^2_{s,\ell}}
	\right\},~\forall t.\label{ineq:65}
}    	
If instead $\{w_s\}$ satisfies Assumption~\ref{assum:sub_gaussian}, then~\eqref{ineq:65} holds on $\cG_w\cap \cG_{\pj}$. 	 
\end{theorem}
\begin{proof} Let $Z_{\cdot,j}$ be the $j$-th column of $Z$, and $\tilde{Z}_{\cdot,j}$ the projection of $Z_{\cdot,j}$ onto the linear space spanned by $Z_{\cdot,1},Z_{\cdot,2},\ldots,Z_{\cdot,j-1}$. We let $\tilde{Z}_{\cdot,1}=\mbf{0}$. Consider the orthogonal vectors $Z_{\cdot,1},Z_{\cdot,2}-\tilde{Z}_{\cdot,2},\ldots,Z_{\cdot,p}-\tilde{Z}_{\cdot,p}$. These span the space $L(Z)$. Since $\hat{w}$ is the projection of $w$ onto $L(Z)$, we have,
\al{
	\|\hat{w}\|^2 = \sum_{j=1}^{p} \frac{ \left\{ \br{Z_{\cdot,j} - \tilde{Z}_{\cdot,j}}\cdot w \right\}^{2}}{\| Z_{\cdot,j} - \tilde{Z}_{\cdot,j} \|^2}.\label{eq:7}
}
In case the denominator of a summand in the above is $0$, we set that term to $0$. From Theorem~\ref{th:proj_1},
\al{
&	\Big| \br{Z_{\cdot,j} - \tilde{Z}_{\cdot,j}}\cdot w \Big|  \le \| Z_{\cdot,j} - \tilde{Z}_{\cdot,j} \| \times \left\{ 1 \vee \sqrt{\log^{+}\br{\frac{\|Z_{\cdot,j} - \tilde{Z}_{\cdot,j} \|}{\delta} } } \vee 
\sqrt{2\log\br{\sum_{s=1}^{t}\sum_{m=1}^{j-1}z^2_{s,m}} }\right\} \notag\\	
	& \le \| Z_{\cdot,j} - \tilde{Z}_{\cdot,j} \| \times \left\{ 1 \vee 
	\sqrt{ \log^{+}\br{\frac{ \sum_{s=1}^{t} z^2_{s,j} }{\delta}    } }	
	 \vee 	\sqrt{2\log\br{\sum_{s=1}^{t}\sum_{m=1}^{j-1}z^2_{s,m} }}
\right\},
}
where in the last inequality we have used $\| Z_{\cdot,j} - \tilde{Z}_{\cdot,j} \|^2  \le \| Z_{\cdot,j} \|^2 
= \sum_{s=1}^{t} z^2_{s,j}$. The proof is then completed by substituting the above bound into~\eqref{eq:7}. 
\end{proof}

\begin{proposition}\label{prop:2}
If $\{w_t\}$ satisfies Assumption~\ref{assum:noise_var} and Assumption~\ref{assum:sub_gaussian}, then the folowing holds on the set $\cG_{\pj} \cap \cG_{w^2_{UB}}$,
	\al{
		& \| w\ut - \hat{w}\ut \|^2 \ge tc_1 -  \sqrt{2t c^2 \log\br{\frac{1}{\delta}}}  \notag\\
		&- p~\left\{ 1 \vee \max_{j\in\{1,2,\ldots,p\}} \log^{+}\br{\frac{ \sum_{s=1}^{t} z^2_{s,j} }{\delta}    } \vee 	2\log\br{\sum_{s=1}^{t}\sum_{m=1}^{j}z^2_{s,m}}
		\right\},~\forall t,\label{ineq:66}
	} 
 where $\cG_{w^2_{UB}}$ is as in~\eqref{def:G4}, while $\cG_{\pj}$ is as in~\eqref{ineq:52}.  
If instead $\{w_t\}$ satisfies Assumption~\ref{assum:noise_var} and Assumption~\ref{assum:bounded_noise}, then~\eqref{ineq:66} holds on $\cG_{\pj} \cap \cG_{w^2_{B}}$. The sets $\cG_{w^2_{B}}$ and $\cG_{w^2_{UB}}$ are as in Lemma~\ref{lemma:bound_var} and Lemma~\ref{lemma:unbound_var} respectively.
\end{proposition}
\begin{proof} We have,
\al{
	\| w\ut - \hat{w}\ut \|^2 & = \|w\ut\|^2 - \|\hat{w}\ut\|^2 \notag\\
	& \ge tc_1 -  \sqrt{2t \br{ \frac{1}{c}\log \br{\frac{T}{\eps'}} }^2 \log\br{\frac{1}{\delta}}}  - \|\hat{w}\ut\|^2 \notag\\
	& \ge t c_1 -  \sqrt{2t \br{ \frac{1}{c}\log \br{\frac{T}{\eps'}} }^2 \log\br{\frac{1}{\delta}}}\notag\\
	& - p~\left\{ 1 \vee \max_{j\in\{1,2,\ldots,p\}}\log^{+}\br{\frac{ \sum_{s=1}^{t} z^2_{s,j} }{\delta}    } \vee 	2\log\br{\sum_{s=1}^{t}\sum_{m=1}^{j}z^2_{s,m}}
	\right\},
}
where the first inequality follows from Lemma~\ref{lemma:unbound_var},~while the second follows from Theorem~\ref{th:proj_2}. Both the inequalities hold on h.p. sets.
\end{proof} 

\begin{proposition}\label{prop:3}
	Define,
	\al{
		r := x - \hat{x}.	\label{def:r}
	}
Let $\hat{\hat{r}}$ be the projection of $r$ onto $L(v+w)$, and $\hat{\hat{x}}$ the projection of $x$ onto $L(Z,v+w)$.~Then,
	\al{
		\| x - \hat{x}\|^2 = \| \hat{\hat{r}} \|^2 + \| x - \hat{\hat{x}}\|^2.\label{eq:8}
	}
	Also,
	\al{
		\| \hat{\hat{r}} \|^2 = \frac{\left[r\cdot \br{v+w-\hat{v}-\hat{w}}\right]^2}{\| v+w-\hat{v}-\hat{w}\|^2}.~\label{eq:9}
	}
	
\end{proposition}
\begin{proof} We clearly have,
\al{
	\hat{\hat{x}} = \hat{x} + \hat{\hat{r}},
}
or equivalently, 
\al{
	r = x - \hat{x} = \hat{\hat{r}} + (x-\hat{\hat{x}}).
}
Since $\hat{\hat{r}} $ and $x-\hat{\hat{x}}$ are orthogonal, we get
\al{
	\| x - \hat{x} \|^2= \|\hat{\hat{r}}\|^2 + \| x-\hat{\hat{x}}\|^2.
}
This proves~\eqref{eq:8}.

Since $r(=x-\hat{x})$ is orthogonal to $z$, and $\hat{v}+\hat{w}$ is the projection of $v+w$ onto $L(Z)$, we have that $\hat{\hat{r}}$ is also equal to the projection of $r$ onto $L(v+w-\hat{v}-\hat{w})$. Hence,
\al{
	\| \hat{\hat{r}} \|^2 = \frac{r\cdot \br{v+w-\hat{v}-\hat{w}}}{\| v+w-\hat{v}-\hat{w}\|^2},
}
which proves~\eqref{eq:9}.

\end{proof}

\begin{proposition}\label{prop:4}
	Let $v\ust$ denote the projection of $v$ onto $L(Z,x)$.~Then,
	\al{
		& \|r\|^2 \| v -  \hat{v} + w - \hat{w}\|^2 - \Big| r \cdot \br{v+w-\hat{v}-\hat{w}}	\Big|^2\notag\\
		&=\|r\|^2\left\{\|v-v\ust\|^2 + \| w - \hat{w}\|^2+2\br{v-v\ust}\cdot w\right\} - \br{r\cdot w}^2,
	}
where $r = x - \hat{x}$~\eqref{def:r}.
\end{proposition}
\begin{proof} Since $v-\hat{v}$ is orthogonal to $L(Z)$, its projection on $L(Z,x)$ is the same as its projection onto $L(x-\hat{x})$. Suppose that this projection is equal to $a(x-\hat{x})=a~r$. Then we have,
\al{
	v - \hat{v} = a~r + \br{v-v\ust},\label{eq:10}
}
where the second quantity in the r.h.s. above is the component that is orthogonal to $L(Z,x)$. Since the vectors $r$ and $v-v\ust$ are orthogonal,
\al{
	\|v - \hat{v}\|^2 = a^2~\|r\|^2 + \|v-v\ust\|^2.\label{eq:11}
}
Upon taking dot product with the vector $r$ on both sides of~\eqref{eq:10}, we get,
\al{
	r \cdot \br{v-\hat{v}} = a~\|r\|^2.
} 
This gives,
\al{
	r\cdot \br{v-\hat{v}+w-\hat{w}} = a \|r\|^2 + r\cdot (w-\hat{w}).
}
Note that since $r = x-\hat{x}$, it is orthogonal to $Z$, hence it is also orthogonal to $\hat{w}$ so that we have $r\cdot \hat{w}=0$. Upon substituting this into the above relation, we get
\al{
	r\cdot \br{v-\hat{v}+w-\hat{w}} = a \|r\|^2 + r\cdot w.
}
Taking squares on both sides,
\al{
	\left[r\cdot \br{v-\hat{v}+w-\hat{w}}\right]^2 = a^2 \|r\|^4 + \br{r\cdot w}^2 + 2a\|r\|^2 \br{r\cdot w}.\label{eq:13}
}

Now,
\al{
	\|(v-\hat{v}) + (w-\hat{w})\|^2 & = \|v-\hat{v}\|^2 + \|w-\hat{w}\|^2 + 2\br{v-\hat{v}}\cdot \br{w-\hat{w}}\notag\\
	&  = \|v-\hat{v}\|^2 + \|w-\hat{w}\|^2 + 2\br{v-\hat{v}}\cdot w\notag\\
	& = a^2~\|r\|^2 + \|v-v\ust\|^2+ \|w-\hat{w}\|^2 + 2\br{v-\hat{v}}\cdot w\notag\\
	& = a^2~\|r\|^2 + \|v-v\ust\|^2+ \|w-\hat{w}\|^2 + 2a~r\cdot w+ \br{v-v\ust}\cdot w,\label{eq:12}
}
where the second equality follows since $v-\hat{v}$ is orthogonal to $L(Z)$, and hence $\br{v-\hat{v}}\cdot \hat{w}=0$, the third follows from~\eqref{eq:11}, and the last follows from~\eqref{eq:10}.~Upon multiplying
~\eqref{eq:12} by $\|r\|^2$ and subtracting~\eqref{eq:13} from it, we get,
\al{
	&\|r\|^2\|(v-\hat{v}) + (w-\hat{w})\|^2 -\left[r\cdot \br{v-\hat{v}+w-\hat{w}}\right]^2\notag\\
	&= 
	\|r\|^2\left\{\|v-v\ust\|^2 + \| w - \hat{w}\|^2+2\br{v-v\ust}\cdot w \right\} - \br{r\cdot w}^2.
}
This completes the proof.
\end{proof}

\begin{theorem}\label{th:proj_3}
If $\{w_t\}$ satisfies Assumption~\ref{assum:sub_gaussian}, then on $\cG_{\pj}$~\eqref{ineq:52} we have the following for all $t$,
	\al{
		& \|x\ut-\hat{\hat{x}}\ut\|^2 \ge \notag \\
		& \frac{\|r\|^2\left\{\|v-v\ust\|^2 + \| w - \hat{w}\|^2\right\} -  \| x - \hat{x} \|^2 
			\left\{ 1 \vee \sqrt{ \log^{+}\br{\frac{\|x - \hat{x} \|}{\delta}} }
			 \vee 
			 \sqrt{2\log\br{\sum_{s=1}^{t}\sum_{j=1}^{p}z^2_{s,j}} } \right\}}{\|\br{v-\hat{v}}\|^2 + \| w-\hat{w}\|^2 + 2 \| v-\hat{v} \|\left\{ 1 \vee \sqrt{ \log^{+}\br{\frac{\| v-\hat{v} \|}{\delta}} } \vee \sqrt{2\log\br{\sum_{s=1}^{t}\sum_{j=1}^{p}z^2_{s,p}} }  
			\right\}},
	}   
	where $r= x\ut -\hat{x}\ut$, $\hat{\hat{x}}\ut$ is the projection of $x\ut$ onto $L(Z\ut,v\ut+w\ut)$, and $(v\ut)\ust$ is the projection of $v\ut$ onto $L(Z\ut,x\ut)$. If the assumption on $\{w_t\}$ is replaced by Assumption~\ref{assum:sub_gaussian}-(i), then the same conclusion holds on $\cG_{\pj}$~\eqref{ineq:52}.
	
\end{theorem}
\begin{proof} We note that,
\al{
	\|x-\hat{\hat{x}}\|^2 & = \|r\|^2 - \|\hat{\hat{r}} \|^2\notag\\
	& = \frac{\|r\|^2\| v+w-\hat{v}-\hat{w}\|^2-\left[r\cdot \br{v+w-\hat{v}-\hat{w}}\right]^2}{\| v+w-\hat{v}-\hat{w}\|^2},\label{eq:14}
}
where the first equality follows from~\eqref{eq:8}, while the second one follows from Proposition~\ref{prop:3}. 

Next, we derive an upper-bound on the denominator of the above expression. We have,
\al{
	\| \br{v-\hat{v}}	+ \br{w-\hat{w}}\|^2 &= \|\br{v-\hat{v}}\|^2 + \| w-\hat{w}\|^2+2\br{v-\hat{v}}\cdot \br{ w-\hat{w}}\notag\\
	&\le \|\br{v-\hat{v}}\|^2 + \| w-\hat{w}\|^2\notag \\
	& + 2 \| v-\hat{v} \|\left\{ 1 \vee \log^{+}\br{\frac{\| v-\hat{v} \|}{\delta} } \vee 2\log\br{\sum_{s=1}^{t}\sum_{j=1}^{p}z^2_{s,j}}
\right\}^{1\slash 2},\label{ineq:3}
}
where the inequality follows from Theorem~\ref{th:proj_1}. From Proposition~\ref{prop:4}, the numerator in~\eqref{eq:14} can be bounded as follows,
\al{
	& \|r\|^2\| v+w-\hat{v}-\hat{w}\|^2-\left[r\cdot \br{v+w-\hat{v}-\hat{w}}\right]^2 \notag\\
	& = \|r\|^2\left\{\|v-v\ust\|^2 + \| w - \hat{w}\|^2+2\br{v-v\ust}\cdot w \right\} - \br{r\cdot w}^2.\label{eq:15}
}
The terms $ \br{r\cdot w}^2$ and $\Big| \br{v-v\ust}\cdot w \Big|$ can be bounded using Theorem~\ref{th:proj_1} as follows,
\al{
	\br{r\cdot w}^2 & = \left[\br{x - \hat{x}}\cdot w \right]^2 \notag\\
	& \le  \| x - \hat{x} \|^2 \left\{ 1 \vee \log^{+}\br{\frac{\| x - \hat{x} \|}{\delta }} \vee 2\log\br{\sum_{s=1}^{t}\sum_{j=1}^{p}z^2_{s,j}}
	 \right\},
}
and,
\al{
	\Big| \br{v-v\ust}\cdot w \Big| \le \| v-v\ust \|\left\{ 1 \vee \log^{+}(\| v-v\ust \|) \vee 2\log\br{\sum_{s=1}^{t}\sum_{j=1}^{p}z^2_{s,j} + \sum_{s=1}^{t}x^2_s} \right\}^{1\slash 2}.
}
Upon substituting these into~\eqref{eq:15}, the numerator in~\eqref{eq:14} can be lower-bounded as follows,
\al{
	& \|r\|^2\| v+w-\hat{v}-\hat{w}\|^2-\left[r\cdot \br{v+w-\hat{v}-\hat{w}}\right]^2 \notag\\
	& \ge  \|r\|^2\left\{\|v-v\ust\|^2 + \| w - \hat{w}\|^2\right\} -  \| x - \hat{x} \|^2 \left\{ 1 \vee \log^{+}
	\br{\| x - \hat{x} \|} \vee 2\log\br{\sum_{s=1}^{t}\sum_{j=1}^{p}z^2_{s,j}}
 \right\}.~\label{ineq:8}
}
Substituting the bounds ~\eqref{ineq:3} and~\eqref{ineq:8} into~\eqref{eq:14}, we get
\al{
	& \|x-\hat{\hat{x}}\|^2 \ge \notag\\
	& \frac{\|r\|^2\left\{\|v-v\ust\|^2 + \| w- \hat{w}\|^2\right\} -  \| x - \hat{x} \|^2 \left\{ 1 \vee \log^{+}
		\br{\frac{\| x - \hat{x} \|}{\delta}}	
		\vee 2\log\br{\sum_{s=1}^{t}\sum_{j=1}^{p}z^2_{s,j}}  \right\}}{\|\br{v-\hat{v}}\|^2 + \| w-\hat{w}\|^2 + 2 \| v-\hat{v} \|\left\{ 1 \vee \log^{+}\br{ \frac{\| v-\hat{v} \|}{\delta} } \vee 2\log\br{\sum_{s=1}^{t}\sum_{j=1}^{p}z^2_{s,j}} \right\}^{1\slash 2}}.
}   
This completes the proof. 
\end{proof} 

	

\subsection{Auxiliary Results}
Recall that $\{w_s\}$ is a martingale difference sequence w.r.t. $\{\cF_s\}$. For $\cj$ a non-empty subset of $\{1,2,\ldots,p\}$, $\tilde{\tau}_{\cj}=\inf \left\{\sum_{s=1}^{\ell}Z_{i,\cj} Z'_{i,\cj} \mbox{ is non-singular }\right\}$, and for $s\ge \tilde{\tau}_{\cj}$
\al{
	V_{s,\cj} = \br{\sum_{k=1}^{s} Z_{k,\cj} Z^{'}_{k,\cj}}^{-1}.\label{def:v_n}
}
Recall $X_t = \sum_{s=1}^{t} x_s Z_{s,\cj}$. Define,
\al{
	s_t :&= \sum_{s=1}^{t} \br{x_s - X^{'}_t V_{t,\cj} Z_{s,\cj}}^2,\label{def:sn}\\
	d_t :&= x_t - X^{'}_t V_{t,\cj} Z_{t,\cj}, and
}
\al{
& \cG_{\pj}:=\{\omega: \mbox{ (1, 2) below hold}\}, \mbox{ where}\notag\\	
& 1)  \sum_{k=1}^{s} \frac{Z^{'}_{k,\cj} V_{k-1} \br{\sum_{i=1}^{k-1}Z_{i,\cj} w_i } w_k}{1+ Z^{'}_{k,\cj} V_{k-1}Z_k} \le \sqrt{\sum_{k=1}^{s} \frac{\br{ Z^{'}_{k,\cj} V_{k-1} \sum_{i=1}^{k-1}Z_{i,\cj} w_i }^2 }{\br{1+ Z^{'}_{k,\cj} V_{k-1}Z_{k,\cj}}^2}     }\notag\\
& \times \sqrt{\log\br{\frac{1}{\delta} \sum_{k=1}^{s} \frac{\br{ Z^{'}_{k,\cj} V_{k-1} \sum_{i=1}^{k-1}Z_{i,\cj} w_i }^2 }{\br{1+ Z^{'}_{k,\cj} V_{k-1}Z_{k,\cj}}^2}  }}, \forall \cj,~s\mbox{ and }\notag\\
& 2) \sum_{s=1}^{t} d_s w_s  \le \sqrt{ \br{\sum_{s=1}^{t} d^{2}_s} \log \br{\frac{\sum_{s=1}^{t} d^{2}_s}{\delta}} }~\forall \cj \subseteq \{1,2,\ldots,p\} ,\quad \forall t
\Bigg\}\notag\\
&3) \Big|\sum_{s=1}^{t} x_s w_s \Big|  \le \sqrt{\sum_s x^2_s} \sqrt{\log \br{\frac{\sum_s x^2_s}{\delta}}},~\forall t.
\label{ineq:52}
}

\begin{lemma}

	\al{
		\bP\br{\cG_{\pj}^c}\le 3\delta.\label{ineq:5}
	}
\end{lemma}
\begin{proof}
	Follows from the self-normalization bound.~\eqref{eq:self_normalized}
\end{proof}

\begin{lemma}\label{lemma:82-1}
If $\{w_s\}$ satisfies either Assumption~\ref{assum:bounded_noise} or Assumption~\ref{assum:sub_gaussian}, then on $\cG_{\pj}$ \eqref{ineq:52} we have,	
	\al{
		\sum_{k=1}^{s}\frac{\br{Z^{'}_{k,\cj} V_{k-1}\sum_{i=1}^{k-1}Z_{i,\cj} w_i }^2}{1+ Z'_{k,\cj} V_{k-1} Z_{k,\cj}} \le 				
		2 \sigma^2_w	\log\br{\lm_{\max}\br{\sum_{k=1}^{s} Z_{k,\cj} Z'_{k,\cj} }  }+8\log\br{\frac{1}{\delta}} +\log 8~\forall s.
	}
\end{lemma}
\begin{proof} 
	The following is essentially (2.17) of~\cite{lai1982least},
\al{
	&\br{\sum_{i=1}^{s} Z^{'}_{i,\cj} w_i}V_{s} \br{\sum_{i=1}^{s} Z_{i,\cj} w_i} + \sum_{k=1}^{s} \frac{\br{ Z^{'}_{k,\cj} V_{k-1} \sum_{i=1}^{k-1}Z_{i,\cj} w_i }^2}{1+Z^{'}_{k,\cj} V_{k-1}Z_{k,\cj}}\notag\\
	&= \sum_{k=1}^{s} Z^{'}_{k,\cj}V_k Z_{k,\cj} ~w^2_k+ 2 \sum_{k=1}^{s} \frac{Z^{'}_{k,\cj} V_{k-1} \br{\sum_{i=1}^{k-1}Z_{i,\cj} w_i } w_k}{1+ Z^{'}_{k,\cj} V_{k-1}Z_{k,\cj}}.
}
Since $V_s$ is positive semi-definite, this yields,
\al{
	&\sum_{k=1}^{s} \frac{\br{ Z^{'}_{k,\cj} V_{k-1} \sum_{i=1}^{k-1}Z_{i,\cj} w_i }^2}{1+Z^{'}_{k,\cj} V_{k-1}Z_{k,\cj}}\notag\\
	&\le \sum_{k=1}^{s} Z^{'}_{k,\cj}V_k Z_{k,\cj} ~w^2_k+ 2 \sum_{k=1}^{s} \frac{Z^{'}_{k,\cj} V_{k-1} \br{\sum_{i=1}^{k-1}Z_{i,\cj} w_i } w_k}{1+ Z^{'}_{k,\cj} V_{k-1}Z_{k,\cj}}.\label{ineq:4}
}
We will derive upper bounds on both the terms in the r.h.s. above. For the first term we have,
\al{
	\sum_{k=1}^{s} Z^{'}_{k,\cj}V_k Z_{k,\cj} ~w^2_k = \sum_{k=1}^{s} Z^{'}_{k,\cj}V_k Z_{k,\cj} ~\sigma^2_w +  \sum_{k=1}^{s} Z^{'}_{k,\cj}V_k Z_{k,\cj} \br{ w^2_k - \bE\left\{ w^2_k|\cF_{k-1}\right\}    }
	.\label{ineq:53}
} 
Now, 
\al{
 \sum_{k=1}^{s} Z^{'}_{k,\cj}V_k Z_{k,\cj} ~\sigma^2_w \le \sigma^2_w	\log\br{\lm_{\max}\br{\sum_{k=1}^{s} Z_{k,\cj} Z'_{k,\cj} }  }.\label{ineq:85}
}
To bound the second term on the r.h.s. of~\eqref{ineq:4}, we note that from the definition of $\cG_{\pj}$,
\al{
& \sum_{k=1}^{s} \frac{Z^{'}_{k,\cj} V_{k-1} \br{\sum_{i=1}^{k-1}Z_{i,\cj} w_i } w_k}{1+ Z^{'}_{k,\cj} V_{k-1}Z_k} \le
\sqrt{\sum_{k=1}^{s} \frac{\br{ Z^{'}_{k,\cj} V_{k-1} \sum_{i=1}^{k-1}Z_{i,\cj} w_i }^2 }{\br{1+ Z^{'}_{k,\cj} V_{k-1}Z_{k,\cj}}^2}     }\notag\\
& \times \sqrt{\log\br{\frac{1}{\delta} \sum_{k=1}^{s} \frac{\br{ Z^{'}_{k,\cj} V_{k-1} \sum_{i=1}^{k-1}Z_{i,\cj} w_i }^2 }{\br{1+ Z^{'}_{k,\cj} V_{k-1}Z_{k,\cj}}^2}  }}\notag\\
& \le \sqrt{\sum_{k=1}^{s} \frac{\br{ Z^{'}_{k,\cj} V_{k-1} \sum_{i=1}^{k-1}Z_{i,\cj} w_i }^2 }{\br{1+ Z^{'}_{k,\cj} V_{k-1}Z_{k,\cj}}}     }\sqrt{\log\br{\frac{1}{\delta} \sum_{k=1}^{s} \frac{\br{ Z^{'}_{k,\cj} V_{k-1} \sum_{i=1}^{k-1}Z_{i,\cj} w_i }^2 }{\br{1+ Z^{'}_{k,\cj} V_{k-1}Z_{k,\cj}}}  }}.\label{ineq:86}
}
Upon substituting~\eqref{ineq:85},~\eqref{ineq:86} into~\eqref{ineq:4}, we get the following relation,
\al{
& \sum_{k=1}^{s}\frac{\br{Z^{'}_{k,\cj} V_{k-1}\sum_{i=1}^{k-1}Z_{i,\cj} w_i }^2}{1+ Z'_{k,\cj} V_{k-1} Z_{k,\cj}}  \le \sigma^2_w	\log\br{\lm_{\max}\br{\sum_{k=1}^{s} Z_{k,\cj} Z'_{k,\cj} }  }\\
& + \sqrt{\sum_{k=1}^{s}\frac{\br{Z^{'}_{k,\cj} V_{k-1}\sum_{i=1}^{k-1}Z_{i,\cj} w_i }^2}{1+ Z'_{k,\cj} V_{k-1} Z_{k,\cj}}  \log\br{\frac{\sum_{k=1}^{s}\frac{\br{Z^{'}_{k,\cj} V_{k-1}\sum_{i=1}^{k-1}Z_{i,\cj} w_i }^2}{1+ Z'_{k,\cj} V_{k-1} Z_{k,\cj}} }{\delta}}}.
} 
The proof is then completed by  algebraic manipulations.
\end{proof}

\begin{lemma}\label{lemma:82_2}
	If $\{w_s\}$ satisfies either Assumption~\ref{assum:bounded_noise} or Assumption~\ref{assum:sub_gaussian}, then, on $\cG_{\pj}$~\eqref{ineq:52}, we have the following bound,
	\al{
		\Big| \sum_{s=1}^{t} & \br{x_s - X^{'}_tV_t Z_{s,\cj}}w_s \Big|   \le  \br{s_t}^{1\slash 2}\notag\\	
		&\times \max \left\{
		\sqrt{ \log \br{\frac{s_t}{\delta}}},\left[2B^2_w	\log\br{\lm_{\max}\br{\sum_{s=1}^{t} Y_s Y'_s }  }\right]^{1\slash 2} \right\},~\forall t.
	}
\end{lemma}
\begin{proof} 
	The following results are essentially Lemma~3 of~\cite{lai1979strong},
	\al{
		\sum_{s=\tilde{\tau}_{\cj}+1}^{t}\br{ x_s - X^{'}_t V_t Z_{s,\cj} }w_s & = \sum_{s=\tilde{\tau}_{\cj}+1}^{t} d_s \left\{ w_s - Z^{'}_{s,\cj} V_{s-1} \br{\sum_{j=1}^{s-1}Z_{j,\cj} w_j} \right\},\label{eq:20}\\
		s_t & = s_{\tilde{\tau}_{\cj}} + \sum_{s=\tilde{\tau}_{\cj}+1}^{t} d^2_s \br{1+Z^{'}_{s,\cj} V_{s-1} Z_{s,\cj}}.\label{eq:26}
	}
We now bound each term on the r.h.s. of~\eqref{eq:20} separately. On $\cG_{\pj}$, by definition we have the following bound on the first term,
\al{
\Big|\sum_s d_s w_s \Big| & \le \sqrt{ \br{\sum_{s=1}^{t} d^{2}_s} \log \br{\frac{\sum_{s=1}^{t} d^{2}_s}{\delta}} }\notag\\
& 	\le \sqrt{ s_t \log \br{\frac{s_t}{\delta}} },\label{ineq:54}
}
where the second inequality follows from~\eqref{eq:26}.~Using the Cauchy-Schwartz inequality,
	\al{
	&	\Bigg| \sum_{s=\tilde{\tau}_{\cj}+1}^{t} d_s Z^{'}_{i,\cj} V_{s-1} \br{\sum_{j=1}^{s-1} Z_{j,\cj} w_j } \Bigg| \notag \\
	&\le \left\{\sum_{s=\tilde{\tau}_{\cj}+1}^{t} d^2_s \br{1+Z^{'}_{s,\cj} V_{s-1} Z_{s,\cj}}  \right\}^{1\slash 2} \cdot \left\{ \sum_{s=\tilde{\tau}_{\cj}+1}^{t}\frac{\br{ Z^{'}_{s,\cj} V_{s-1} \sum_{j=1}^{s-1} Z_{j,\cj} w_j}^{2}}{\br{1+ Z^{'}_{s,\cj} V_{s-1} Z_{s,\cj} }}  \right\}^{1\slash 2}\notag\\
		& \le \br{s_n}^{1\slash 2}\cdot \left\{ \sum_{s=\tilde{\tau}_{\cj}+1}^{t}\frac{\br{ Z^{'}_{s,\cj} V_{s-1} \sum_{j=1}^{s-1} Z_{j,\cj} w_j}^{2}}{\br{1+ Z^{'}_{i,\cj} V_{s-1} Z_{i,\cj} }}  \right\}^{1\slash 2}\notag\\
		&\le \br{s_n}^{1\slash 2} \cdot \left[2 \sigma^2_w	\log\br{\lm_{\max}\br{\sum_{k=1}^{s} Z_{k,\cj} Z'_{k,\cj} }  }+8\log\br{\frac{1}{\delta}} +\log 8\right]^{1\slash 2},\label{ineq:55}
	}
	where the last inequality follows from Lemma~\ref{lemma:82-1}. The proof is completed by substituting~\eqref{ineq:54},~\eqref{ineq:55} into~\eqref{eq:20}.
\end{proof}

Under Assumption~\ref{assum:sub_gaussian} the process $\{w_t\}$ is conditionally sub-Gaussian, and hence the process $|w^2_t - \bE \br{w^2_{t-1}|\cF_{t-1} }|$ is sub-exponential~\cite{vershynin2018high}, i.e. we have
	\al{
		\bP\br{|w^2_t - \bE \br{w^2_{t-1}|\cF_{t-1} } |> x} \le \exp\br{-cx},~\forall x>0,~\mbox{ for some }c>0.\label{def:c}
	}
\begin{lemma}\label{lemma:unbound_var}
	Define,
	\al{
		\cG_{w^2_{UB}} := \left\{	\sum_{s=1}^{t} w^2_s \ge  c_1 t -  \sqrt{2t \br{ \frac{1}{c}\log \br{\frac{T}{\eps'}} }^2 \log\br{\frac{1}{\delta}}},~\forall t=1,2,\ldots,T \right\}.\label{def:G4}
	}
Let $\{w_t\}$ satisfy Assumption~\ref{assum:noise_var} and Assumption~\ref{assum:sub_gaussian}.~Then, 
	\al{
		\bP\br{	\cG_{w^2_{UB}} } \ge 1-\br{\delta + \eps'}.
	}
\end{lemma}
\begin{proof}
	Using the union bound on individual increments and~\eqref{def:c}, we obtain that the following occurs w.p. less than $\eps'$,
	\al{
		\bP\br{	\exists s \in \{1,2,\ldots,T\}~\mbox{ s.t. } \Big| w^2_s - \bE\br{w^2_s|\cF_{s-1}}\Big| > \frac{1}{c}\log \br{\frac{T}{\eps'}} }\le \eps',\label{ineq:72}
	}
	where $\eps'>0$.

	Upon letting $B= \frac{1}{c}\log \br{\frac{T}{\eps'}}$ in Theorem~\ref{th:azuma_unbounded}, and applying~\eqref{ineq:72}, the union bound over $t$, and letting $\eps' \leftarrow \frac{\eps}{T},~\delta \leftarrow \frac{\delta}{T}$, we get,
		\al{
		\bP\br{	\left\{	\sum_{s=1}^{t} \left\{w^2_s- \bE\br{w^2_s|\cF_{s-1}} \right\} < \sqrt{2t \br{ \frac{1}{c}\log \br{\frac{T}{\eps'}} }^2 \log\br{\frac{1}{\delta}}},~\forall t=1,2,\ldots,T	\right\}}	\le \delta + \eps'.
	}
The proof is then completed by noting that $\bE\br{w^2_t|\cF_{t-1}} >c_1$.
\end{proof}

\begin{lemma}\label{lemma:bound_var}
	Let $\{w_t\}$ satisfy Assumption~\ref{assum:noise_var} and Assumption~\ref{assum:sub_gaussian}. Define,
	\al{
		\cG_{w^2_B} := \left\{	\sum_{s=1}^{t} w^2_s \ge  c_1 t -  \sqrt{2t \br{2B_3}^2 \log\br{\frac{1}{\delta}}}  \right\}.\label{def:G_tilde4}
	}
	Then, 
	\al{
		\bP\br{	\cG_{w^2_B}  } \ge 1-\delta.
	}	
\end{lemma}
\begin{proof} Follows from Azuma-Hoeffding~\ref{th:azuma_unbounded} after noting that $|  w^2_t - \bE\br{ w^2_t | \cF_{t-1}} |$ is bounded by $2B_w$.
\end{proof}
\section{Duration of the first exploratory episode $H_1(\Theta,\eps)$} \label{sec:first_epi}
We set
\al{
H_1(\Theta,\eps)  =  t\ust_1(\rho) \vee t\ust_2(\rho) \vee t\ust_3(\rho)\vee t\ust_5(\eps_1) \vee t\ust_6(\eps_3,\delta).
}
Though these have been defined earlier, we repeat these for convenience,
 \al{
		t\ust_1(\rho) = \inf \left\{ t\in \bN: B_2 \frac{\log N\ui_t}{\sqrt{N\ui_t}} + \cE(t;\te\ust,\delta) < 1, C_1 \rho^{t} \| Y_{0}\| < B_u \right\},
	}	 

	\al{
		&		t\ust_2(\rho) = \inf\Big\{t\in\bN: \notag\\
		&						\left[ B_2 \frac{\log N^{(\cI)}_t}{\sqrt{N^{(\cI)}_t}} + \cE(t;\te\ust,\delta)\right] \cdot \left[b_1\br{p+ 1 + \frac{C_1}{1-\rho}  \left\{1  +  \sum_{\ell=1}^{q} |b_\ell|   \right\} }\right] \le \frac{\delta^2_1}{2},\notag\\
		&\qquad C_1 \rho^{t} \| Y_{0}\| < \delta_1 B_u\Big\}.
	}
 \al{
t\ust_3 = \inf \left\{t \in\bN: \cE(t;\te\ust,\delta) \le \frac{b_1}{2},\mbox{ and }\frac{B_2 \sqrt{\log N\ui_{\ell}}}{2}\ge 2\right\}.
}
\al{
	t\ust_5(\eps_1) & = \frac{2\br{C_1\| Y_{0}\| }^2+2\br{\frac{B_u C_1}{1-\rho}  \left\{1+  \sum_{\ell=1}^{q} |b_\ell|  \right\}}^2 + q B^2_u}{\min\left\{\frac{\sigma^2_e}{2q} \min_{\ell \in \{1,2,\ldots,p\} }  \beta^{\ell\slash 2}_1, \frac{c_1}{4} \min_{\ell \in \{1,2,\ldots,p\}} \beta^{\ell\slash 2}_2\right\}\eps_1},~\eps_1 >0\\
	t\ust_6(\eps_3,\delta) & = \inf \left\{t \in\bN: \cE(t;\te\ust,\delta) \le b_1 \eps_3 \right\},~\eps_3>0.
 }
\section{Choosing $B_u$, the threshold for clipping inputs}\label{sec:threshold}
We begin with few definitions.
\begin{definition}\label{def:delta_1}
	Let $\te\in\bR^{p+q}$ be a possible parameter associated with ARX~\eqref{def:arx}. Let $(\delta_1(\te),B_u(\te))$ be a tuple that satisfies the following set of inequalities,
	\al{
		\br{\|\lm(\te)\|+1} \br{p \delta_1(\te)^2 +q \delta_1(\te) } &< 1,\label{cond:delta_1}	\\
		\frac{1}{3}\left[ \frac{C_1(\te)}{b_1(\te)\br{1-\rho(\te)}}  \left\{ \frac{1}{2B_u(\te)} + \sum_{\ell=1}^{p}| a_\ell(\te) |   \right\}\right]^{-1} & \ge \delta_1(\te),\label{cond:delta_3}\\
		\br{\|\lm(\te)\|+1} \left[\br{1 +  \frac{C_1(\te)}{1-\rho(\te)} \left\{1+  \sum_{\ell=1}^{q} |b_\ell(\te)|  \right\}}+q\right]\delta_1(\te) & < 1.\label{cond:delta_4}
	}
	\al{
		\frac{B_w}{B_u(\te)} &\le \frac{\delta_1(\te)}{2},~\frac{B_w}{B_u(\te)} \le  \frac{\delta_1(\te)}{\br{	1+ \frac{C_1(\te)}{1-\rho(\te)} \left\{1+  \sum_{\ell=1}^{q} |b_\ell(\te)|  \right\}}},\label{cond:B1_B3_4_s}\\
		\frac{B_w}{B_u(\te)} &\le  \delta_1(\te),\label{cond:B1_B3_4}\\
		 B_u(\te)&> 1.
	}
	Here $C_1(\te),\rho(\te)$ are as in Lemma~\ref{lemma:state_bound}.~When~\eqref{cond:delta_1}-\eqref{cond:B1_B3_4} are required to hold for every $\te\in\Theta$, denote a solution to those inequalities by $(\delta_1(\Theta),B_u(\Theta))$ and denote $\delta_1 = \delta_1(\Theta),B_u = B_u(\Theta)$.
\end{definition}
\emph{Obtaining $\delta_1,B_u$}: Our interest will be in obtaining a $(\delta_1(\Theta),B_u(\Theta))$.~A solution to the above set of inequalities can be found using the following set of simplified inequalities. It can be verified that a solution to these inequalities also satisfies the above set of inequalities.~Define
\al{
	M(\Te) := \sup_{\te\in\Theta} \frac{C_1(\te)}{1-\rho(\te)}  \left\{1+  \sum_{\ell=1}^{q} |b_\ell(\te)|  \right\}.\label{def:M_th}
}

We use 
\al{
	B_u = \frac{B_w}{\delta^2_1}\cdot  \br{1+M(\Theta)},\label{eq:24}
}
where $\delta_1$ satisfies the following inequalities,
\al{
	\delta_1 &\le \frac{1}{(p+q) \sup_{\te\in\Theta} \br{1+ \|\lm(\te)\| } }\label{deltaineq_1}\\
	\delta_1 & \le    	 \frac{1}{3}  \inf_{\te\in\Theta}\left[\frac{b_1(\te)\br{1-\rho(\te)}}{C_1(\te)}\right]
	\left[   \frac{\delta_1^2}{2B_w M(\Te)} +\sup_{\te\in\Theta} \sum_{\ell=1}^{p}| a_\ell(\te) | \right]^{-1} \label{deltaineq_2}\\
	1 &\ge		\br{\sup_{\te\in\Theta}\|\lm(\te)\|+1} \left[M(\Te)+q\right] \delta_1 \label{deltaineq_3}
}

\section{Unbounded Noise Case}\label{sec:unbounded_noise}
Recall that we assumed the following holds,
\al{
		\sup_{t} \bE\left\{\exp \br{\gamma |w_t |} \Big| \cF_{t-1} \right\} \le \exp\br{\gamma^2\sigma^2 \slash 2}, a.s.   ~\forall t.
	}
Since the noise is not bounded, we will restrict our analysis to the following set.
\begin{lemma}\label{lemma:noise_bound}
	Define,
	\al{
		\cG_w := \Big\{ |w_t| \le \sigma\sqrt{\log\br{\frac{T}{\delta}}},~\forall t=1,2,\ldots,T \Big\}
		,	\label{def:g1}
	}
	where $\sigma>0$.~Then,
	\al{
		\bP\br{\cG^c_w} \le \delta.
	}
\end{lemma}
\begin{proof}
    It follows from Chernoff bound that $\bP\br{|w_t|>x} \le \exp\br{x^2\slash (2\sigma^2)}$. The proof then follows by letting $x= \sqrt{\log\br{\frac{T}{\delta}}}$, and using union bound for $t=1,2,\ldots,T$.
\end{proof}
Define,
\al{
		B_w(T) := \sigma\sqrt{\log\br{\frac{T}{\delta}}}.\label{def:B3}
	}	
Since unlike the bounded noise case, in which we had $|w_t|\le B_w,~\forall t$, there is no upper-bound on the noise values, the quantity $B_w(T)$ serves as a high-probability upper-bound. Indeed, most of the results derived under the Assumption~\ref{assum:bounded_noise} continue to hold under Assumption~\ref{assum:sub_gaussian} upon replacing $B_w$ by $B_w(T)$. Since the analysis of regret for sub-Gaussian noise closely follows that of bounded noise, we will only highlight the differences between the two.

We begin with the lower-bound on $\lm_{\min}(V_t)$ that was derived in Theorem~\ref{th:1}. Notice that $\beta_3$ involves $B^2_w$ in the denominator, and hence after replacing it by $B_w(T)$, we have that $\beta_3 \propto \frac{1}{\lesssim \log\br{T\slash \delta}}$, and hence decays with time-horizon $T$. In order to compensate for this, the algorithm explores more often, so that we let the episode duration $H$ be equal to $1\slash \beta_3$, and let number of episodes until $t$ be $\log t$, which yields $N\ui_t \approx \frac{\log t}{\beta_3}$. With this change, the high-probability bound on the estimation error derived in Theorem~\ref{th:estimation_error} is modified, so that after $x$ episodes the error is bounded by
$$
\frac{\br{p+q} \log \br{x}   }{2x}
	+ \frac{\br{p+q}\log\left[\br{C_1(\te)  \| Y_{0}\|  +  \frac{C_1(\te) B_u}{1-\rho(\te)}\left\{1+  \sum_{\ell=1}^{q} |b_\ell(\te)|  \right\} }^2 
		+ qB^2_u \right]	- 2\log\br{\delta}}{x}.
$$
The relation $r_t \le (e_t-w_t)^2$ for $t\notin \cI$ and $t\ge t\ust_1 \vee t\ust_2 \vee t\ust_3$ that was derived in Section~\ref{sec:predict_error} continues to hold, except that now in the definition of $t\ust_1,t\ust_2, t\ust_3$ we replace $B_w$ by $B_w(T)$ and this introduces additional dependency upon $T$. We now discuss changes made while analyzing regret in Section~\ref{sec:regret_analysis}.

The cumulative regret $\sum_{t\in\cI} r_t$ for times $t\in\cI$ was determined by the cumulative number number of exploratory instants $N\ui_T$, and the bound $\log\br{T\slash \delta}\log(T)$ on $\|\phi_t\|^2$, where we had $\|\phi_t\|^2 \lesssim B^2_w$. We now have $N\ui_T\gtrsim \log\br{T\slash \delta}\log(T)$, while $\|\phi_t\|^2 \lesssim \log\br{T\slash \delta}$. In summary, this regret is now bounded $\br{\log\br{T\slash \delta}}^2\log(T)$.

We recollect that the analysis of $\sum_{t\notin\cI} r_t$ involved summation of~\eqref{eq:q_recursion}, and bounding $\cT_2,\cT_3,\cT_4$. Thus, we will have to consider the dependence upon $B_w$ of the bounds derived in Section~\ref{sec:T_bounds} on $\cT_2,\cT_3,\cT_4$ therein. From Proposition~\ref{prop:t3}, the bound on $\cT_3$ is $\lesssim B_u$ (after hiding terms that are $\log B_u,\log\log B_u$), or equivalently $\lesssim B_w$, and hence this bound is now $\lesssim B_w(T)$. Bound on $\cT_{4,2}$ is $\lesssim \log(B_u)$, which is same as $\lesssim \log(B_w)$, and hence contributes a term that grows as $\log\log\br{\frac{T}{\delta}}$. Upon making the changes described above, we obtain the desired result.

\section{Useful Results}

\subsection{Self-Normalized Martingales Concentration Results} 
Let $\{ \cF_t, t\in\bN \}$ be a filtration and $\{ \eta_t,t\in\bN \}$ an $\cF_t$-adapted process such that $\bE\left\{ \exp(\lm\eta(t) \Big| \cF_{t-1}\right\}\le \exp(\lm^2 R^2 \slash 2)$. Let $\{X(t)\}_{t\in\bN}$ be a predictable process, i.e., $X(t)$ is $\cF_{t-1}$ measurable. Define $\bar{V}(t):= V + \sum_{s=1}^{t} X(s) X(s)'$ and $S(t):= \sum_{s=1}^{t}\eta(s)X(s)$. The following holds w.p. greater than $1-\delta$:
\al{
\|S(t)\|^{2}_{\bar{V}(t)^{-1}} \le R^2 \log\left( \frac{\det(V(t))\det(V)}{\delta} \right), ~\forall t\in \bN.\label{eq:self_normalized}
}
The following result is essentially (3.6) of~\cite{lai1982asymptotic}. 
\begin{lemma}\label{lemma:lmin_proj}
	Consider an $n\times m$ matrix $A = \{a_{i,j}\}$, and denote its columns by $A_{\cdot,1},A_{\cdot,2},\ldots,A_{\cdot,m}$. Let $\hat{A}_{\cdot,j}$ denote the projection of the $j$-th column on the linear space spanned by the remaining $m-1$ columns. Then 
	\al{
m^{-1}\min_{1\le j\le m} \|  A_{\cdot,j} - \hat{A}_{\cdot,j}\|^2 \le \lm_{\min}\br{A'A} \le m \min_{1\le j\le m} \|  A_{\cdot,j} - \hat{A}_{\cdot,j}\|^2.
}
\end{lemma}

The following result from~\cite{tao_vu} is essentially the Azuma-Hoeffding concentration inequality for unbounded random variables.
\begin{theorem}\label{th:azuma_unbounded}
	Let $\{X_i\}_{i=1}^{n}$ be a supermartingale such that the differences are bounded w.h.p., i.e.
	\al{
		\bP\br{~\exists~ i ~\mbox{ s.t. } |X_i - X_{i-1}|> B } \le \eps,	
	}
	where $\eps,B>0$. Then, 
	\al{
		\bP\br{X_n > X_0 + x} \le \exp\br{-\frac{x^2}{2n B^2}} +\eps.
	}
\end{theorem}

\section{Bounds on $\|Y_t\|,\|U_t\|$}\label{sec:bound_yt}
Consider the following vector-valued processes associated with the ARX model~\eqref{def:arx}: $Y_t = \br{y_t,y_{t-1},\ldots,y_{t-p+1}}'$ and $U_t = \br{u_t,u_{t-1},\ldots,u_{t-q+2}}'$. Consider the matrices
\al{
	A= 
	\begin{pmatrix}
		a_1 & \cdots & a_{p-1} & a_p \\
		I_{p-1} &        &     & 0 
	\end{pmatrix}
	,}
and,
\al{
	B =  
	\begin{pmatrix}
		-b_2 \slash b_1 & \cdots & -b_q \slash b_1 \\
		I_{q-2} &   &  0
	\end{pmatrix},
}
where $I_{p-1},I_{q-2}$ are identity matrices of sizes $p-1$ and $q-2$ respectively. We have the following bounds, which are essentially Lemma 2-(i),~(ii) of~\cite{lai1987asymptotically}.
\begin{lemma}\label{lemma:state_bound}
	Consider times $t_1 > t_0$, and let Assumption~\ref{assum:unit_circle} hold true for the ARX model~\eqref{def:arx}.~There exists $0< \rho<1$, $C_1 >0$ such that:
	\begin{enumerate}[(i)]
		\item 
		\al{
			\| Y_{t_1} \| \le 	C_1 \rho^{t_1 - t_0} \| Y_{t_0}\| + C_1 \sum_{s=0}^{t_1 - t_0 -1} \rho^{s} \left\{|w_{t_1-s}| +  \sum_{\ell=1}^{q} |b_\ell|  | u_{t_1 -s - \ell } |  \right\}.
		}
		\item We have the following bound on $\|U_t\|$:
		\al{
			\|U_t \| \le C_1 \rho^{t_1 - t_0} \|U_{t_0}\| + \frac{C_1}{b_1} \sum_{s=0}^{t_1 - t_0 -1} \rho^{s} \left\{ | w_{t_1 + 1 - s}|  + \sum_{\ell=1}^{p}| a_\ell | | y_{t_1 + 1 - s - \ell}  | \right\}.
		}
	\end{enumerate}		
\end{lemma}
We note that $\rho$ above can be taken to be any number greater than the spectral radius of $A$ but less than $1$. 
\begin{lemma}\label{lemma:bound_phi}
For the case when $|w_s|\le B_w$ for all $s=1,2,\ldots$, we have,
	$$
	\|Y_t\| \le C_1 \rho^{t_1 - t_0} \| Y_{t_0}\| +\frac{B_u C_1}{1-\rho}  \left\{1+  \sum_{\ell=1}^{q} |b_\ell|  \right\}.
	$$
\end{lemma}
\proof 
The proof follows from Lemma~\ref{lemma:state_bound}-(i) after noting that $|w_s|\le B_w \le B_u$, and also $|u_s|\le B_u$ for all $s$.
\endproof

\begin{remark}
	When we want to indicate the dependence of $\rho$ on the system parameter, we will write $\rho(\te)$. Similarly for $a_{\ell}(\te),b_{\ell}(\te),C_1(\te)$.
\end{remark}

We now exhibit a result for $|y_t|,~t\in\cI$ that holds for the ARX process evolving under the PIECE algorithm.~Note that the inputs $\{u_t\}$ are chosen so as to satisfy the following bounds,
\al{
	|u_t| & \le B_u~\forall t, ~~|u_t| \le B_w, \mbox{ for }t\in\cI.\label{ineq:bd1}
}
Moreover, the noise process $\{w_t\}$ is also bounded as, 
\al{
	|w_t| & \le B_w,~\mbox{ where } B_w \le B_u.\label{ineq:bd2}
}
The exploratory phase $\cI$ is comprised of several episodes, where the $i$-th episode consists of $m_i$ consecutive steps, starting at time $n_i$ and ending at time-step $n_i + m_i$.

Since $|u_s|,|w_s|\le B_u$, we obtain the following bound from Lemma~\ref{lemma:state_bound}-(i) by setting $t_1=n_i$ and $t_0 = 0$:	
\al{
	\| Y_{n_i} \| & \le 	C_1 \rho^{n_i} \| Y_{0}\| + B_u C_1 \sum_{s=0}^{t_1  -1} \rho^{s} \left\{1+  \sum_{\ell=1}^{q} |b_\ell|   \right\}\notag\\
	&\le C_1 \rho^{n_i} \| Y_{0}\| + \frac{B_u C_1}{1-\rho} \left\{1+  \sum_{\ell=1}^{q} |b_\ell|   \right\}.\label{ineq:21}
}
Now, in Lemma~\ref{lemma:state_bound}-(i), we let $t_0 = n_i$, and $t_1 = n_i + m$ where $m<m_i$, so that we have $t_1 < n_{i} + m_i$, and also $|u_t|,|w_t|\le B_3$ for $n_i < t<n_i + m_i$. So,
\al{
	\| Y_{n_i + m} \| & \le 	C_1 \rho^{m} \| Y_{n_i}\| + C_1 \sum_{s=0}^{m} \rho^{s} \left\{|w_{n_i + m -s}| +  \sum_{\ell=1}^{q} |b_\ell|  | u_{n_i + m -s - \ell } |  \right\}\notag\\
	&\le 	C_1 \rho^{m} \| Y_{n_i}\| + B_w  C_1 \sum_{s=0}^{m} \rho^{s} \left\{1+  \sum_{\ell=1}^{q} |b_\ell|  \right\}\notag\\
	&\le 	C_1 \rho^{m} \left[C_1 \rho^{n_i} \| Y_{0}\| + \frac{B_u C_1}{1-\rho} \left\{1+  \sum_{\ell=1}^{q} |b_\ell|   \right\}\right]+ B_w  C_1 \sum_{s=0}^{m} \rho^{s} \left\{1+  \sum_{\ell=1}^{q} |b_\ell|  \right\}\notag\\
	&\le 	C_1 \rho^{m} \left[C_1  \| Y_{0}\| + \frac{B_u C_1}{1-\rho} \left\{1+  \sum_{\ell=1}^{q} |b_\ell|   \right\}\right]+ B_w  C_1 \sum_{s=0}^{m} \rho^{s} \left\{1+  \sum_{\ell=1}^{q} |b_\ell|  \right\}.	
	\label{ineq:2}
}
Define,
\al{
m\ust := \Bigg\lceil\frac{1}{\log \rho}\log \br{\frac{B_w}{\left[C_1 \| Y_{0}\| + \frac{B_u C_1}{1-\rho} \left\{1+  \sum_{\ell=1}^{q} |b_\ell|   \right\}\right]
		 \cdot C_1}} \Bigg\rceil.\label{def:must}
}
~Upon choosing $m>m\ust$, we obtain the following bound on $\|Y_{n_i + m}\|$.

\begin{lemma}\label{lemma:bound_ys}
Consider the ARX system~\eqref{def:arx} evolving under PIECE, in which $\{w_t\}$ satisfies Assumption~\ref{assum:bounded_noise}, and $\{u_t\}$ satisfies~\eqref{ineq:bd1}.~Consider the $i$-th exploratory episode, and let $m$ satisfy,
	\al{
		m_i \ge m\ge m\ust.\label{ineq:30}
	}
Then,
\al{
	\| Y_{n_i + m} \| \le 	B_w \left[1+  C_1 \sum_{s=0}^{m} \rho^{s} \left\{1+  \sum_{\ell=1}^{q} |b_\ell|  \right\} \right].
}
If instead $\{w_t\}$ satisfies Assumption~\ref{assum:sub_gaussian}, the same conclusion holds on $\cG_w$.
\end{lemma}

\begin{proposition}\label{prop:5}
	Consider the ARX model~\eqref{def:arx}. Under Assumption~\ref{assum:bounded_noise},
	\al{
		\|Y_s\|^2 \le 2\|A^s\|^2 \|Y_0\|^2 + \frac{4\|b\|^2}{1-\rho} \br{\sum_{\ell=1}^{s} \rho^{s-\ell} \|U_{\ell}\|^2 } + \frac{4}{\br{1-\rho}^2}B^2_w,
	}
	and,	
	\al{
		{
			\Tr\br{\Psi'_t \Psi_t} \le \br{4\left[\frac{\|b\|}{1-\rho} \right]^2 + 1  }q \Tr\br{\sum_{s=p+1}^{t}U_s U'_s}+ 4\left[\frac{1}{1-\rho} \right]^2 B^2_w t,~\label{ineq:10} 
		}
	}
	where
	$
	b= \br{b_1,b_2,\ldots,b_q}
	$,
	and $\rho$ is as in Lemma~\ref{lemma:state_bound}.~Note that $tr\br{\Psi'_t \Psi_t}  = \sum_{s=I}^{t}\| \phi_s \|^2$, where $\phi_s$ is the regressor during time $s$.
\end{proposition}
\begin{proof} We have,
	\al{
		|w_t| \le 	B_w.
	}
	We will derive an upper-bound on $\sum_{s=1}^{t}y^2_s$. Define,
	\al{
		\tilde{U}_s := \br{b_1 u_{s-1}+ b_2 u_{s-2}+\ldots +  b_q u_{s-q} + w_s,0,\ldots,0 }'.
	}
	
	We have,
	\al{
		Y_s = A Y_{s-1} + \tilde{U}_s,
	}
	where $A$ is as in~\eqref{def:A}. This yields,
	\al{
		\|Y_s\|^2 & = \| A^s Y_0 + \sum_{\ell=1}^{s}  A^{s-\ell}  \tilde{U}_{\ell} \|^2 \notag\\
		& \le 2\|A^s\|^2 \|Y_0\|^2 + 2\|\sum_{\ell=1}^{s} A^{s-\ell} \tilde{U}_{\ell} \|^2 \notag\\
		& \le 2\|A^s\|^2 \|Y_0\|^2 + 2\br{\sum_{\ell=1}^{s} \|A^{s-\ell}\| \|\tilde{U}_{\ell} \|}^2 \notag\\
		& \le 2\|A^s\|^2 \|Y_0\|^2 + 2\br{\sum_{\ell=1}^{s} \|A^{s-\ell}\|} \br{\sum_{\ell=1}^{s} \|A^{s-\ell}\| \br{b_1 u_{\ell-1}+ b_2 u_{\ell-2}+\ldots +  b_q u_{\ell-q} + w_\ell}^2} \notag\\
		& \le 2\|A^s\|^2 \|Y_0\|^2 + \frac{2}{1-\rho} \br{\sum_{\ell=1}^{s} \|A^{s-\ell}\| \br{b_1 u_{\ell-1}+ b_2 u_{\ell-2}+\ldots +  b_q u_{\ell-q} + w_\ell}^2} \notag\\
		& = 2\|A^s\|^2 \|Y_0\|^2 + \frac{2}{1-\rho} \br{\sum_{\ell=1}^{s} \|A^{s-\ell}\| \br{b\cdot U_{\ell} + w_\ell}^2} \notag\\
		& = 2\|A^s\|^2 \|Y_0\|^2 + \frac{4}{1-\rho} \br{\sum_{\ell=1}^{s} \|A^{s-\ell}\| \left\{ \|b\|^2 \|U_{\ell}\|^2+ w^2_\ell \right\}} \notag\\
		& = 2\|A^s\|^2 \|Y_0\|^2 + \frac{4}{1-\rho} \br{\sum_{\ell=1}^{s} \rho^{s-\ell} \left\{ \|b\|^2 \|U_{\ell}\|^2+ w^2_\ell \right\}} \notag\\
		& = 2\|A^s\|^2 \|Y_0\|^2 + \frac{4\|b\|^2}{1-\rho} \br{\sum_{\ell=1}^{s} \rho^{s-\ell} \|U_{\ell}\|^2 } + \frac{4}{\br{1-\rho}^2}B^2_w,\notag
	}
	where the third inequality follows from the Cauchy-Schwartz inequality, and $\rho$ is as in Lemma~\ref{lemma:state_bound}. Upon using $\|\phi_{s}\|^2 = \|Y_s\|^2 + \|U_s\|^2$, we get,
	\al{
		\|\phi_{s}\|^2 \le 2\|A^s\|^2 \|Y_0\|^2 + \frac{4\|b\|^2}{1-\rho} \br{\sum_{\ell=1}^{s} \rho^{s-\ell} \|U_{\ell}\|^2 } + \frac{4}{\br{1-\rho}^2}B^2_w +  q B^2_u.\label{ineq:84}
	}
	Summing up the above inequality from $s=1$ to $t$ we get,
	\al{
		\sum_{s=1}^{t}	\|\phi_{s}\|^2  & \le 2\sum_{s=1}^{t}\|A^s\|^2 \|Y_0\|^2 + \frac{4\|b\|^2}{1-\rho} \br{\sum_{s=1}^{t}\sum_{\ell=1}^{s} \rho^{s-\ell} \|U_{\ell}\|^2 } + \frac{4}{\br{1-\rho}^2}B^2_w t+ \sum_{s=1}^{t}\|U_s\|^2\\
		& \le \frac{2}{\br{1-\rho}^2 }Y_0\|^2 + \frac{4\|b\|^2}{\br{1-\rho}^2} \br{\sum_{s=1}^{t} \|U_{\ell}\|^2 } + \frac{4}{\br{1-\rho}^2}B^2_w t+ \sum_{s=1}^{t}\|U_s\|^2.
	}
\end{proof}

\section{Simulation Setup}\label{sec:simulations_appendix}
In this section, we provide additional experimental results and the details of simulation setup. In the experiments, we compare the empirical performance of the PIECE algorithm with that of the following two baseline algorithms:\\\\
\textbf{1. Certainty Equivalence:}
The first baseline algorithm is the standard certainty equivalence controller as described in Algorithm \ref{algo:ce}.
\begin{algorithm}[H]
   \caption{Certainty Equivalence~(CE)}
   \label{algo:ce}
    \begin{algorithmic}
       \STATE {\bfseries Input} Horizon, $H$.
       \STATE $\tau = \inf\left\{t>0: \sum_{s \leq t}{\phi_s \phi_s^\prime}\textit{ is invertible and } b_{1,t} \neq 0\right\}$.
       \IF{$t \le \tau$}
       \STATE Generate an exploratory zero mean white noise input $u_t$.
       \ELSE
       \STATE Compute the estimates $\te_{t-1}$ and $\lm_{t-1}$ as follows:
       \nal{
        \te_{t-1} & = \left(a_{1,t-1},\ldots,a_{p,t-1},{b}_{1,t-1},\ldots,{b}_{q,t-1} \right)' \notag\\
        & :=\br{\sum_{s < t} \phi_s \phi'_s }^{-1}\left(\sum_{s < t} \phi_s y_{s+1}\right),}\label{def:lse}
        
        and 
        $$\lm_{t-1} := (-1\slash b_{1,t-1}) \left(a_{1,t-1},\ldots,a_{p,t-1},{b}_{2,t-1},\ldots,{b}_{q,t-1} \right)'$$.
       \STATE Apply control, $u_t = \lm'_{t-1} \psi_t$
       \ENDIF
    \end{algorithmic}
\end{algorithm}
\textbf{2. Lai and Wei (LW) \cite{lai1987asymptotically}}. 
For the second baseline, we compare the results with the algorithm proposed by \cite{lai1987asymptotically}.
This algorithm differs from PIECE in the following ways: (i) Clipping of the input (ii) Durations of the exploration phases. 
The details are described in Algorithm \ref{algo:lw}.
The simulation results show that these modifications lead to significant improvement in the empirical regret.

\begin{algorithm}[H]
   \caption{Lai and Wei~(LW)}
   \label{algo:lw}
    \begin{algorithmic}
       \STATE {\bfseries Input} Algorithm parameters, $\delta > 0, \rho > 1, B_2>0, B_w>0$; horizon, $H$.
       \STATE $\cI = \left\{1, \ldots, \tau\right\}\cup \left(\cup_{i : n_i < H} \left\{n_i+1, \ldots, n_i+m_i\right\}\right)$ where \\$\tau = \inf\left\{t>0: \sum_{s \leq t}{\phi_s \phi_s^\prime}\textit{ is invertible and } b_{1,t} \neq 0\right\}$, $n_i = e^{i^\rho(1+o(1))}$ and $m_i = (\log{i})^\delta$.
       \IF{$t \in \cI$}
       \STATE Generate an exploratory white noise input $u_t$ such that $|u_t|\le B_w \log{\log{t}}$ and has mean $0$.
       \ELSE
       \STATE  Compute the estimates $\tilde{\theta}^{(\cI)}_{t-1}$, $\tilde{\lm}^{(\cI)}_{t-1}$  and $\lm_t$ as defined in \eqref{eq:theta_estimate}.
       \STATE
       \nal{
    	u_t = 
    	\begin{cases}
    		\br{-B_u} \vee \br{\lm'_{t-1} \psi_t} \wedge \br{B_u} ~\mbox{ if } \Big| \lm'_{t-1} \psi_t - \br{\tilde{\lm}^{(\cI)}_{t-1} }'\psi_{t} \Big| \le B_2 \times \frac{\log N^{(\cI)}_t}{\sqrt{N^{(\cI)}_t}}\|\psi_t\|, \\
    		\br{-B_u} \vee \br{  \br{\tilde{\lm}^{(\cI)}_{t-1} }' \psi_{t}} \wedge \br{B_u} ~\mbox{ otherwise }.
    	\end{cases}}
        \ENDIF
    \end{algorithmic}
\end{algorithm}
\textbf{Hyper-Parameters} PIECE needs two system-dependent parameters and a bound on the absolute value of the noise in order to compute algorithm hyper-parameters. $\rho$ is an upper bound on the eigenvalues of matrix $A$~(see \ref{sec:system_model}) and $\|\lambda\|_2$ is 2-norm of vector $\lambda$. The duration of the first exploration episode is of length $\|\lambda\|^3_2$. In the following table, their values are given for the three examples described in $\ref{sec:simulations}$. Other hyper-parameters depend on the system as well as the noise process. $B_w$ is the upper bound for the absolute value of the noise sequence, $B_u$ is the threshold for clipping input, $\delta > 0$ be such that $(B_u, \delta)$ satisfies \ref{eq:24},\ref{deltaineq_1},\ref{deltaineq_2} and \ref{deltaineq_3}. $H$, defined as in \ref{def:n_i}, is the exploration episode duration other than the first. Among the noise process-dependent hyper-parameters, we observed $\delta$ does not vary significantly across experiments. So, we kept it constant for a particular system and added its values in the following table. The experiment-dependent hyper-parameter values are given along with the corresponding experimental results.
\begin{table}[H]
    \centering
    \begin{tabular}{|c|c|c|c|}
        \hline
         & $\rho$ & $\|\lambda\|_2$ & $\delta$\\
        \hline
        Example I & 0.8986 & 5.1 & 0.0038\\
        \hline
        Example II & 0.6782 & 4.69 & 0.0227\\
        \hline
        Example III & 0.8282 & 3.36 & 0.0104\\
        \hline
    \end{tabular}
    \caption{System parameters}
    \label{tab:sys_param}
\end{table}
\subsection{Additional Simulation Results}
In addition to the the results provided in Section \ref{sec:simulations}, we provide further results for a variety of system noise processes. The empirical cumulative regret of the PIECE algorithm is consistently lower than those of LW and CE. Especially, the key advantage of the PIECE algorithm over the LW and CE algorithms is in the initial transient phase of the experiment where both baseline algorithms suffer large regret. Similar to the results in Section \ref{sec:simulations}, the explotration scheme of PIECE is more efficient in achieving lower regret, which is the primary objective of the controller, at
the cost of a higher estimation error than LW.

Table \ref{tab:summary} summaries the additional results:
\begin{table}[H]
    \centering
    \begin{tabular}{|c|c|c|c|}
        \hline
        System Noise & Regret & Estimation Error & Terminal Regret \\
        \hline
        iid Gaussian Noise, $\sigma=0.6$ & Figure \ref{fig:cumulative_regret_iidG0.6} & Figure \ref{fig:estimation_error_iidG0.6} & Table \ref{tab:iidG0.6}\\\hline
        iid Gaussian Noise, $\sigma=1$ & Figure \ref{fig:cumulative_regret_iidG1.0} & Figure \ref{fig:estimation_error_iidG1.0} & Table \ref{tab:iidG1.0}\\\hline
        Random walk with iid Gaussian steps, $\sigma=0.5$ & Figure \ref{fig:cumulative_regret_rwalk0.5} & Figure \ref{fig:estimation_error_rwalk0.5} & Table \ref{tab:rwalk_iidG0.5}\\\hline
        Random walk with iid Gaussian steps, $\sigma=1$ & Figure \ref{fig:cumulative_regret_rwalk1.0} & Figure \ref{fig:estimation_error_rwalk1.0} & Table \ref{tab:rwalk_iidG1.0}\\\hline
    \end{tabular}
    \caption{Summary of Results}
    \label{tab:summary}
\end{table}
\begin{figure}[H]
    \centering
    \begin{subfigure}[b]{0.32\textwidth}
        \centering
        \includegraphics[width=\textwidth]{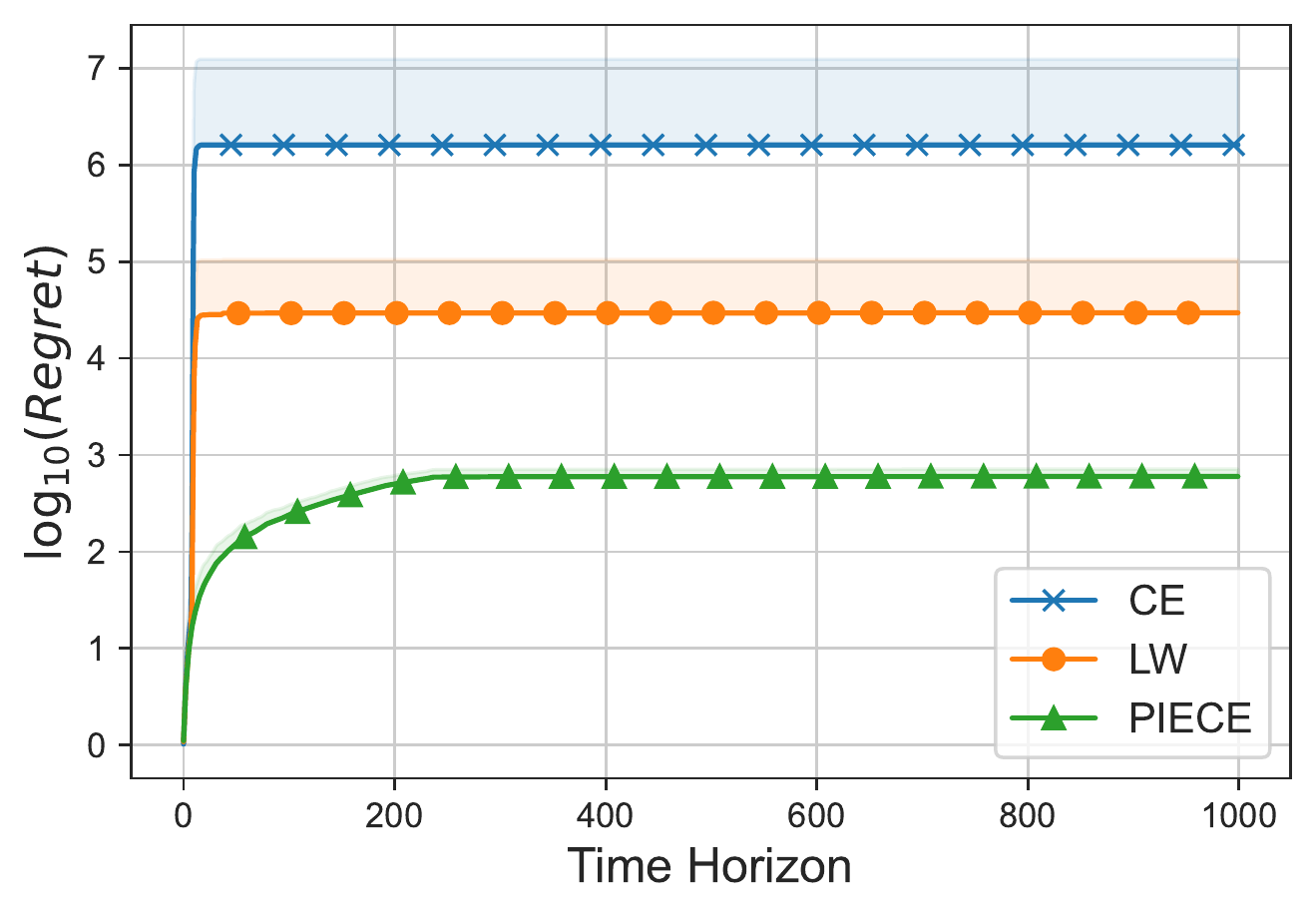}
        \caption{Example I}
        \label{fig:s1_lreg_g0.6}
    \end{subfigure}
    \hfill
    \begin{subfigure}[b]{0.32\textwidth}
        \centering        \includegraphics[width=\textwidth]{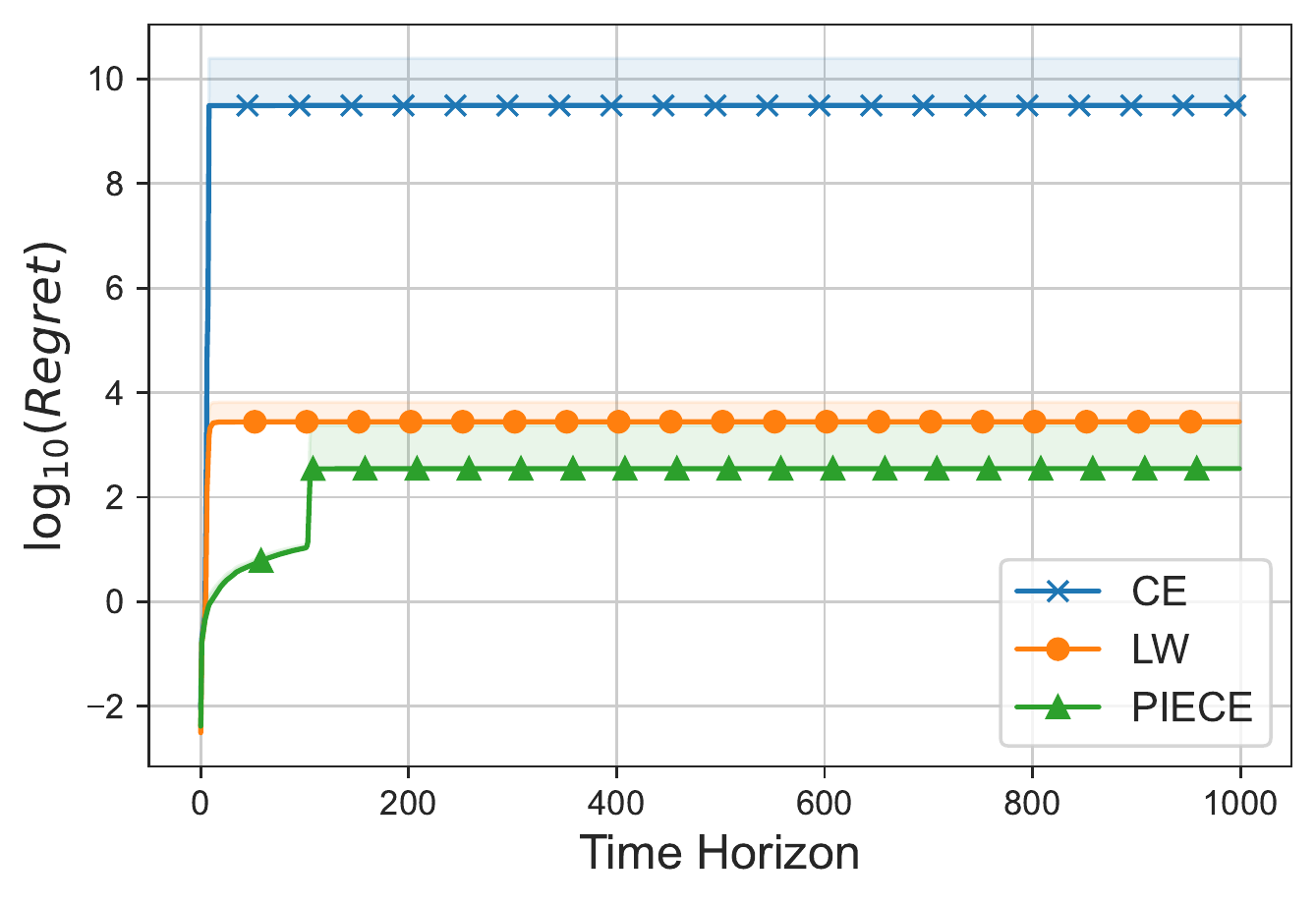} 
        \caption{Example II}         
        \label{fig:s2_lreg_g0.6}
    \end{subfigure}
    \hfill
    \begin{subfigure}[b]{0.32\textwidth}
        \centering
        \includegraphics[width=\textwidth]{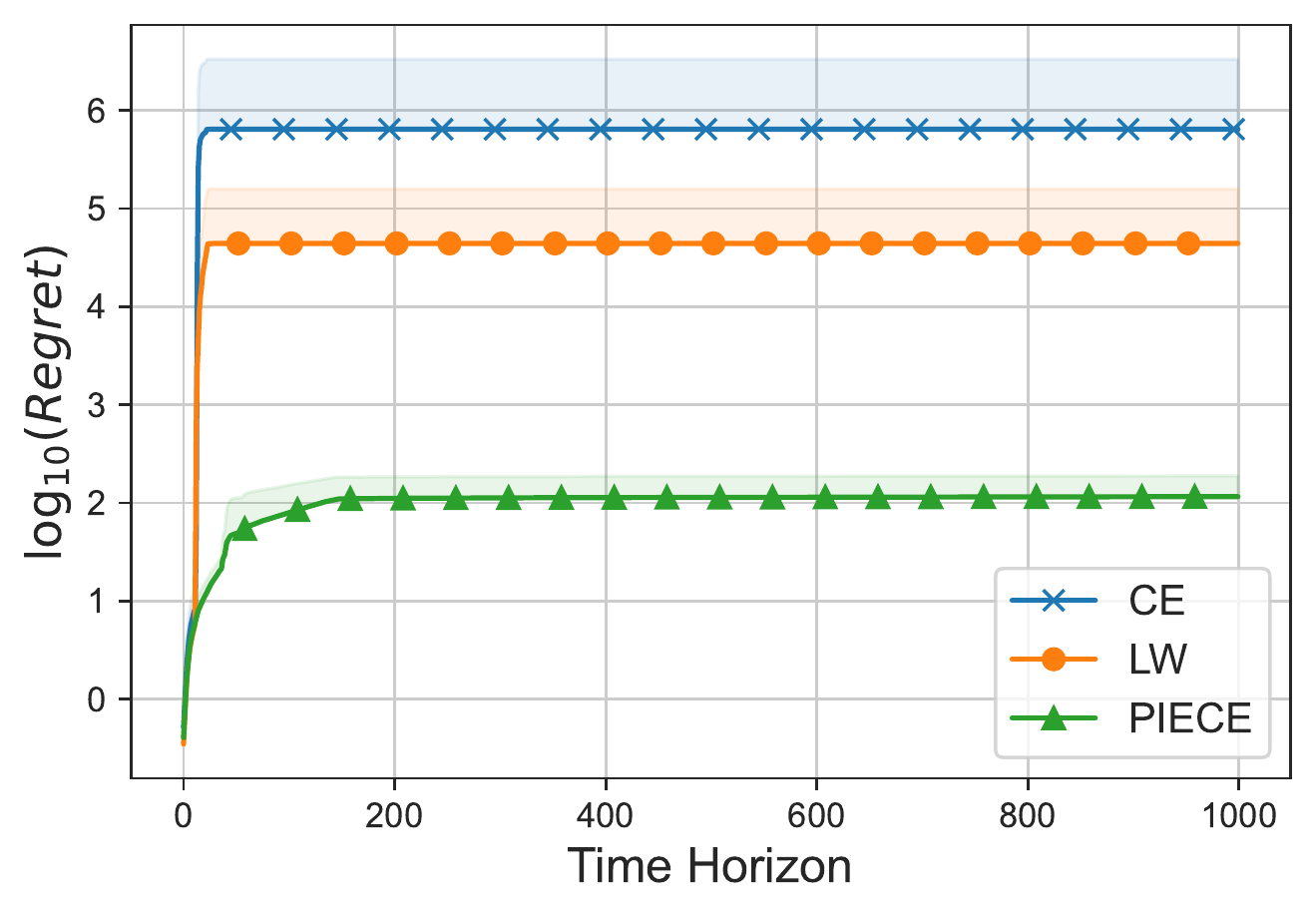}
        \caption{Example III}
        \label{fig:s3_lreg_g0.6}
    \end{subfigure}
    \caption{Log(Cumulative Regret) averaged over 50 runs for Gaussian noise with mean $0$ and standard deviation $0.6$.}
    \label{fig:cumulative_regret_iidG0.6}
\end{figure}

\begin{figure}[H]
    \centering
    \begin{subfigure}[b]{0.32\textwidth}
        \centering
        \includegraphics[width=\textwidth]{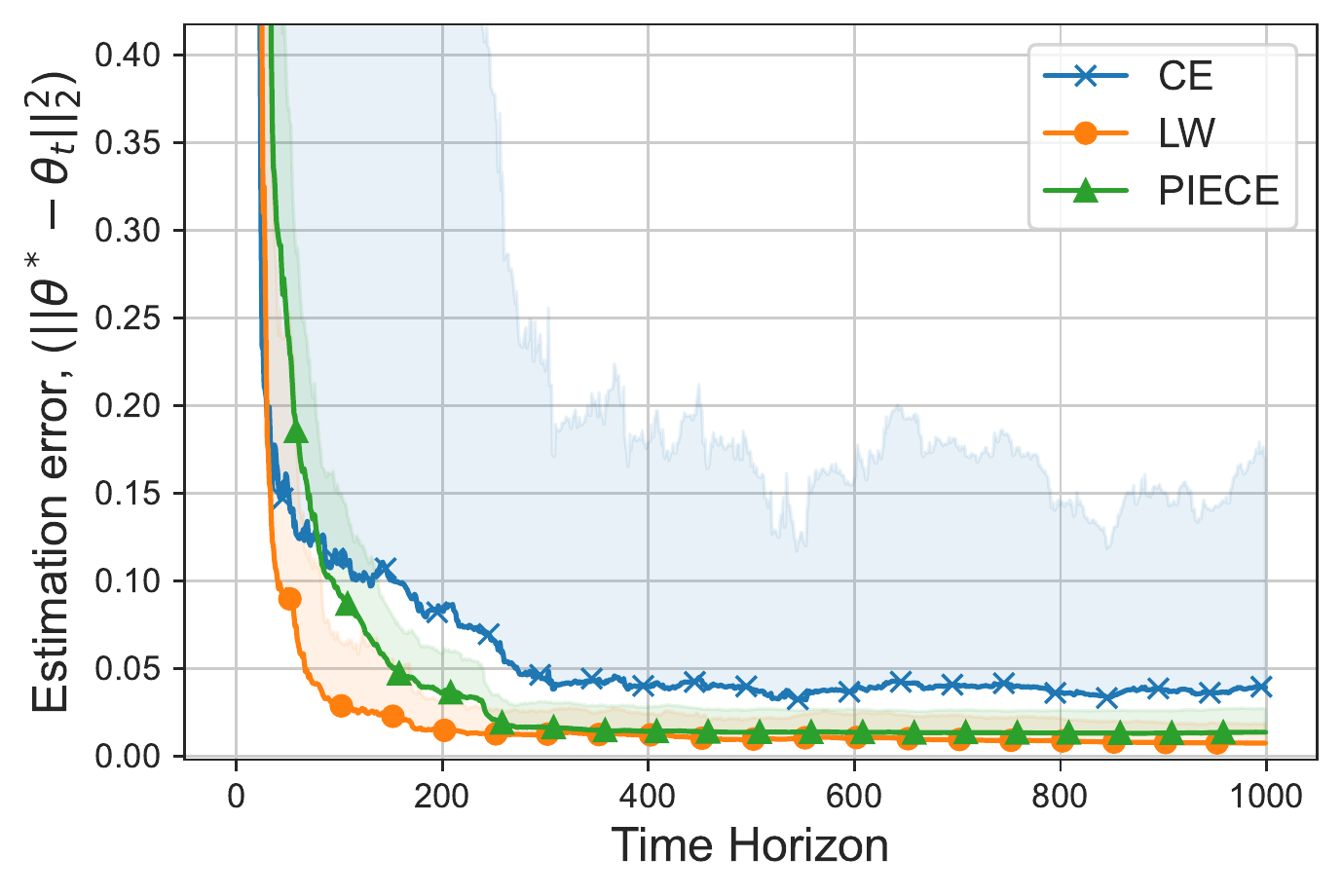}
        \caption{Example I}
        \label{fig:s1_err_g0.6}
    \end{subfigure}
    \hfill
    \begin{subfigure}[b]{0.32\textwidth}
        \centering
        \includegraphics[width=\textwidth]{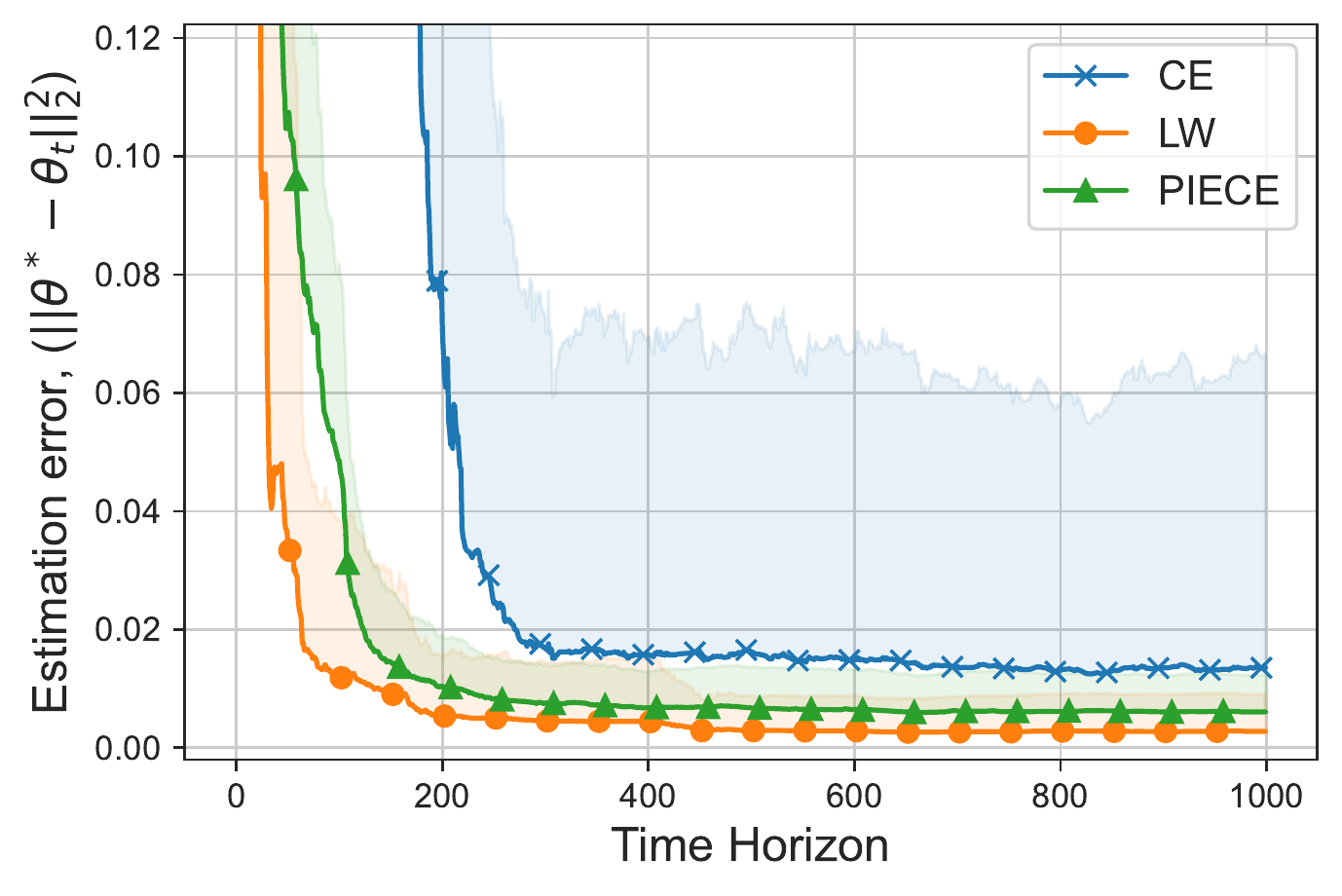}      
        \caption{Example II}         
        \label{fig:s2_err_g0.6}
    \end{subfigure}
    \hfill
    \begin{subfigure}[b]{0.32\textwidth}
        \includegraphics[width=\textwidth]{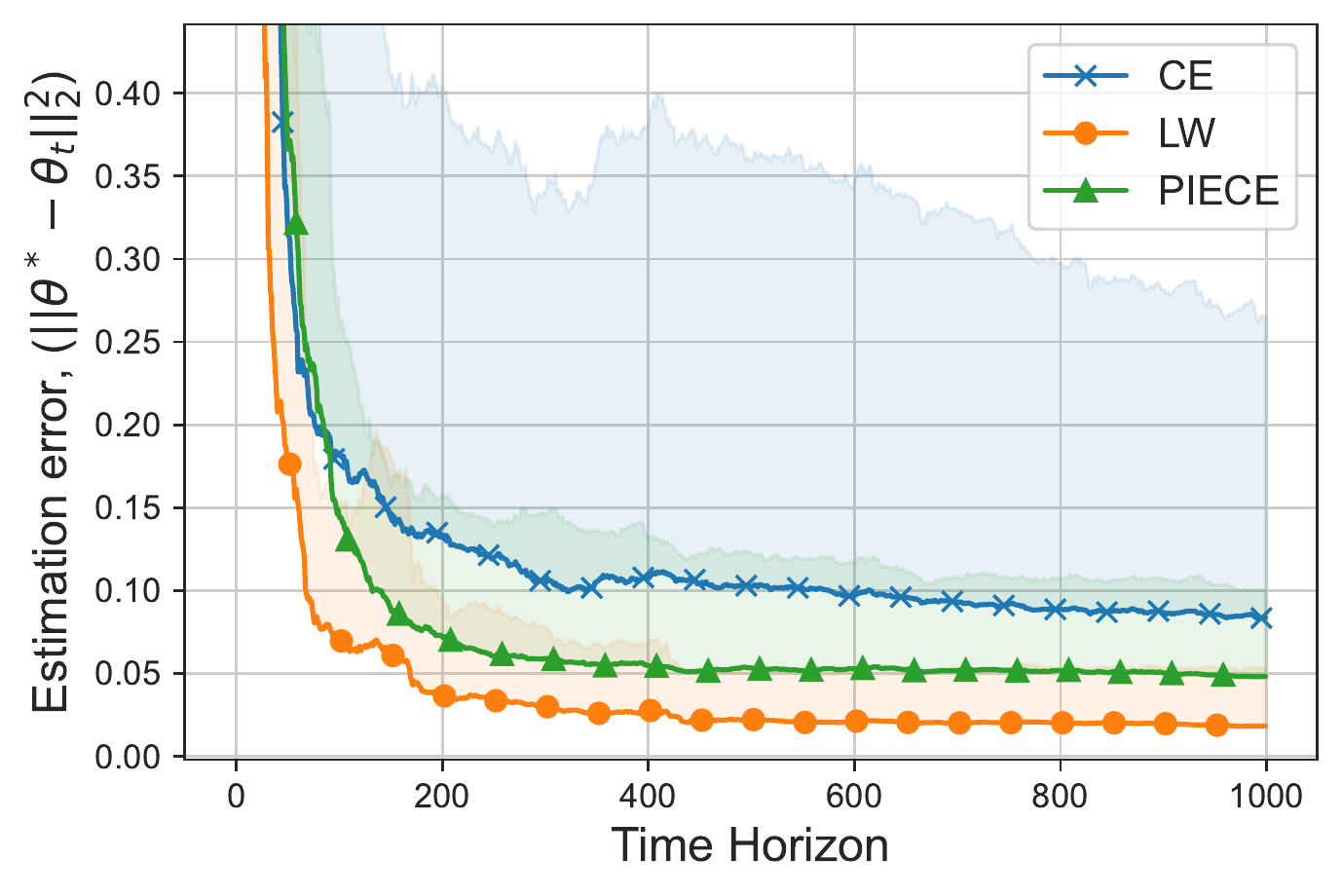}
        \caption{Example III}
        \label{fig:s3_err_g0.6}
    \end{subfigure}
    \caption{Estimation Error ($||\theta^\star-\theta_t||^2_2$) for Gaussian noise with mean $0$ and standard deviation $0.6$.}
    \label{fig:estimation_error_iidG0.6}
\end{figure}   
\begin{table}[H]
    \centering
    \begin{tabular}{|c|c|c|c|}
        \hline
        & CE & LW & PIECE \\
        \hline
        Example I & 1611301 & 29569 & 600\\
        \hline
        Example II & 3129063898 & 2786 & 353\\
        \hline
        Example III & 645073 & 44267 & 116\\
        \hline
    \end{tabular}
    \caption{Average Regret at $T=1000$ for Gaussian noise with mean $0$ and standard deviation $0.6$.}
    \label{tab:iidG0.6}
\end{table}
\begin{table}[H]
    \centering
    \begin{tabular}{|c|c|c|c|}
        \hline
        & $B_w$ & $B_u$ & $H$\\
        \hline
        Example I & 1.8 & 2219381.4 & 182\\
        \hline
        Example II & 1.8 & 17187.19 & 34\\
        \hline
        Example III & 1.8 & 186218.2 & 90\\
        \hline
    \end{tabular}
    \caption{PIECE hyper-parameters for Gaussian noise with mean $0$ and standard deviation $0.6$.}
    \label{tab:hpiidG0.6}
\end{table}
\begin{figure}[H]
    \centering
    \begin{subfigure}[b]{0.32\textwidth}
        \centering
        \includegraphics[width=\textwidth]{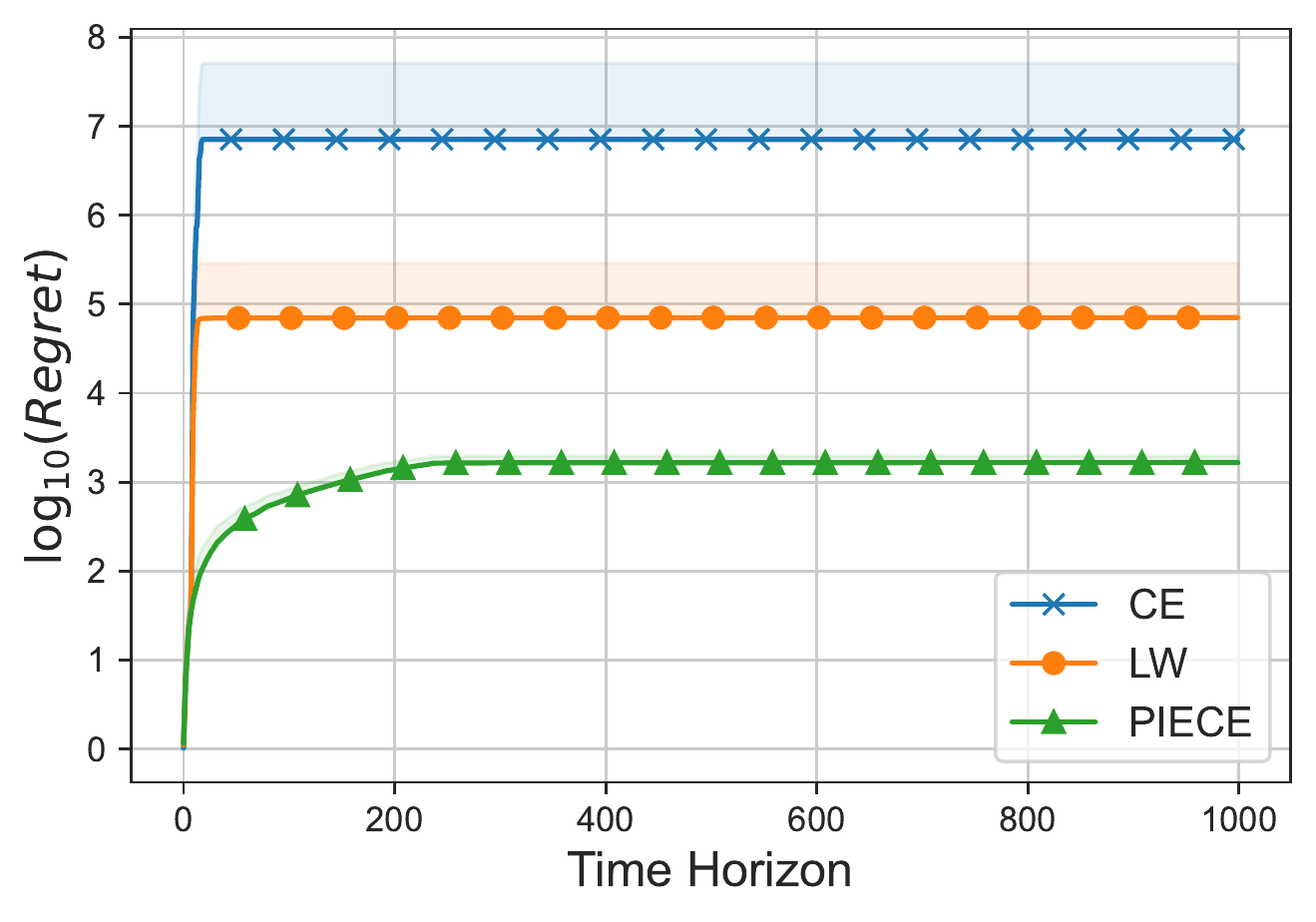}
        \caption{Example I}
        \label{fig:s1_lreg_g1.0}
    \end{subfigure}
    \hfill
    \begin{subfigure}[b]{0.32\textwidth}
        \centering
        \includegraphics[width=\textwidth]{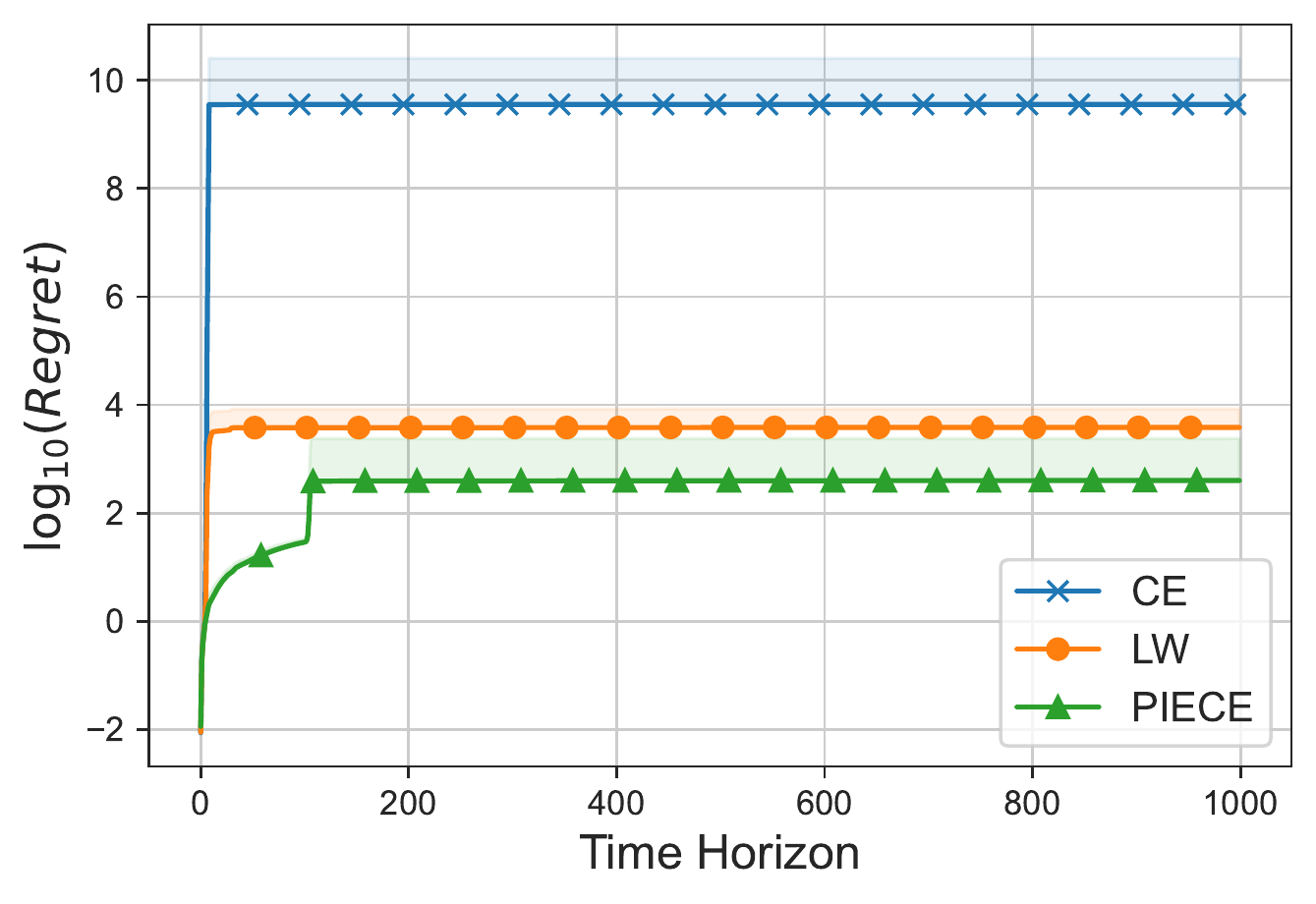} 
        \caption{Example II}         
        \label{fig:s2_lreg_g1.0}
    \end{subfigure}
    \hfill
    \begin{subfigure}[b]{0.32\textwidth}
        \centering
        \includegraphics[width=\textwidth]{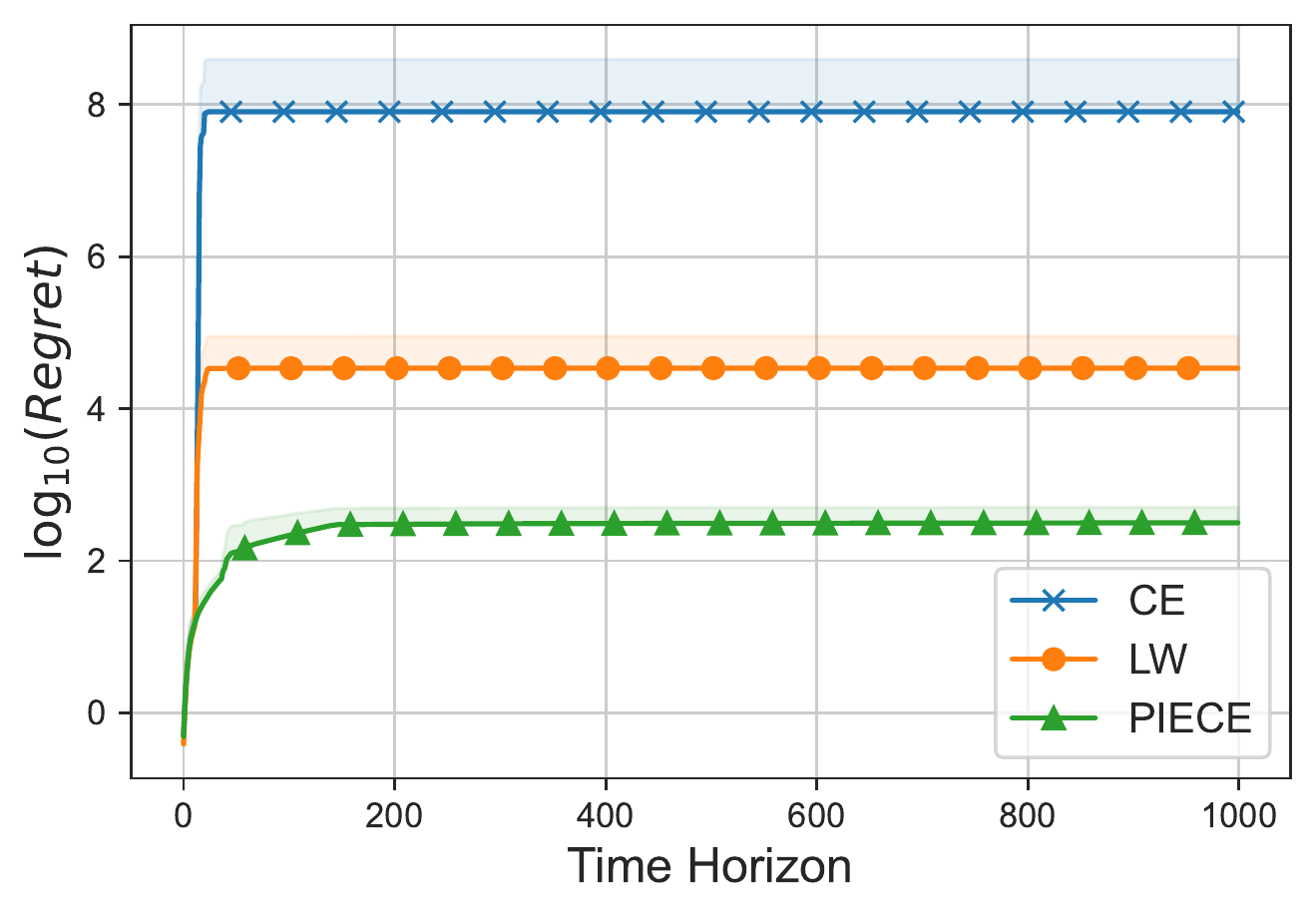}
        \caption{Example III}
        \label{fig:s3_lreg_g1.0}
    \end{subfigure}
    \caption{Log(Cumulative Regret) averaged over 50 runs for Gaussian noise with mean $0$ and standard deviation $1.0$.}
    \label{fig:cumulative_regret_iidG1.0}
\end{figure}
\begin{figure}[H]
    \centering
    \begin{subfigure}[b]{0.32\textwidth}
        \centering
        \includegraphics[width=\textwidth]{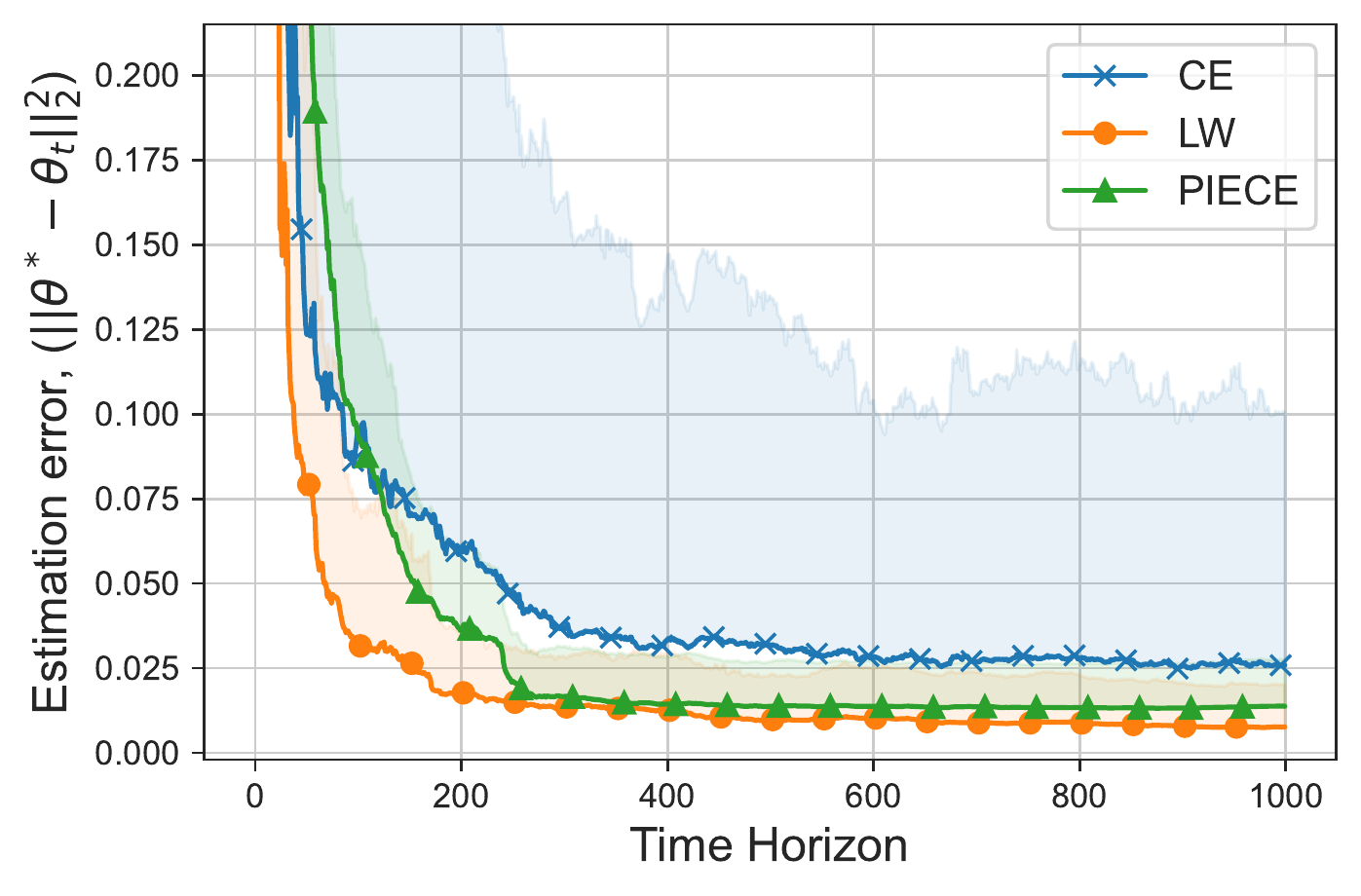}
        \caption{Example I}
        \label{fig:s1_err_g1.0}
    \end{subfigure}
    \hfill
    \begin{subfigure}[b]{0.32\textwidth}
        \centering
        \includegraphics[width=\textwidth]{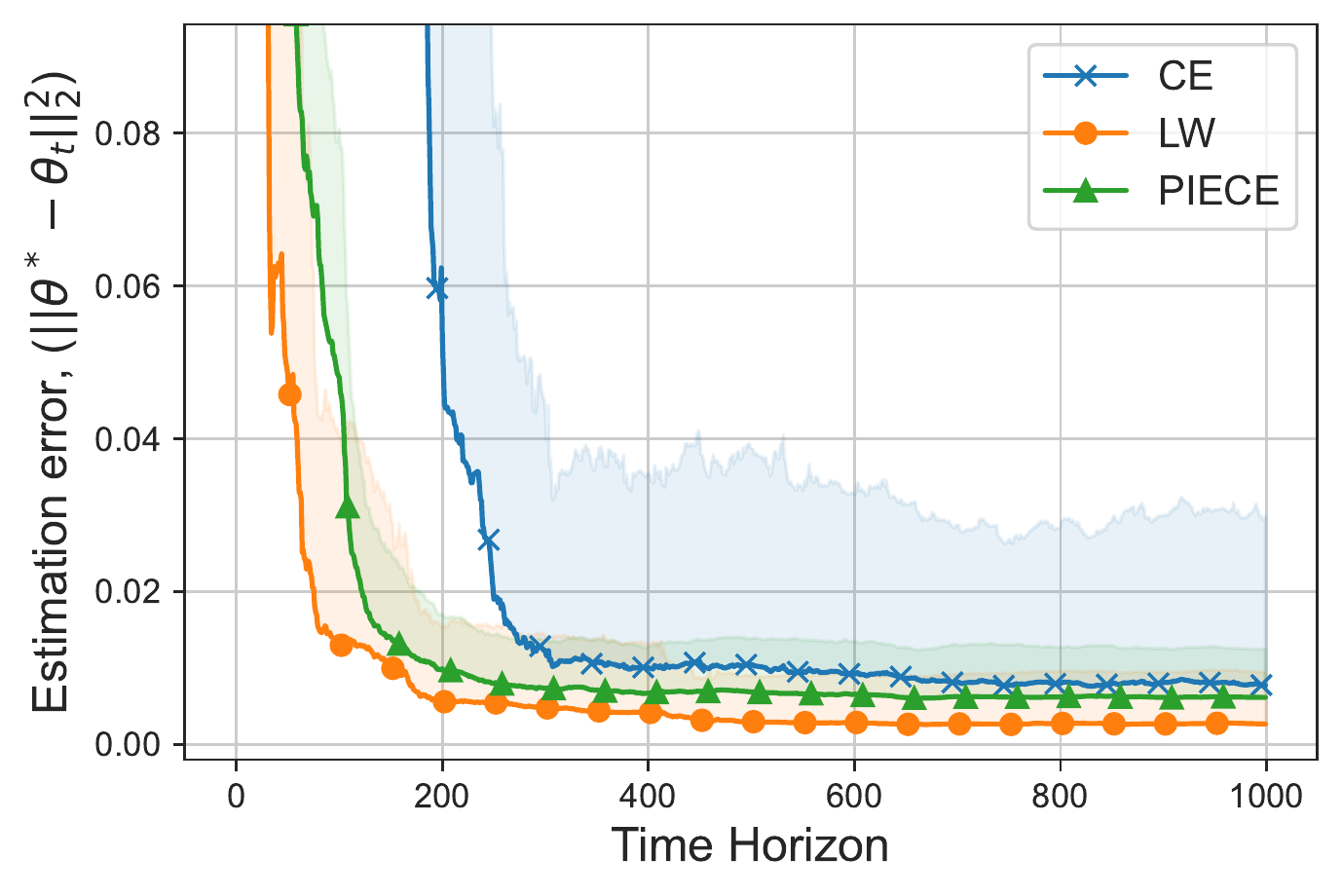}      
        \caption{Example II}         
        \label{fig:s2_err_g1.0}
    \end{subfigure}
    \hfill
    \begin{subfigure}[b]{0.32\textwidth}
        \includegraphics[width=\textwidth]{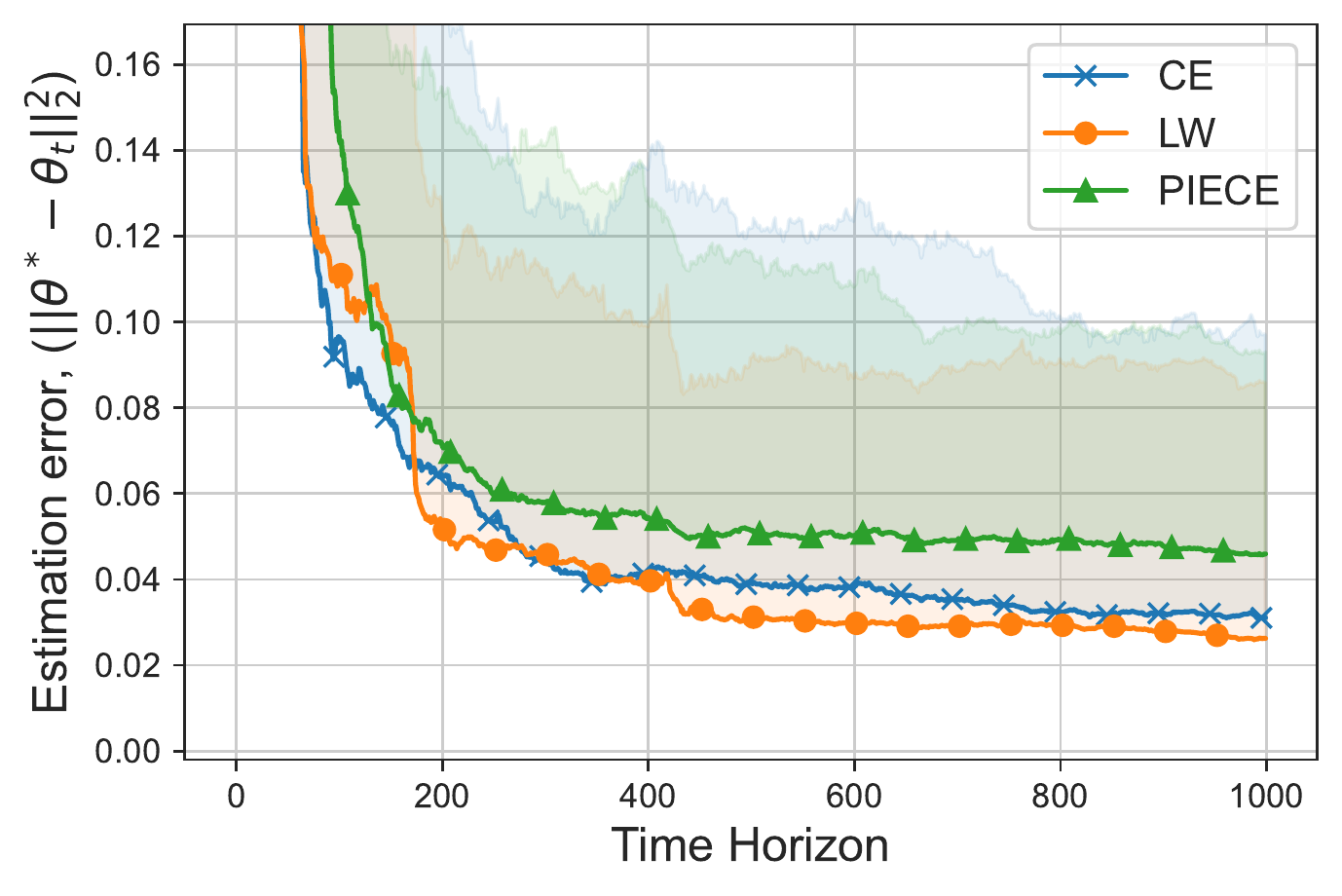}
        \caption{Example III}
        \label{fig:s3_err_g1.0}
    \end{subfigure}
    \caption{Estimation Error ($||\theta^\star-\theta_t||^2_2$) for Gaussian noise with mean $0$ and standard deviation $1.0$.}
    \label{fig:estimation_error_iidG1.0}
\end{figure}

\begin{table}[H]
    \centering
    \begin{tabular}{|c|c|c|c|}
        \hline
        & CE & LW & PIECE \\
        \hline
        Example I & 7135276 & 70694 & 1659\\
        \hline
        Example II & 3570341415 & 3833 & 400\\
        \hline
        Example III & 80269387 & 34165 & 316\\
        \hline
    \end{tabular}
    \caption{Average regret at $T=1000$ for Gaussian noise with mean $0$ and standard deviation $1.0$.}
    \label{tab:iidG1.0}
\end{table}

\begin{table}[H]
    \centering
    \begin{tabular}{|c|c|c|c|}
        \hline
        & $B_w$ & $B_u$ & $H$\\
        \hline
        Example I & 3.0 & 3931769.53 & 182\\
        \hline
        Example II & 3.0 & 27680.15 & 33\\
        \hline
        Example III & 3.0 & 299906.28 & 89\\
        \hline
    \end{tabular}
    \caption{PIECE hyper-parameters for Gaussian noise with mean $0$ and standard deviation $1.0$.}
    \label{tab:hpiidG1.0}
\end{table}
    \begin{figure}[H]
         \centering
              \begin{subfigure}[b]{0.32\textwidth}
             \centering
             \includegraphics[width=\textwidth]{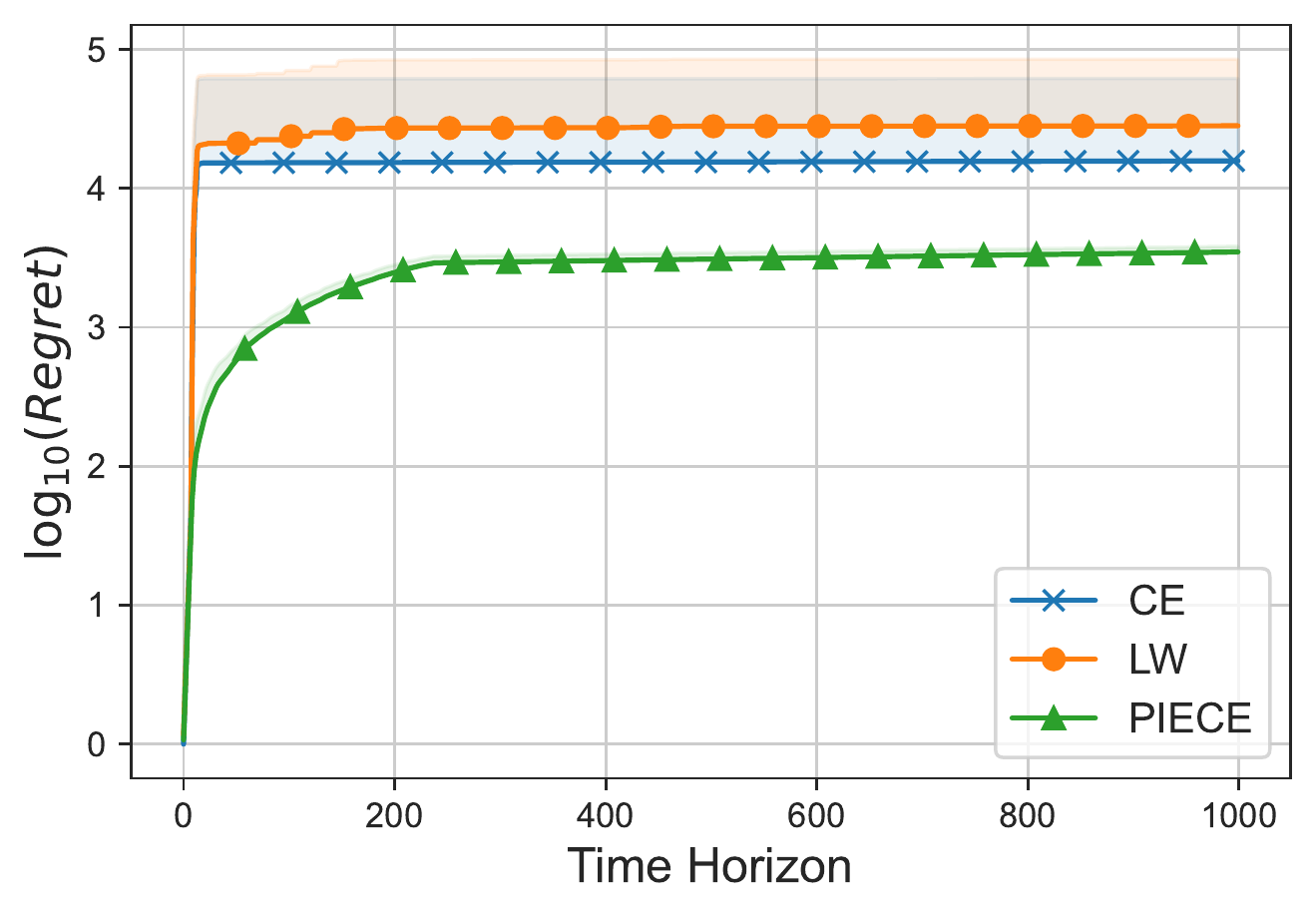}
             \caption{Example I}
             \label{fig:s1_lreg_rw0.5}
         \end{subfigure}
         \hfill
              \begin{subfigure}[b]{0.32\textwidth}
             \centering
             \includegraphics[width=\textwidth]{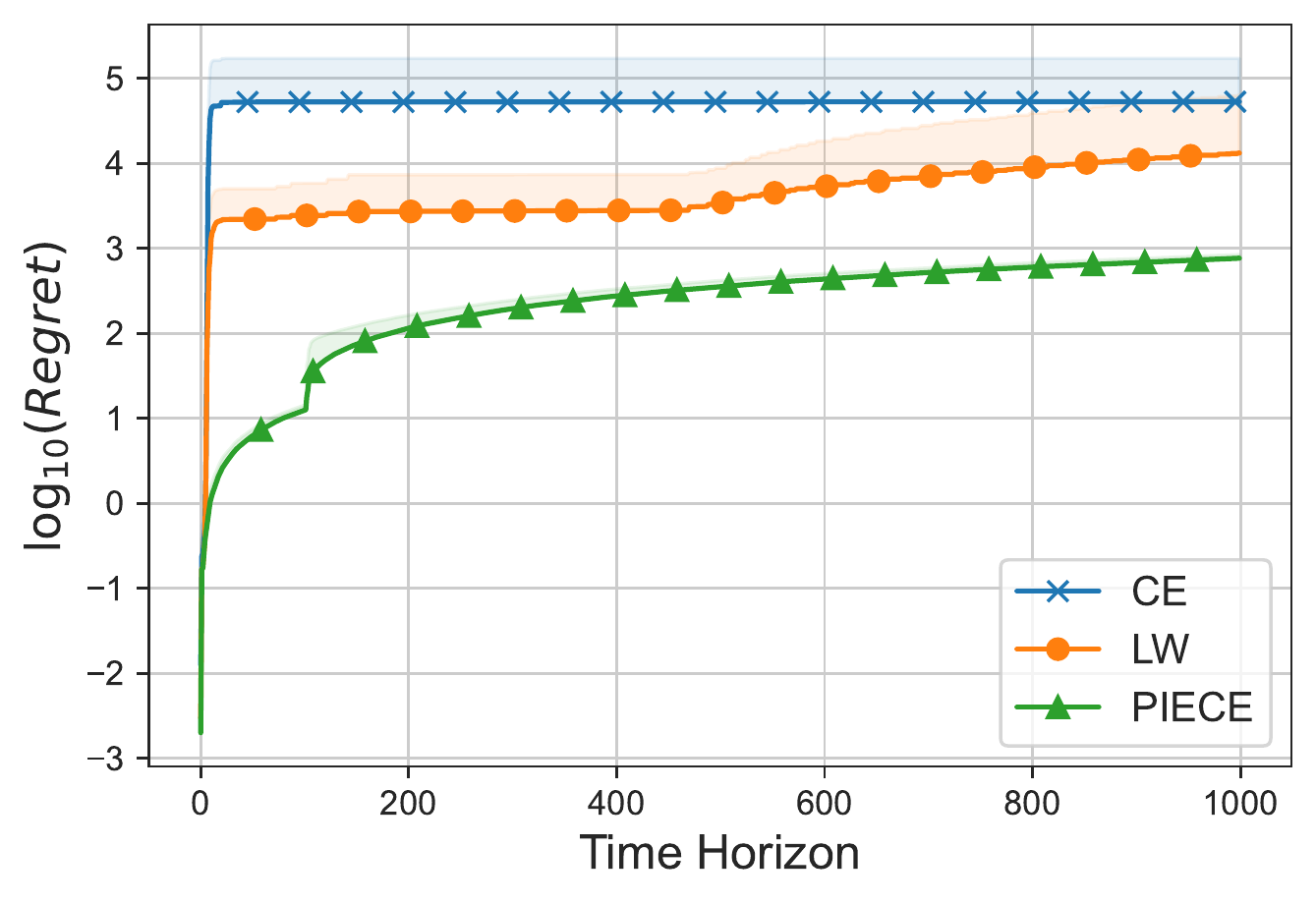} 
             \caption{Example II}         
             \label{fig:s2_lreg_rw0.5}
         \end{subfigure}
              \hfill
                   \begin{subfigure}[b]{0.32\textwidth}
             \centering
             \includegraphics[width=\textwidth]{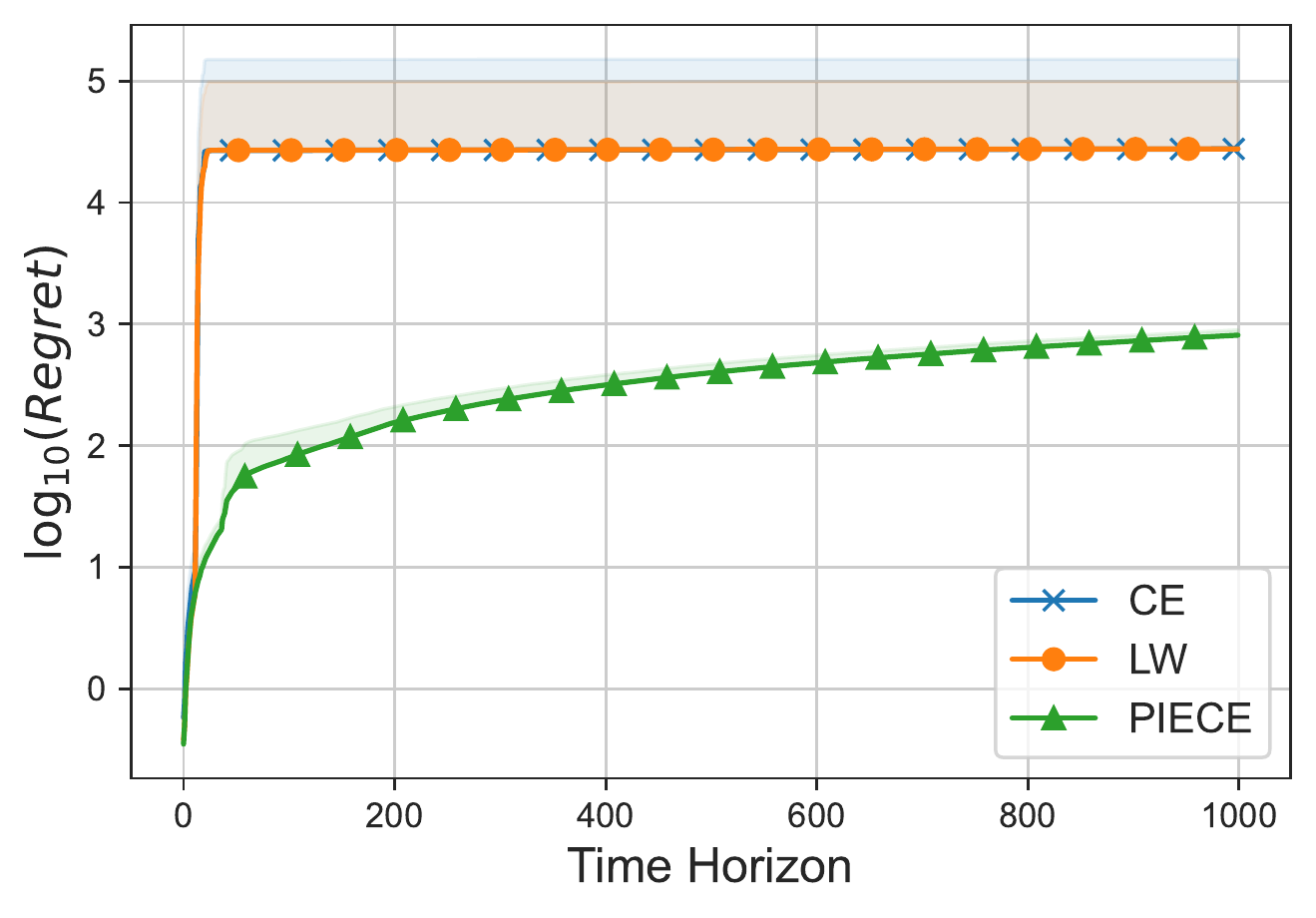}
             \caption{Example III}
             \label{fig:s3_lreg_rw0.5}
         \end{subfigure}
            \caption{Log(Cumulative Regret) averaged over 50 runs (Noise: Random walk with iid Gaussian steps, $\sigma=0.5 $).}
            \label{fig:cumulative_regret_rwalk0.5}
    \end{figure}
    \begin{figure}[H]
         \centering
              \begin{subfigure}[b]{0.32\textwidth}
             \centering
             \includegraphics[width=\textwidth]{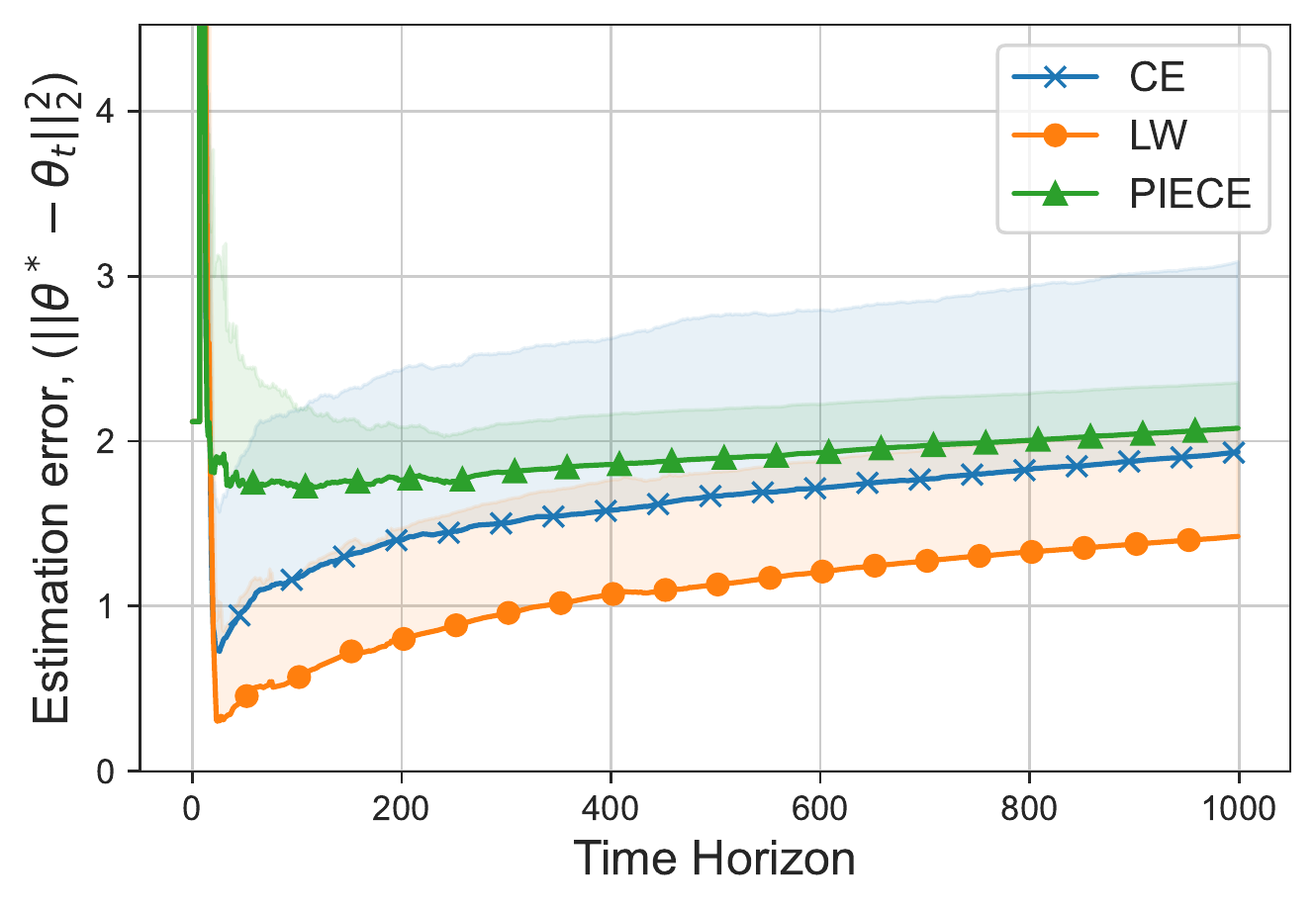}
             \caption{Example I}
             \label{fig:s1_err_rw0.5}
         \end{subfigure}
         \hfill
              \begin{subfigure}[b]{0.32\textwidth}
             \centering
             \includegraphics[width=\textwidth]{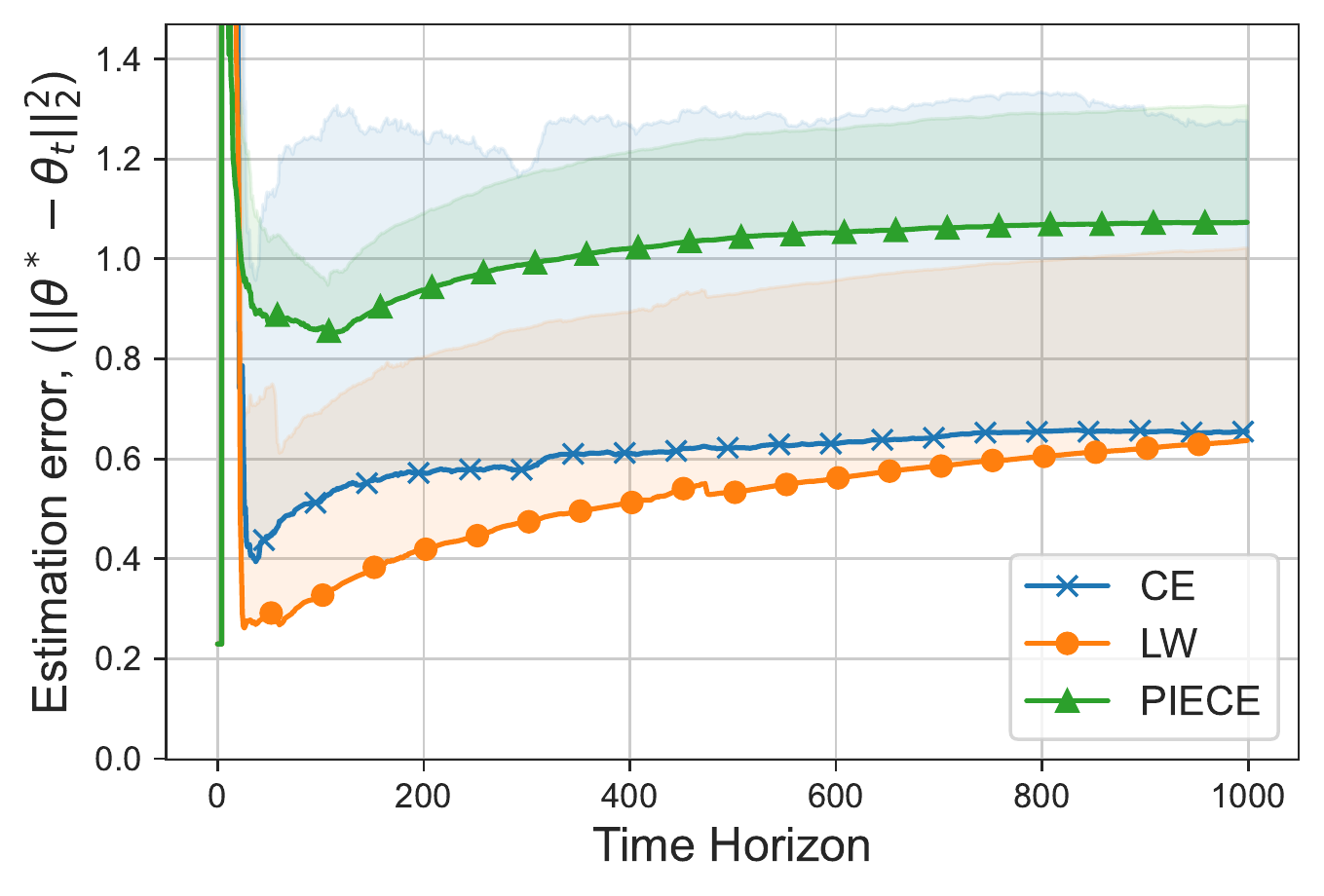}      
             \caption{Example II}         
             \label{fig:s2_err_rw0.5}
         \end{subfigure}
         \hfill
        \begin{subfigure}[b]{0.32\textwidth}
             \includegraphics[width=\textwidth]{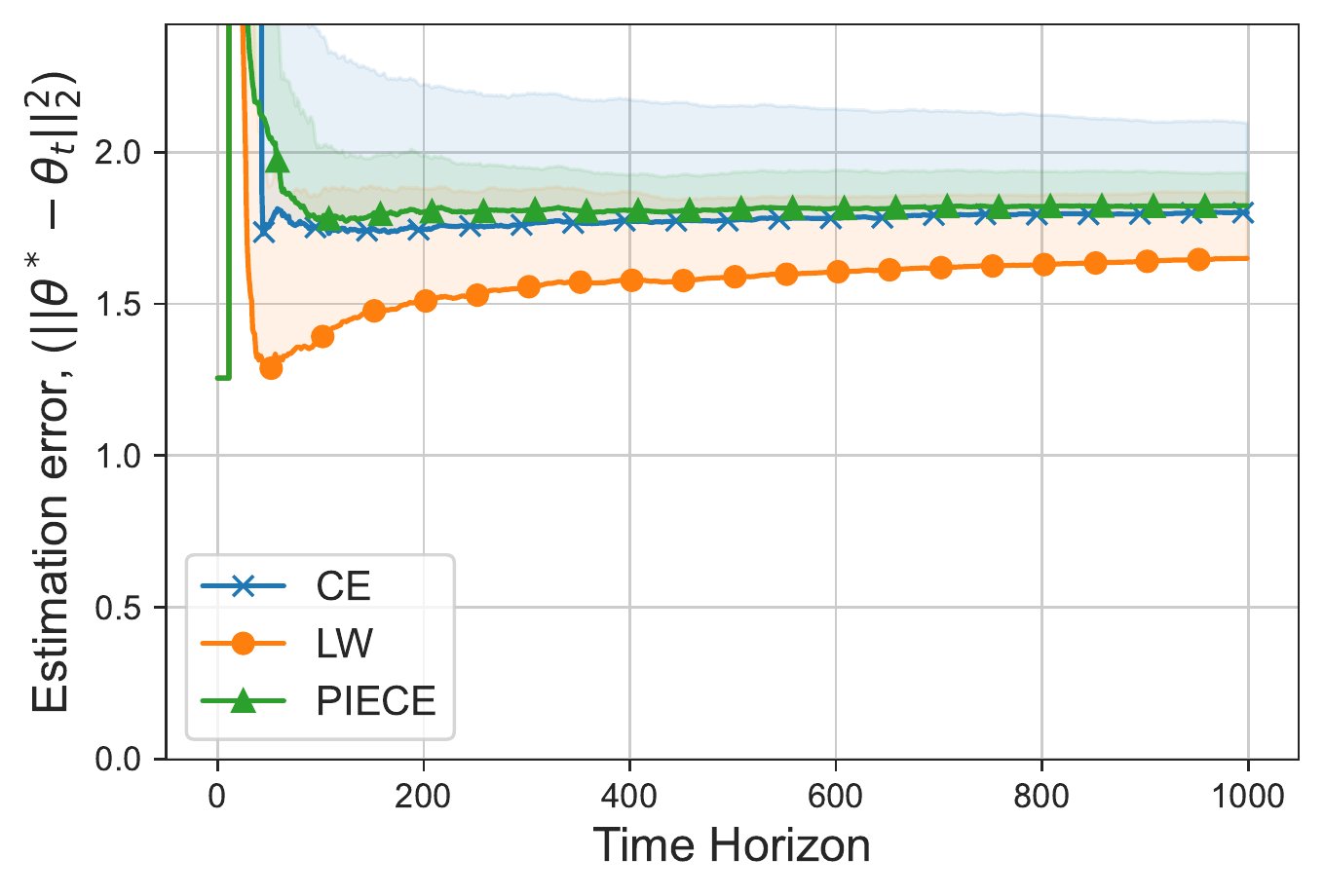}
             \caption{Example III}
             \label{fig:s3_err_rw0.5}
         \end{subfigure}
            \caption{Estimation Error ($||\theta^\star-\theta_t||^2_2$) (Noise: Random walk with iid Gaussian steps, $\sigma=0.5 $).}
            \label{fig:estimation_error_rwalk0.5}
    \end{figure}
    \begin{table}[H]
        \centering
        \begin{tabular}{|c|c|c|c|}
            \hline
             & CE & LW & PIECE \\
            \hline
            Example I & 15797 & 28249 & 3491\\
            \hline
            Example II & 52913 & 13177 & 763\\
            \hline
            Example III & 27607 & 27695 & 812\\
            \hline
        \end{tabular}
        \caption{Average regret at $T=1000$ (Noise: Random walk with iid Gaussian steps, $\sigma=0.5 $).}
        \label{tab:rwalk_iidG0.5}
    \end{table}
    \begin{table}[H]
        \centering
        \begin{tabular}{|c|c|c|c|}
            \hline
             & $B_w$ & $B_u$ & $H$\\
            \hline
            Example I & 1.5 & 1834199.5 & 181\\
            \hline
            Example II & 1.5 & 14204.29 & 33\\
            \hline
            Example III & 1.5 & 153899.34 & 90\\
            \hline
        \end{tabular}
        \caption{PIECE hyper-parameters (Noise: Random walk with iid Gaussian steps, $\sigma=0.5 $).}
        \label{tab:hprwG0.5}
    \end{table} 

\begin{figure}[H]
    \centering
    \begin{subfigure}[b]{0.32\textwidth}
        \centering
        \includegraphics[width=\textwidth]{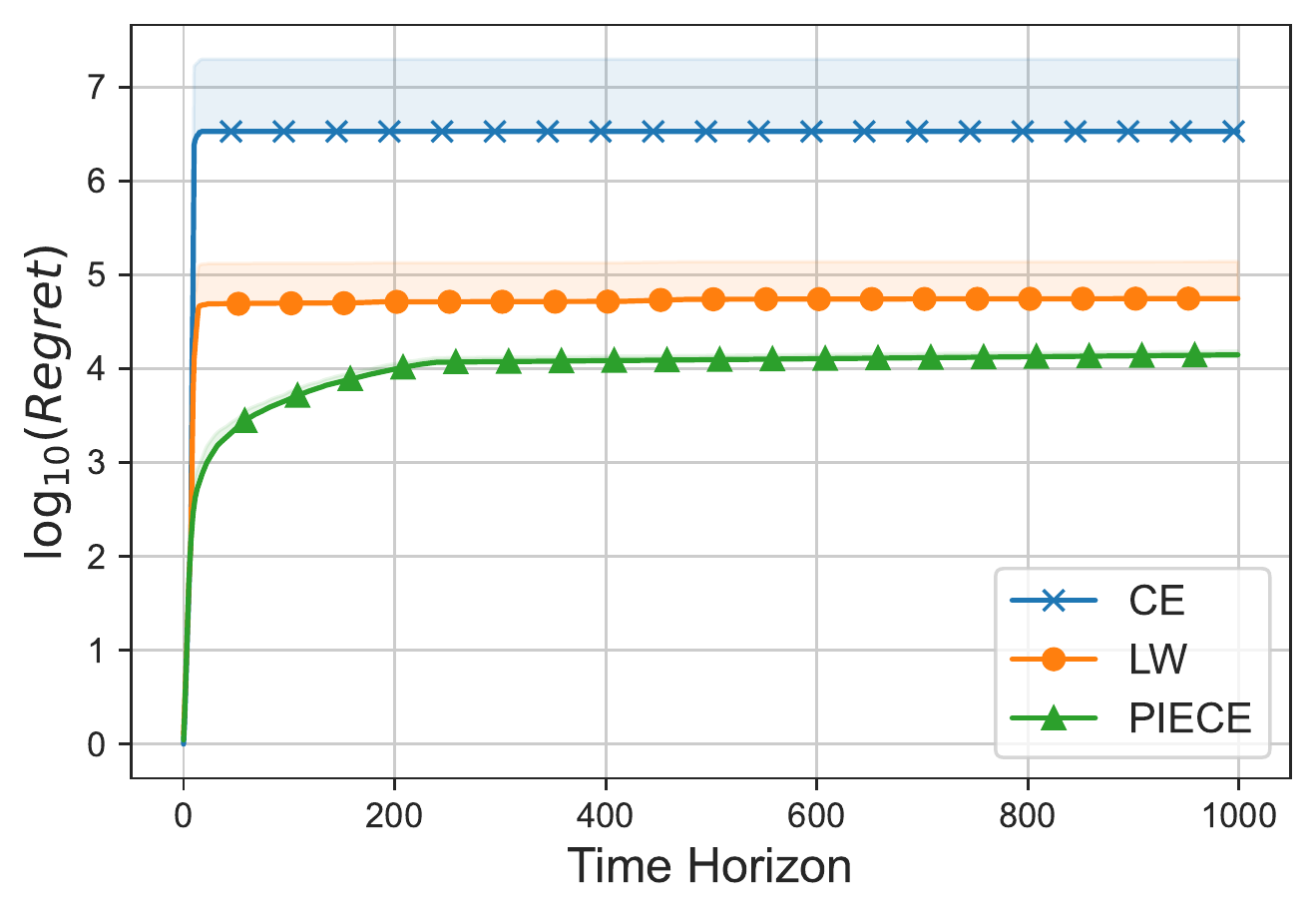}
        \caption{Example I}
        \label{fig:s1_lreg_rw1.0}
    \end{subfigure}
    \hfill
    \begin{subfigure}[b]{0.32\textwidth}
        \centering
        \includegraphics[width=\textwidth]{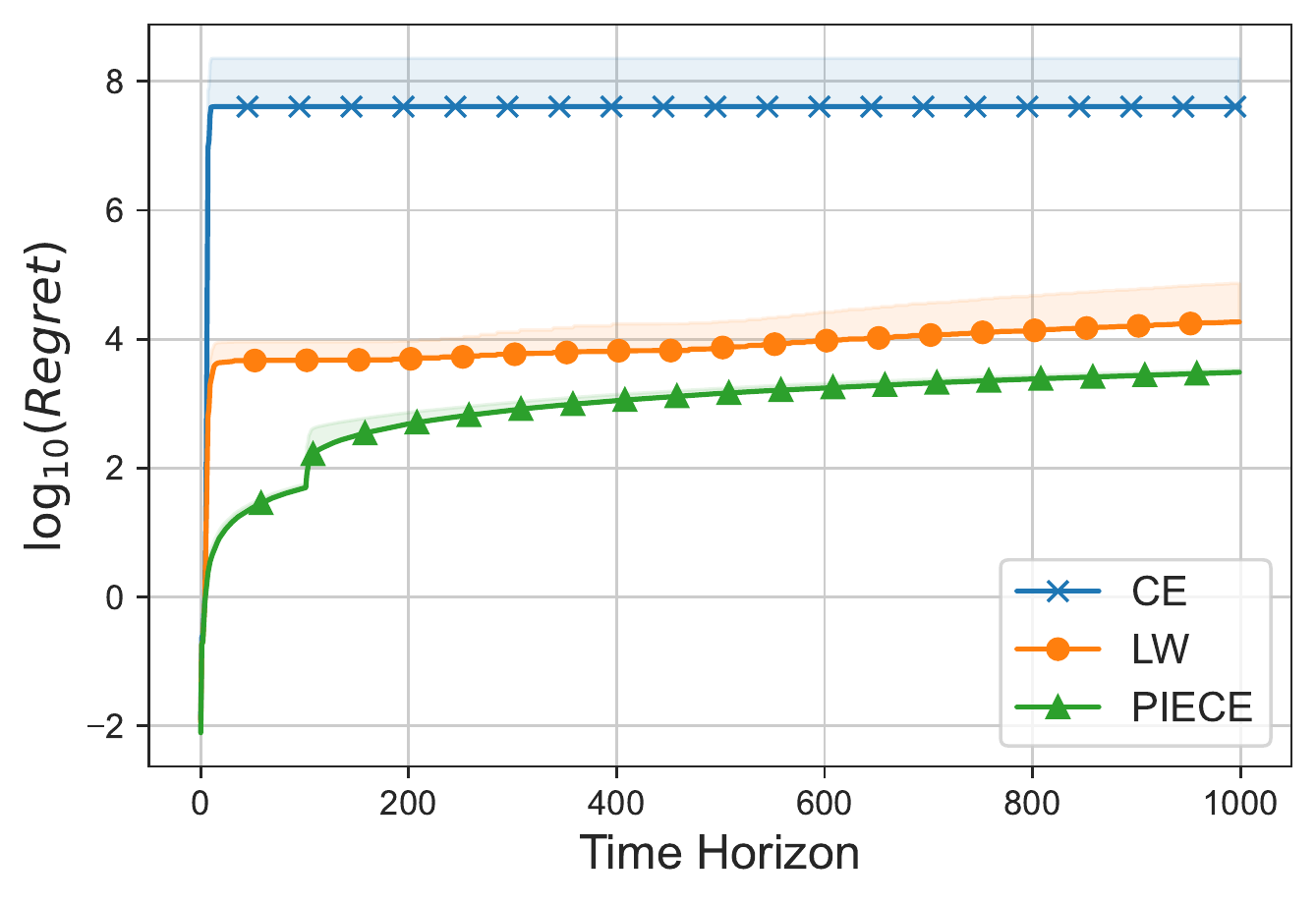} 
        \caption{Example II}         
        \label{fig:s2_lreg_rw1.0}
    \end{subfigure}
    \hfill
    \begin{subfigure}[b]{0.32\textwidth}
        \centering
        \includegraphics[width=\textwidth]{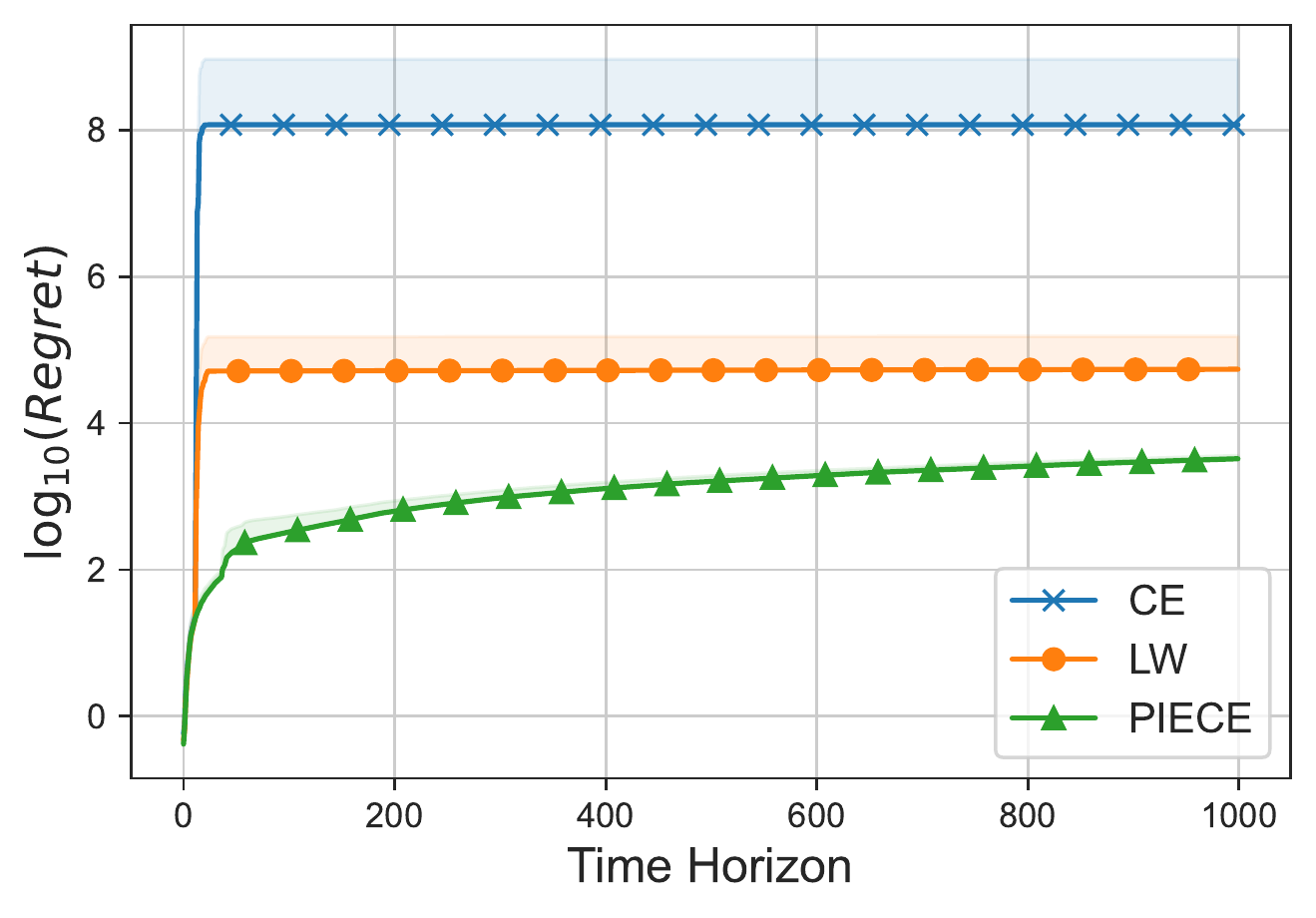}
        \caption{Example III}
        \label{fig:s3_lreg_rw1.0}
    \end{subfigure}
    \caption{Log(Cumulative Regret) averaged over 50 runs (Noise: Random walk sequence with iid Gaussian steps with mean $0$ and standard deviation $1.0$).}
    \label{fig:cumulative_regret_rwalk1.0}
\end{figure}

\begin{figure}[H]
    \centering
    \begin{subfigure}[b]{0.32\textwidth}
        \centering
        \includegraphics[width=\textwidth]{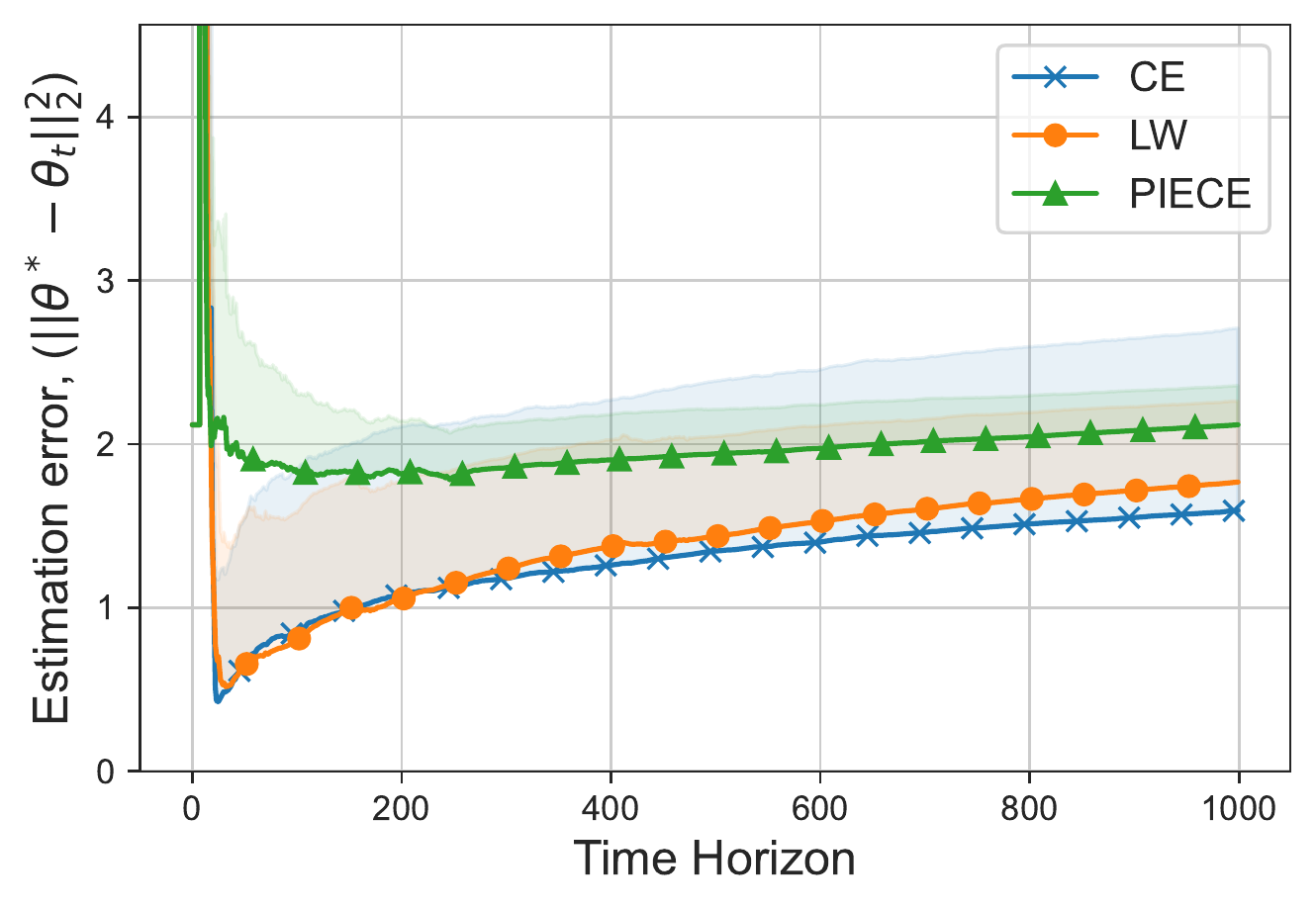}
        \caption{Example I}
        \label{fig:s1_err_rw1.0}
    \end{subfigure}
    \hfill
    \begin{subfigure}[b]{0.32\textwidth}
        \centering
        \includegraphics[width=\textwidth]{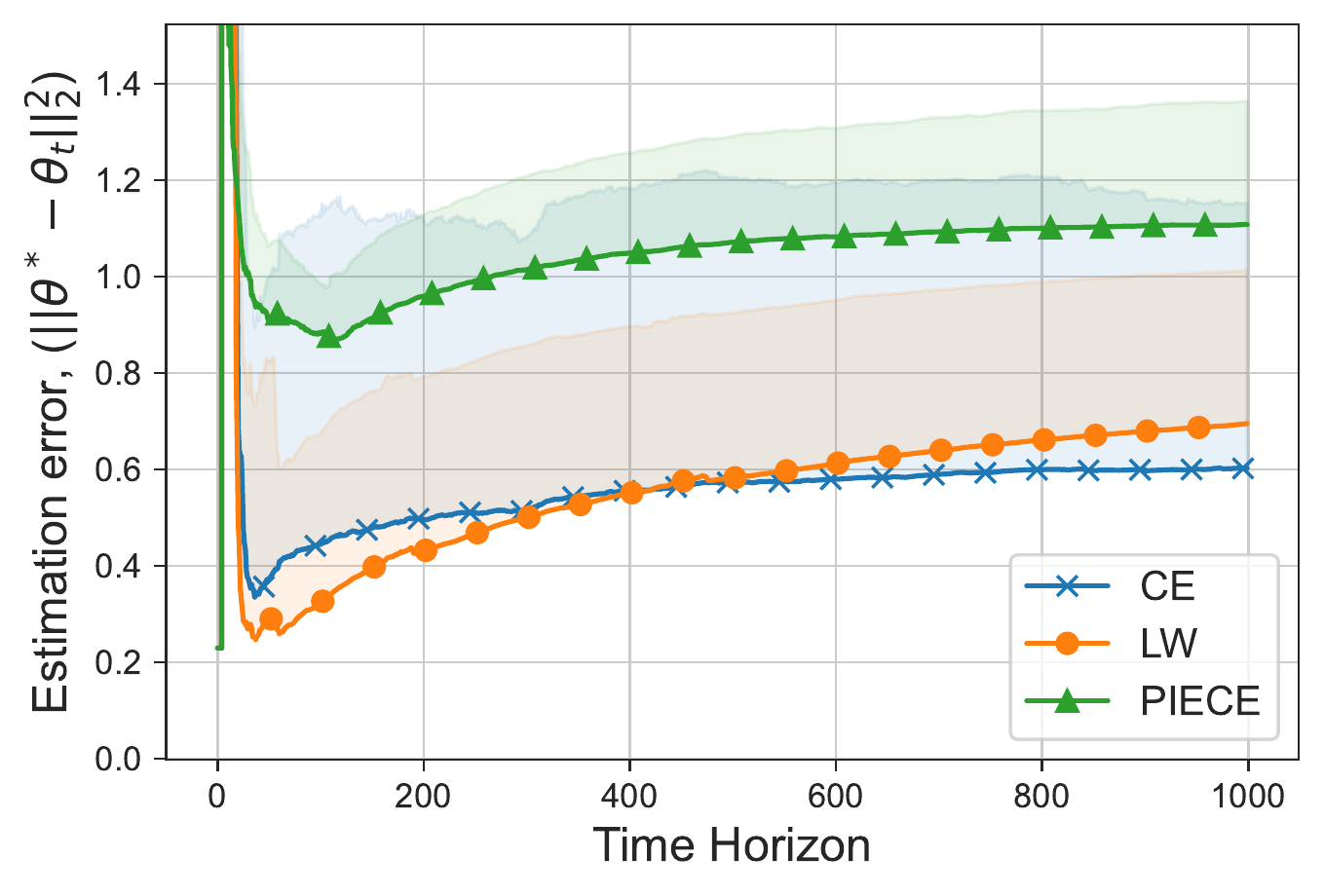}      
        \caption{Example II}         
        \label{fig:s2_err_rw1.0}
    \end{subfigure}
    \hfill
    \begin{subfigure}[b]{0.32\textwidth}
        \includegraphics[width=\textwidth]{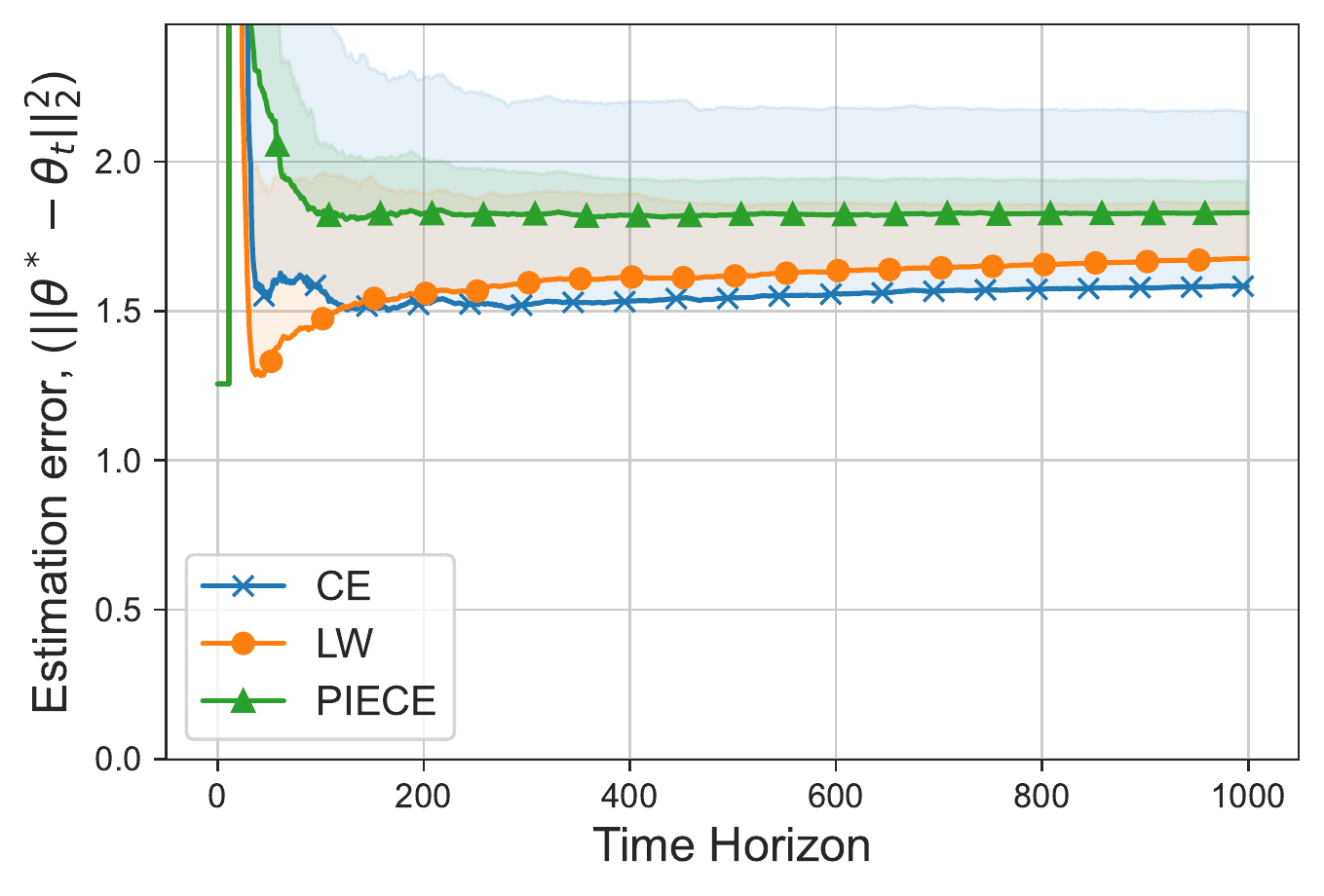}
        \caption{Example III}
        \label{fig:s3_err_rw1.0}
    \end{subfigure}
    \caption{Estimation Error ($||\theta^\star-\theta_t||^2_2$) (Noise: Random walk sequence with iid Gaussian steps with mean $0$ and standard deviation $1.0$).}
    \label{fig:estimation_error_rwalk1.0}
\end{figure}
         \begin{table}[H]
            \centering
            \begin{tabular}{|c|c|c|c|}
                \hline
                 & CE & LW & PIECE \\
                \hline
                Example I & 3361802 & 55812 & 14011\\
                \hline
                Example II & 40626507 & 18555 & 3071\\
                \hline
                Example III & 118066019 & 54318 & 3265\\
                \hline
            \end{tabular}
            \caption{Average regret at $T=1000$ (Noise: Random walk sequence with iid Gaussian steps with mean $0$ and standard deviation $1.0$).}
            \label{tab:rwalk_iidG1.0}
        \end{table}   
        \begin{table}[H]
            \centering
            \begin{tabular}{|c|c|c|c|}
                \hline
                & $B_w$ & $B_u$ & $H$\\
                \hline
                Example I & 3.0 & 3931769.53 & 182\\
                \hline
                Example II & 3.0 & 27680.15 & 33\\
                \hline
                Example III & 3.0 & 299906.28 & 89\\
                \hline
            \end{tabular}
            \caption{PIECE hyper-parameters (Noise: Random walk sequence with iid Gaussian steps with mean $0$ and standard deviation $1.0$).}
            \label{tab:hprwG1.0}
        \end{table}

\end{document}